\documentclass{article}

\usepackage[utf8]{inputenc}
\usepackage{amsmath}
\usepackage{amsfonts}
\usepackage{amssymb}
\usepackage{amsthm}
\usepackage{geometry}
\usepackage{graphicx}
\usepackage{thmtools, thm-restate}
\usepackage{listings}
\usepackage[framemethod=tikz]{mdframed}
\usepackage[sort]{natbib}
\usepackage{float}
\usepackage{makecell}
\usepackage{cancel}
\usepackage[shortlabels]{enumitem}
\usepackage{bm}
\usepackage{caption, subcaption}
\usepackage{multirow}
\usepackage{booktabs}  
\usepackage{array}     

\usepackage{xcolor}

\definecolor{DarkGreen}{rgb}{0,0.5,0}
\newcolumntype{L}[1]{>{\raggedright\arraybackslash}p{#1}}
\newcolumntype{R}[1]{>{\raggedleft\arraybackslash}p{#1}}

\usepackage[colorlinks,citecolor=blue,urlcolor=blue,linkcolor=blue,bookmarks=false]{hyperref}

\newtheorem{theorem}{Theorem}
\newtheorem{lemma}{Lemma}
\newtheorem{corollary}{Corollary}
\newtheorem{proposition}{Proposition}

\newenvironment{newproof}{\paragraph{\emph{Proof:}}}{\hfill$\square$ \\}

\theoremstyle{definition}
\newtheorem{definition}{Definition}
\newtheorem{property}{Property}

\newtheorem{csmodel}{Calibrated sensitivity model}
\newtheorem{assumption}{Assumption}

\theoremstyle{remark}
\newtheorem{remark}{Remark}

\def\inprob{\stackrel{p}{\rightarrow}}
\def\indist{\rightsquigarrow}
\newcommand\ind{\protect\mathpalette{\protect\independenT}{\perp}}
\def\independenT#1#2{\mathrel{\rlap{$#1#2$}\mkern4mu{#1#2}}}

\DeclareMathOperator*{\argmax}{arg\,max}

\DeclareSymbolFont{bbold}{U}{bbold}{m}{n}
\DeclareSymbolFontAlphabet{\mathbbold}{bbold}
\newcommand{\one}{\mathbbold{1}}
\def\cov{\text{cov}}
\def\bbE{\mathbb{E}}
\def\bbP{\mathbb{P}}
\def\bbR{\mathbb{R}}
\def\bbV{\mathbb{V}}

\title{\textbf{Calibrated sensitivity models}}
\author{Alec McClean\footnote{Division of Biostatistics, Department of Population Health, New York University Grossman School of Medicine}, Zach Branson$^\dagger$, and Edward H. Kennedy\footnote{Department of Statistics \& Data Science, Carnegie Mellon University} \\ \\ 
	\texttt{hadera01@nyu.edu, \{zach, edward\} @ stat.cmu.edu} 
}

\begin{document}
\maketitle
\begin{abstract}
	In causal inference, sensitivity models assess how unmeasured confounders could alter causal analyses, but the sensitivity parameter --- which quantifies the degree of unmeasured confounding --- is often difficult to interpret. For this reason, researchers sometimes compare the sensitivity parameter to an estimate of measured confounding. This is known as calibration, or benchmarking. However, calibrated estimates are not always interpreted correctly, and uncertainty in the estimate of measured confounding is rarely accounted for. To address these limitations, we propose calibrated sensitivity models, which directly bound the degree of unmeasured confounding by a multiple of measured confounding. We develop a clear framework for interpreting calibrated sensitivity models and derive statistical methods for accounting for uncertainty due to estimating measured confounding. Incorporating this uncertainty shows causal analyses may be either less or more robust to unmeasured confounding than suggested by standard approaches. We develop efficient estimators and inferential methods for bounds on the average treatment effect with three calibrated sensitivity models, establishing parametric efficiency and asymptotic normality under doubly robust style nonparametric conditions. We illustrate our methods with an analysis of the effect of mothers' smoking on infant birthweight. \\ 

\end{abstract}

\newpage
\section{Introduction} \label{sec: intro}

In causal inference, the goal is often to estimate whether exposure to a treatment causes a change in outcomes. Experiments where exposure to treatment is randomized facilitate valid estimation of causal effects under minimal assumptions, but are often infeasible, unethical, or too expensive.  Therefore, researchers frequently estimate causal effects from observational data, where exposure to treatment is not randomized. To estimate causal effects with observational data, researchers routinely invoke the \emph{no unmeasured confounding} assumption, which says that the treatment is as-if randomized within observed covariate strata.  Unfortunately, this assumption is often implausible, because there may be variables beyond the ones observed that are associated with the treatment and outcomes of the study.  Thus, it is imperative to conduct a sensitivity analysis to understand how robust causal analyses are to unmeasured confounding.

\bigskip

We focus on \emph{partial identification} sensitivity analyses. In simplified terms, these analyses impose a bound
\begin{equation} \label{eq:ug}
	U \leq \gamma,
\end{equation}
where $\gamma$ is a sensitivity parameter and $U$ is some quantification of unmeasured confounding, e.g., the difference in counterfactual regression functions \citep{diaz2013sensitivity} or the odds ratio of the probability of treatment \citep{rosenbaum2002sensitivity, tan2006distributional}.  The model in \eqref{eq:ug} implies bounds on the causal effect which can be estimated from observed data.  To understand the impact of unmeasured confounding, one can vary the sensitivity parameter, thereby allowing for different degrees of unmeasured confounding, and estimate bounds on and construct confidence intervals for the causal effect. Often, researchers determine the value of $\gamma$ where the confidence interval includes zero, because it indicates the level of unmeasured confounding where the causal effect estimate is non-significant. Researchers often appeal to a sense that if the resulting $\gamma$ value is large, causal analyses are robust to unmeasured confounding; conversely, if the $\gamma$ value is small, analyses are said to be sensitive to unmeasured confounding.

\bigskip

However, it can be difficult to gain intuition for the absolute size of the sensitivity parameter $\gamma$. Therefore, some have proposed \emph{calibrating} (or, \emph{benchmarking}) results by estimating measured confounding  (e.g., \citet{cinelli2020making, veitch2020sense, franks2020flexible}, among others). Typically, researchers leave out one or more variables from their data, allowing them to act as proxies for unmeasured confounders, and estimate a quantification of measured confounding analogous to the unmeasured confounding $U$ in their sensitivity model (e.g., if $U$ is the odds ratio of propensity scores, then so is measured confounding).  Customarily, one can then decide whether the causal effect estimate is robust to unmeasured confounding by comparing the sensitivity parameter to estimated measured confounding.  For example, one might compare the level of the sensitivity parameter where confidence intervals for the causal effect include zero to the estimated measured confounding. If the measured confounding were much smaller, this might be evidence that the causal effect estimate is robust to unmeasured confounding, because unmeasured confounders would need to have a larger impact on the causal effect than measured confounders to potentially reverse conclusions from the causal analysis.

\bigskip

This approach, which we refer to as ``post hoc calibration'', suffers from two drawbacks. First, researchers typically do not interpret their implied model. This can obscure the assumptions being made. By construction, a partial identification sensitivity model bounds the impact of an unmeasured confounder when added to the observed covariates. Post hoc calibration then implicitly links this unidentifiable quantity to the effect of adding an observed covariate --- or group of covariates --- to the other observed covariates.  As we show, this extrapolation implies some desiderata for the calibration procedure and highlights pitfalls that can arise in interpretation. 

\bigskip

Second, existing approaches rarely account for uncertainty in the estimate for measured confounding, though there are notable exceptions. \citet{cinelli2020making} provide uncertainty quantification that is limited to linear models. Meanwhile, other recent works have proposed the bootstrap or stacked estimating equation methods for uncertainty quantification \citep{hong2021did, sjolander2022sensitivity, chernozhukov2022long}, but only apply to specific semiparametric sensitivity models\footnote{We review semiparametric sensitivity models in Section~\ref{sec:setup}. They are an alternative to partial identification sensitivity models that we consider in this paper.} or rely on parametric estimators for measured confounding. Furthermore, prior work seems to have overlooked a crucial phenomenon: correctly accounting for this uncertainty can yield different conclusions for how robust the causal effect is to unmeasured confounding --- both more and less robust are possible, because the relative size of the sensitivity parameter depends on our uncertainty about measured confounding. 
 
\bigskip

To construct a model with an interpretable sensitivity parameter which avoids the deficiencies of post hoc calibration, we propose \emph{\textbf{calibrated sensitivity models}}. In simplified terms, calibrated sensitivity models use sensitivity models as a building block to impose a bound
\begin{equation} \label{eq:ugm}
	U \leq \Gamma M,
\end{equation}
where $U$ is unmeasured confounding, $M$ is \emph{measured confounding}, and $\Gamma$ is a sensitivity parameter. In other words, calibrated sensitivity models bound the degree of unmeasured confounding by a multiple of measured confounding. To our knowledge, explicitly defining bounds on unmeasured confounding within a calibrated sensitivity framework is new. Exploring the implications of defining bounds based on measured confounding is the primary contribution of this work.

\bigskip

Exploring the implications of defining bounds based on measured confounding is the primary contribution of this work.  To that end, our main contributions are as follows. First, we formally define several calibrated sensitivity models. These are based on popular sensitivity models defined at the level of the causal effect, the outcome regression, and the propensity score \citep{luedtke2015statistics,rosenbaum2002sensitivity}. Second, we demonstrate how to interpret these models. The interpretation depends on the observed covariates because the model bounds the additional impact of adding an unmeasured confounder to all observed covariates by (a multiple of) the additional impact of adding an observed covariate (or group of covariates) \emph{to all the other observed covariates}. We construct two simple data generating processes to illustrate possible pitfalls where these models could be misinterpreted. Additionally, this interpretation suggests simple heuristics by which to choose the covariate subsets for calibration and how to analyze results from a calibrated sensitivity analysis. 

\bigskip

Third, in a nonparametric framework, we establish theory and methods for estimating the bounds and the causal effect which account for uncertainty in estimating measured confounding. In Appendix~\ref{app:robustness}, we extend these methods to estimate and conduct inference on a one-number summary of study robustness --- the level of $\Gamma$ such that the bounds implied by the calibrated sensitivity model contain zero. Our estimators attain parametric efficiency and asymptotic normality under doubly-robust style nonparametric conditions on their nuisance function estimators. This allows us to use flexible estimators for nuisance functions like outcome regressions and the propensity score but still conduct valid inference within calibrated sensitivity analyses. We provide an illustrative data analysis, which demonstrates how to analyze data with a calibrated sensitivity model, and highlights the differences between our novel methods and standard approaches.  Our code is available at \url{https://github.com/alecmcclean/Calibrated-sensitivity-models}. Finally, we conclude the paper by discussing our proposed models in light of a popular distinction in the literature, between formal and informal benchmarking approaches.

\bigskip

While this paper focuses on nonparametric sensitivity analyses with doubly robust-style estimators, the core principles of calibrated sensitivity models generalize to other inferential paradigms. In particular, the idea of parameterizing sensitivity analyses directly in terms of measured confounding could be applied to Bayesian causal analysis \citep{zheng2024sensitivity}.

\subsection{Structure of the paper}

In Section~\ref{sec:setup} we provide notation and other setup, and review standard sensitivity analyses. In Section~\ref{sec:csm}, we introduce calibrated sensitivity models, discuss how to interpret them, and provide some guidelines for constructing them. We also introduce three examples to build intuition. In Section~\ref{sec:partial_id}, we identify bounds on the ATE under the three example models. In Section~\ref{sec:est}, we develop theory and methods for estimation and inference for bounds on the ATE. In Section~\ref{sec:illustrations}, we illustrate our methods with a real data analysis on the effect of mothers’ smoking on infant birth weight. In Section~\ref{sec:discussion}, we discuss our models in the context of formal and informal benchmarking.

\section{Setup and background} \label{sec:setup}

\subsection{Notation} \label{sec:notation}

We use $\bbE$ for expectation, $\bbV$ for variance, and $\bbP$ for probability. We use $\bbP_n \{ f(Z) \} = \frac1n \sum_{i=1}^{n} f (Z_i)$ as shorthand for the sample average of $f(X)$. When $x \in \bbR^d$ we let $\lVert x \rVert^2 = \sum_{j=1}^{d} x_j^2$, and for generic possibly random functions $f$ we let $\lVert f(Z) \rVert_p^p = \int_{\mathcal{Z}} f(z)^p d\bbP(z)$ denote the $L_p^p (\bbP)$ norm for $1 \leq p < \infty$ and let $\lVert f(Z) \rVert_{\infty} = \sup_{x \in \mathcal{X}} |f(x)|$. We use $\indist$ to denote convergence in distribution and $\inprob$ for convergence in probability. Finally, for a positive integer $d$ we use the notation $[d] = \{1, \dots, d \}$.

\subsection{Data, assumptions, and the average treatment effect} \label{sec:data}

We assume we observe $n$ observations $\{ Z_i \}_{i=1}^{n} \stackrel{iid}{\sim} \mathcal{P}$ where $Z$ is a tuple $(X, A, Y)$, $X \in \bbR^d$ are $d$-dimensional covariates, $A \in \{0, 1\}$ is a binary exposure, and $Y \in \bbR$ is an outcome. We use subscript notation to develop calibrated sensitivity models: $X_{-j}$ excludes the $j^{th}$ covariate, $X_{j}$ is the $j^{th}$ covariate, and, for any set $S \subseteq [d]$, $X_{-S}$ excludes the covariates corresponding to $S$. We let $Z_{-S} = \{X_{-S}, A, Y \}$ and assume $\{ Z_{-S} \}_{i=1}^{n} \stackrel{iid}{\sim} \mathcal{P}_{-S}$, i.e., the data with only covariates $X_{-S}$ is drawn from some distribution $\mathcal{P}_{-S}$.  Additionally, we denote unmeasured confounders by $W$ and define potential outcomes $Y^a$ as the outcome that would have been observed under exposure $A = a$.  Finally, we will refer to $\pi_a(X) = \bbP(A = a \mid X)$ as the ``propensity score'' and $\mu_a (X) = \bbE(Y \mid A = a, X)$ as the ``outcome regression function'', and generically we will refer to them as ``nuisance functions''.  Other nuisance functions will be defined when they appear. We make two standard causal assumptions.
\begin{assumption} \label{asmp:consistency}
	\emph{Consistency: } $A = a \implies Y = Y^a$.	
\end{assumption}
\begin{assumption} \label{asmp:positivity}
	\emph{Positivity: } $\exists\ \varepsilon > 0$ such that $\mathbb{P}\left\{ \varepsilon \leq \pi_1(X) \leq 1 - \varepsilon \right\} = 1$.	
\end{assumption}
Consistency says that we observe the potential outcome relevant to the observed exposure. It would be violated if, for example, there were interference between subjects such that one subject's treatment affected another's outcome. Positivity says that all subjects have a non-zero probability of exposure.  The literature addressing violations of each of these assumptions is too large to summarize here, but see, for example, \citet{tchetgen2012causal} and \citet{westreich2010invited} for discussion of violations of consistency and positivity, respectively.

\bigskip

While some of our methods generalize to other causal estimands, for simplicity we focus on the ATE, $\psi_\ast = \bbE(Y^1 - Y^0).$ If, in addition to consistency and positivity, no unmeasured confounding held --- i.e., if $(Y^0, Y^1) \ind A \mid X$ --- then the ATE could be identified by fully observed quantities. Specifically, the ATE can be identified by the adjusted mean difference: $\psi = \bbE \{ \mu_1 (X) - \mu_0 (X) \}.$ We define the adjusted mean difference with a general covariate set $X_{-S}$ as
\begin{equation} \label{eq:mde}
	\psi_{-S} = \bbE \{ \mu_1(X_{-S}) - \mu_0 (X_{-S}) \}.
\end{equation}
Without no unmeasured confounding or an alternative assumption (e.g., that an instrumental variable or regression discontinuity exists), one cannot identify $\psi_\ast$ such that $\psi_\ast = \psi$. To address this, there is a large literature on sensitivity analyses, which impose an assumption on the effect of unmeasured confounding.

\subsection{Sensitivity analyses} \label{sec:literature}

Modern sensitivity models could arguably be split into two types: models which admit point identification of the causal effect, and models which admit partial (or, set) identification of the causal effect. Generally, point identification models impose a specific form on the relationship between the unmeasured confounder and the observed data.  For example, \citet{nabi2024semiparametric} consider a model of the form $p(Y^a = y \mid A = 1-a, X = x) = p(Y^a = y \mid A = a, X = x) q_a(Y^a = y, X = x; \gamma)$, where $p$ is the density of $Y^a$ and $q_a(Y^a = y, X = x; \gamma)$ is a researcher-specified transformation which depends on the sensitivity parameter $\gamma$. There are many popular models in the literature (e.g., \citet{robins1999association, brumback2004sensitivity, nabi2024semiparametric}, among others). Under such an assumption, the ATE (for example) can be identified and estimated, and typically researchers examine how the ATE changes with different values of the sensitivity parameter $\gamma$.  

\bigskip

In this paper, we focus on partial identification sensitivity analysis, which generally comprise the following steps.  First, the researcher imposes a bound $U \leq \gamma$ for $\gamma \in (0, \infty)$ which implies bounds on the causal effect.\footnote{Often, researchers consider multi-dimensional constraints, which impose multiple bounds on multiple notions of measured confounding, with $U_1 \leq \gamma_1$ and $U_2 \leq \gamma_2$, etc. We focus on a one-dimensional model for simplicity.} We denote the lower and upper bounds on the ATE as $\ell: \mathcal{P} \times (0, \infty) \to \bbR$ and $u: \mathcal{P} \times (0, \infty) \to \bbR$, respectively. For various levels of $\gamma$, the researcher estimates the bounds and constructs confidence intervals for the bounds and, by extension, the causal effect. Finally, to interpret levels of $\gamma$, the researcher estimates a quantification of measured confounding $M$, and compares levels of $\gamma$ to $\widehat{M}$. Generally, $M$ is a measured quantity analogous to $U$, where one or more variables are left out at a time, capturing measured confounding implied by those variables.  As far as we are aware, these steps constitute a comprehensive sensitivity analysis in the current state of the literature, though the final calibration step, comparing $\gamma$ to $\widehat M$, is not always performed.  Researchers might also examine other statistics, such as the level of $\gamma$ where confidence intervals for the causal effect include zero, i.e., where a significant effect estimate is nullified.

\bigskip

Within partial identification sensitivity analyses, the key choices are how to parameterize the effect of unmeasured confounding, via $U$, and then how to choose its observed data equivalent, $M$. Here, we review choices of unmeasured confounding, summarizing a non-exhaustive snapshot of the literature.  Early partial identification models imposed bounds which follow directly from bounds on the data itself \citep{manski1990bounds, robins1989analysis}, e.g., if the outcome is bounded this implies a bound on the ATE.  More recent work has bounded the error due to unmeasured confounding at the level of the causal effect itself \citep{luedtke2015statistics}, the odds ratio of the propensity score \citep{rosenbaum2002sensitivity, tan2006distributional}, the direct change in the propensity score \citep{masten2018identification}, the change in or ratio of outcome regression functions \citep{diaz2013sensitivity, luedtke2015statistics}, the change in explained variance of the outcome regressions or the propensity score \citep{cinelli2020making,chernozhukov2022long, huang2025variance}, and the proportion of units confounded \citep{bonvini2022sensitivitya}.

\section{Calibrated sensitivity analyses: models and interpretation} \label{sec:csm}

This section introduces calibrated sensitivity analysis. First, we outline the generic steps of partial identification calibrated sensitivity analyses, highlighting where they differ from partial identification sensitivity analyses. After introducing a simple example, we discuss how to interpret calibrated sensitivity models. We then outline some practical guidelines for choosing covariates for calibration. We conclude with several additional examples of calibrated sensitivity models. These examples highlight how calibrated sensitivity models provide interpretable sensitivity parameters, clarify assumptions by explicitly incorporating measured confounding, and naturally accommodate statistical uncertainty---advantages we build on in the sections that follow.

\subsection{Calibrated sensitivity analyses} 

Calibrated sensitivity analyses comprise the following steps. First, the researcher imposes a bound $U \leq \Gamma M$ for $\Gamma \in (0, \infty)$ which implies bounds on the causal effect.  We denote lower and upper bounds on the ATE as $\mathcal{L}: \mathcal{P} \times (0, \infty) \to \bbR$ and $\mathcal{U}: \mathcal{P} \times (0, \infty) \to \bbR$. The bounds differ from a standard sensitivity analysis because $\mathcal{L}$ and $\mathcal{U}$ are different functionals from $\ell$ and $u$ for a standard sensitivity analysis --- $\mathcal{L}$ and $\mathcal{U}$ depend on measured confounding. For various levels of $\Gamma$, the researcher estimates the bounds and constructs confidence intervals for the bounds and, by extension, the causal effect. Unlike with standard sensitivity analyses, the method for constructing confidence intervals accounts for uncertainty in estimating measured confounding. These steps constitute a comprehensive calibrated sensitivity analysis. Unlike with standard sensitivity analyses, a final post hoc calibration step is not necessary, because $\Gamma$ is already calibrated by construction.  

\subsection{Example: maximum leave-one-out effect differences model}

To ground ideas, we begin with a simple subtractive calibrated sensitivity model, which we will revisit in the next two sections to clarify its interpretation and practical use. \color{black} It is the \emph{maximum leave-one-out effect differences} model (which we also refer to as the ``effect differences model'').  
\begin{csmodel} \emph{(Maximum leave-one-out effect differences)} \label{csm:max_loo_fx}
	\begin{equation}
		\left|\psi_\ast - \psi \right|\leq \Gamma \max_{j \in [d]} \left| \psi - \psi_{-j} \right|, \label{eq:max_loo_fx}
	\end{equation}
	where $\Gamma \in (0, \infty)$ is the sensitivity parameter, $\psi$ is the adjusted mean difference in \eqref{eq:mde} with $X$ covariates (i.e., all covariates) and $\psi_{-j}$ is the adjusted mean difference with $X_{-j}$ covariates (i.e., without covariate $X_j$). 
\end{csmodel}
The maximum LOO effect differences model bounds the causal bias by the maximum change in the adjusted mean difference from leaving out one covariate at a time \citep{luedtke2015statistics}.  While quantifying confounding at the level of the causal effect is less popular than other approaches, we include it because the resulting analysis is simple and therefore eases exposition of the main ideas in this paper. A drawback with this model is that it may be too coarse, and would hide information about the changes in the data generating process required to change the causal effect.

\subsection{Interpreting calibrated sensitivity models} \label{sec:interpretation}

The effect differences model illustrates the central, untestable assumption behind calibrated sensitivity analyses: \emph{that the effect of adding an unmeasured confounder to the observed covariates can be bounded by a multiple of the effect of adding a measured covariate to the remaining covariates.} Crucially, this interpretation is tied to the specific set of covariates observed, because the model bounds the additional effect of an unmeasured confounder when added to the measured covariates already included in the analysis. This dependence on the observed covariate set is not merely a technical detail but shapes the substantive meaning of the calibrated sensitivity analysis. Moreover, it dictates our approach to defining measured confounding, which captures the additional effect of a measured confounder (or set of confounders) when added to the other observed covariates, to best mimic the effect of an unmeasured confounder when added to all the observed covariates.

\medskip

This covariate-dependent interpretation of calibrated sensitivity models, while important, is not new. The principle that sensitivity analyses are interpreted in the context of the data observed has been emphasized previously in the literature. \citet{rosenbaum2002covariance, robins2002covariance, rosenbaum2002rejoinder} involve an illuminating discussion on this point. In particular, \citet[Section 7]{rosenbaum2002rejoinder} emphasizes that the meaning and strength of any sensitivity analysis is inextricably linked to the particular set of covariates controlled for in the analysis. Indeed, \citeauthor{rosenbaum2002rejoinder} pointed out this applies to other popular statistical quantities: ``The magnitude and interpretation of a regression coefficient depends upon which other variables are in the model," and ``cannot be understood without reference to the other parts of the model.'' The same argument applies to sensitivity parameters. This contextual dependence is not a limitation but rather an essential feature that ensures the sensitivity analysis addresses the specific inferential problem at hand, given the available data. In our case, the same argument applies to a calibrated sensitivity analysis.

\medskip

We further discuss this interpretation in the context of ``formal benchmarking'' in Section~\ref{sec:discussion} \citep{cinelli2020making}. There, we also provide two examples to illustrate the observed-covariate-dependent interpretation we take throughout the paper. Next, we use this interpretation to suggest some guidelines for constructing the models and interpretating the calibrated sensitivity parameter $\Gamma$.

\subsection{Guidelines for choosing covariates and interpreting the calibrated sensitivity parameter}

Here, we suggest some practical guidelines for choosing the measured confounding $M$ in the bound. First, the measured confounding $M$ should be defined analogously to the unmeasured confounding $U$. This provides a structured basis for reasoning about plausible values of $\Gamma$. For example, if $U = | \psi_\ast - \psi|$, then measured confounding should depend on $| \psi -\psi_{-S}|$ for different covariate sets $S$. When the analogous measure is difficult to estimate, researchers may consider a smooth approximation. For instance, when unmeasured confounding involves an $L_\infty$ norm, researchers might define measured confounding using a smooth approximation, since $L_\infty$ norms cannot be estimated at $\sqrt{n}$ rates under nonparametric assumptions \citep{lepski1999estimation}. This consideration applies to the odds ratio example we introduce in Section~\ref{sec:examples}. In Section~\ref{sec:est}, we return to this issue and show how projecting $M$ onto a finite-dimensional model can enable valid inference.

\smallskip

Second, the choice of which covariates to omit when constructing $M$ should follow a pre-specified protocol that guards against cherry-picking. Ideally, this decision would be informed by substantive knowledge about which covariates, when omitted, might mimic the effect of an unmeasured confounder. However, such prior knowledge is often unavailable. In these cases, we recommend defining $M$ using a maximum over a broad, pre-specified set of covariate subsets to capture a wide range of potential confounding patterns. The covariate subsets used to construct $M$, the rationale for their selection, and any sensitivity of $M$ to these choices should be reported. This procedural approach helps ensure that the resulting $M$ reflects a systematic data-driven exploration rather than selective reporting.

\smallskip

Third, it can be useful to inspect the estimated $M$ and its components, as these values inform the plausible range of $\Gamma$. If a single covariate produces a very large change in the estimand, it may be implausible that an unmeasured confounder could have a similar effect---particularly if the omitted covariate is already known to be a strong confounder. In that case, smaller values of $\Gamma$ may be more appropriate. Conversely, if all covariates appear weak, then larger $\Gamma$ values may be plausible. Either way, reporting the magnitude of $M$ and discussing which covariates drive it provides helpful context for interpreting the sensitivity analysis.

\subsection{Two more example calibrated sensitivity models} \label{sec:examples}

We now present two additional examples of calibrated sensitivity models. These models use different definitions of confounding---based on the propensity score and the outcome regression---to illustrate how the same calibration logic applies across popular sensitivity analyses. The second model is the \emph{maximum leave-one-out $L_\infty$ propensity score odds ratio} model (which we also refer to as the ``odds ratio model''). 
\begin{csmodel} \label{csm:max_loo_odds}
	\emph{(Maximum leave-one-out $L_\infty$ propensity score odds ratio)}
	\begin{equation}
		\sup_{x, w, \widetilde w} \left| \log \left[ \frac{\text{odds} \{ \pi_1(x, w) \}}{\text{odds} \{ \pi_1(x, \widetilde w) \}} \right] \right| \leq \Gamma \max_{j \in [d]} \sup_{x_{-j}, x_j, \widetilde x_j} \left| \log \left[ \frac{\text{odds} \{ \pi_1(x_{-j}, x_j) \}}{\text{odds} \{ \pi_1(x_{-j}, \widetilde x_j) \}} \right] \right|
	\end{equation}
	where $\Gamma \in (0, \infty)$ is the sensitivity parameter, $\pi_1$ is the propensity score, and $\text{odds}(p) = p/(1-p)$. 
\end{csmodel}
The maximum LOO odds ratio model uses the maximum leave-one-out error again and a popular quantification of unmeasured confounding --- the maximum propensity score odds ratio induced by the unmeasured confounder \citep{rosenbaum2002sensitivity, yadlowsky2022bounds}.  We state the model in terms of the log-odds ratio so that the calibrated sensitivity parameter $\Gamma$ takes the same range and has the same meaning as in the other two models. In the odds ratio model, the unmeasured and measured confounding are defined in terms of a supremum, which is a popular definition in the literature. 

\smallskip

The third model is the \emph{average leave-some-out $L_2$ outcome regression differences} model (which we also refer to as the ``outcome model''). 
\begin{csmodel} \label{csm:avg_lso_out}
	\emph{(Average leave-some-out $L_2$ outcome regression differences)} For $a \in \{0,1\}$,
	\begin{align}
		&\left\lVert \bbE(Y^a \mid A=a, X) - \bbE(Y^a \mid A = 1-a, X) \right\rVert_2^2 \nonumber \\
		&\hspace{1in} \leq  \frac{\Gamma}{|\mathcal{S}|} \sum_{S \in \mathcal{S}} \left\lVert \mu_a(X_{-S}) - \bbE \big\{ \mu_a(X) \mid A = 1-a, X_{-S} \big\} \right\rVert_2^2, \label{eq:def_out}
	\end{align}
	where $\Gamma \in (0, \infty)$ is the sensitivity parameter and $\mathcal{S}$ is a set of subsets of $[d]$.
\end{csmodel}
The average LSO outcome model is also based on a popular quantification of unmeasured confounding --- the difference between the potential outcome regressions under the observed treatment and the opposite treatment \citep{diaz2013sensitivity, luedtke2015statistics}.  We use this model to demonstrate alternative choices one might make when constructing a calibrated sensitivity model. First, unmeasured and measured confounding are in terms of $L_2^2$ norms rather than the usual $L_\infty$ norm. Second, measured confounding is the average LSO error rather than the maximum LOO error in the two previous models.

\smallskip

We conclude this section with a general assumption. It asserts that the overall quantification of measured confounding is bounded and non-zero. 
\begin{assumption} \label{asmp:bounded} \emph{Bounded and non-zero measured confounding:}
	\begin{enumerate}
		\item Under the effect differences model, $|\psi - \psi_{-j}| < \infty$ for all $j \in [d]$ and there exists $j \in [d]$ such that $|\psi - \psi_{-j}| > 0$.
		\item Under the odds ratio model, $\sup_{x_{-j}, x_j, \widetilde x_j} \left| \log \left[ \frac{\text{odds} \{ \pi_1(x_{-j}, x_j) \}}{\text{odds} \{ \pi_1(x_{-j}, \widetilde x_j) \}} \right] \right| < \infty$ for all $j \in [d]$ and there exists $j \in [d]$ such that $\sup_{x_{-j}, x_j, \widetilde x_j} \left| \log \left[ \frac{\text{odds} \{ \pi_1(x_{-j}, x_j) \}}{\text{odds} \{ \pi_1(x_{-j}, \widetilde x_j) \}} \right] \right| > 0$.
		\item Under the outcome model, for $a\in \{0,1\}$, $\lVert \mu_a(X_{-S}) - \bbE \{ \mu_a(X) \mid A = 1-a, X_{-S} \} \rVert_2^2 < \infty$ for all $S \in \mathcal{S}$ and there exists $S \in \mathcal{S}$ such that $\lVert \mu_a(X_{-S}) - \bbE \{ \mu_a(X) \mid A = 1-a, X_{-S} \} \rVert_2^2 > 0$.
	\end{enumerate}
\end{assumption}
This assumption is necessary so that the calibrated sensitivity analyses are meaningful. If measured confounding is infinite, then the bound on unmeasured confounding is also infinite, and the bounds on the ATE will be infinitely wide. Meanwhile, if measured confounding is zero, then the bound on unmeasured confounding is again zero, and the model assumes that no unmeasured confounding holds.

\section{Partial identification} \label{sec:partial_id}

In this section, we provide partial identification results for bounds on the ATE from the three examples in Section~\ref{sec:csm}, and then establish that the bounds are differentiable with respect to measured confounding, which we use to provide estimation convergence guarantees subsequently.

\begin{restatable}{proposition}{proppartialid} \label{prop:partial_id}
	\emph{\textbf{(Partial identification)}} Suppose Assumptions~\ref{asmp:consistency} and~\ref{asmp:positivity} hold. Under the maximum LOO effect differences model (calibrated sensitivity model~\ref{csm:max_loo_fx}),
	\begin{align}
		\mathcal{U}(\Gamma) &= \psi + \Gamma \max_{j \in [d]} \left| \psi - \psi_{-j} \right| \text{ and } \nonumber  \\
		\mathcal{L}(\Gamma) &= \psi - \Gamma \max_{j \in [d]} \left| \psi - \psi_{-j} \right| \label{eq:fx_id}
	\end{align}
	Under the maximum LOO odds ratio model (calibrated sensitivity model~\ref{csm:max_loo_odds}), 
	\begin{align}
		\mathcal{L}(\Gamma) &= \bbE \Big[ A Y + (1-A) \theta_1^- \{ X; \exp(\Gamma M) \} - \Big\{ (1-A)Y + A\theta_0^+ \{ X; \exp(\Gamma M) \} \Big\} \Big] \text{ and } \nonumber \\
		\mathcal{U}(\Gamma) &= \bbE \Big[ AY + (1-A) \theta_1^+ \{ X; \exp(\Gamma M) \} - \Big\{ (1-A) Y + A\theta_0^- \{ X; \exp(\Gamma M) \} \Big\} \Big] \label{eq:odds_id_upper}
	\end{align}
	where, e.g., 
	\begin{equation} \label{eq:theta}
		\theta_1^{-} (X; t) = \inf \Big\{ \bbE \{ Y L(Y) \mid A = 1, X \} : L \in \mathbb{L} \Big\}, 
	\end{equation}
	and 
	$\mathbb{L} = \left\{ L : \bbR \to \bbR \text{ measurable } : \begin{array}{l}
		0 \leq L(y) \leq t L(\tilde y) \text{ for all } y, \tilde y, \\
		\bbE \{ L(Y) \mid A = 1, X \} = 1
	\end{array} \right\}$, and 
	\begin{equation} \label{eq:M_odds}
		M = \max_{j \in [d]} \sup_{x_{-j}, x_j} \left| \log \left[ \frac{\text{odds} \{ \pi_1(x_{-j}, x_j) \}}{\text{odds} \{ \pi_1(x_{-j}, \widetilde x_j) \}} \right] \right|.
	\end{equation}
	Finally, under the average LSO outcome model (calibrated sensitivity model~\ref{csm:avg_lso_out}),
	\begin{align}
		\mathcal{U}(\Gamma) &= \psi + \Gamma \sum_{a \in \{0,1\}} \lVert \pi_{1-a}(X) \rVert_2 \sqrt{\frac{1}{|\mathcal{S}|} \sum_{S\in \mathcal{S}} \lVert \mu_a(X_{-S}) - \bbE \big\{ \mu_a(X) \mid A = 1-a, X_{-S} \} \rVert_2^2 } \text{ and } \nonumber \\ 
		\mathcal{L}(\Gamma) &= \psi - \Gamma \sum_{a \in \{0,1\}} \lVert \pi_{1-a}(X) \rVert_2 \sqrt{\frac{1}{|\mathcal{S}|} \sum_{S\in \mathcal{S}} \lVert \mu_a(X_{-S}) - \bbE \big\{ \mu_a(X) \mid A = 1-a, X_{-S} \} \rVert_2^2 } \label{eq:out_id}
	\end{align}
\end{restatable}
Proposition~\ref{prop:partial_id} shows that all three calibrated sensitivity models induce bounds on the ATE which follow naturally from the bounds induced by the relevant sensitivity model, but including measured confounding in the bounds. Expression \eqref{eq:fx_id} follows directly from the definition of the model; \eqref{eq:odds_id_upper} follows from the definition of the model and Lemma 2.1 in \citet{yadlowsky2022bounds}; and \eqref{eq:out_id} follows by H\"{o}lder's inequality.  A formal proof, and all subsequent proofs, are delayed to the supplementary materials. For partial identification in the odds ratio model, only $\theta_1^-$ is defined, in \eqref{eq:theta}. Taking a supremum instead of an infimum gives $\theta_1^+$, and swapping the conditioning to $A = 0$ from $A=1$ gives $\theta_0^-$ and $\theta_0^+$.

\subsection{Differentiable bounds}

When estimating the bounds, we will rely on them being differentiable with respect to measured confounding to allow use of Taylor's theorem and the delta method when providing convergence guarantees. Under Assumption~\ref{asmp:bounded}, this is trivially true for the effect differences and outcome regression models because measured confounding appears linearly in the bounds, but proving this is more involved for the odds ratio model.  The derivative for the upper bound is established in the next result, while the derivative for the lower bound follows by a similar analysis.  In addition to Assumption~\ref{asmp:bounded}, the result relies on the positivity assumption (Assumption~\ref{asmp:positivity}) and a regularity assumption (Assumption~\ref{asmp:continuity}), detailed in the supplementary materials, which asserts that the covariate density, propensity score, and $\theta_a^\pm$ from \eqref{eq:theta} are continuous in $x$, while the outcome $Y$ is bounded and has continuous and upper bounded conditional density.
\begin{restatable}{lemma}{lemshowdiffmonotone} \label{lem:show_diff_monotone} 
	\emph{\textbf{(Differentiable and monotone bounds in $M$ for the odds ratio model)}} \\
	Let the upper bound $\mathcal{U}(\Gamma)$ be defined as in \eqref{eq:odds_id_upper} and $M$ denote measured confounding as in \eqref{eq:M_odds}. Suppose Assumptions~\ref{asmp:positivity}, \ref{asmp:bounded}, and \ref{asmp:continuity} hold. Then,
	\begin{equation} \label{eq:odds_deriv}
		\frac{\partial}{\partial M} \mathcal{U}(\Gamma) =  \Gamma \bbE \left\{ \pi_0(X) \frac{\widetilde{f}_1 ( X; \theta_1^+ )}{\nu_1^+ (X)} \right\} + \Gamma \exp(\Gamma M) \bbE \left\{ \pi_1(X) \frac{f_0 ( X; \theta_0^- )}{\nu_0^-(X)} \right\},
	\end{equation}
	where $f_a, \widetilde f_a$, and $\nu_a^\pm$ are defined in Lemma~\ref{lem:diff} in the supplementary materials. Moreover,
	\begin{equation*}
		\exists\ C > 0 \text{ such that } -C < \frac{\partial}{\partial M} \mathcal{L}(\Gamma) < 0 < \frac{\partial}{\partial M} \mathcal{U}(\Gamma) < C.
	\end{equation*} 
\end{restatable}
Lemma~\ref{lem:show_diff_monotone} shows that the bounds implied by the odds ratio model widen as measured confounding increases. Unlike the effect differences and outcome models, the bounds in the odds ratio model widen in a way that's non-linear with $M$.

\section{Estimation and inference} \label{sec:est}

In this section, we establish convergence guarantees for estimators of measured confounding and bounds on the ATE with two of the three example models.  In the supplementary materials, we establish estimation with the outcome model and demonstrate when sensitivity analyses with post hoc calibration might over- or under-estimate robustness to unmeasured confounding by not accounting for uncertainty in estimating measured confounding.   

\bigskip

Before stating the results, we make several general observations. First, many of the convergence guarantees in this paper are doubly robust, showing $\sqrt{n}$-estimation can be possible even when nuisance functions are estimated at slower rates. This occurs because the estimators are based on the efficient influence function (EIF), which allows their bias to be a second-order product of errors, so they can achieve $\sqrt{n}$-efficiency even when the nuisance functions are estimated at slower-than-$\sqrt{n}$ rates \citep{chernozhukov2018double, kennedy2022semiparametric}.  These slower rates could be achieved under nonparametric structural assumptions on the nuisance functions, such as smoothness, sparsity, or bounded variation \citep{gyorfi2002distribution}.

\bigskip

Second, the results require stronger convergence guarantees on the nuisance function estimators than what is required for $\sqrt{n}$-consistent estimation of the bounds on the ATE in a standard sensitivity analysis. This occurs because it is necessary to construct $\sqrt{n}$-consistent estimators for measured confounding, which require accurate nuisance function estimators with multiple sets of covariates. By contrast, a standard sensitivity analysis only requires accurate nuisance function estimators with all covariates. However, it is worth noting that the requirements here are no stronger than what would be required to construct $\sqrt{n}$-consistent estimators of measured confounding in a typical post hoc calibration analysis.

\bigskip

Third, throughout, we assume the nuisance functions estimators are constructed on a sample of $n$ observations which is separate and independent from the sample used to estimate measured confounding and the bounds on the ATE.  Sample splitting and cross-fitting allows us to avoid imposing Donsker or stability conditions on the nuisance function estimators (e.g., \citet{van1996weak, chen2022debiased}). Although this seemingly cuts the sample size in half, one could retain full sample efficiency by swapping folds, repeating the estimator, and averaging the two estimates. For better stability, one could use more than two splits; five and ten are common. To simplify notation, we focus on the single split estimator. 



\subsection{Effect differences}

First, we consider estimating the bound on the ATE in the effect differences model (calibrated sensitivity model~\ref{csm:max_loo_fx}). 
\begin{definition} \label{def:effect}
	Let $\mathcal{U}(\Gamma)$ be as in \eqref{eq:fx_id}, and construct an estimator for the upper bound as
	\begin{equation} \label{eq:fx_estimator}
		\widehat{\mathcal{U}}(\Gamma) := \bbP_n \{ \widehat \phi(Z) \} + \Gamma \max_{j \in [d]} \left| \bbP_n \{ \widehat \phi(Z) - \widehat \phi(Z_{-j})\} \right|,
	\end{equation}
	where $\phi(Z_{-j})$ is the EIF of the adjusted mean difference $\psi_{-j}$, i.e., 
	\begin{equation} \label{eq:mde_if}
		\phi(Z_{-j}) = \mu_1(X_{-j}) - \mu_0(X_{-j}) + \left\{ \frac{A}{\pi_1(X_{-j})} - \frac{1-A}{\pi_0(X_{-j})} \right\} \{ Y - \mu_A(X_{-j}) \},
	\end{equation}
	and the estimated nuisance functions constituting $\widehat \phi$ are constructed on a separate sample ($\{ \widehat \pi_a, \widehat \mu_a \}$ for $a \in \{0,1\}$ and $X_{-j}$ for $j \in [d] \cup \emptyset$).
\end{definition}

The upper bound is identified as $\mathcal{U}(\Gamma) = \psi + \Gamma \max_{j \in [d]} |\psi - \psi_{-j}|$. Therefore, the estimator $\widehat{\mathcal{U}}(\Gamma)$ plugs in estimators for each adjusted mean difference and takes the maximum difference $\left|\widehat \psi - \widehat \psi_{-j} \right|$ across $j \in [d]$. The estimator uses the EIF of the adjusted mean difference, which is well-studied in the doubly robust estimation literature \citep{robins1994estimation}. As a result, the convergence guarantees will be doubly robust, as communicated in Theorem~\ref{thm:fx_conv} below. To facilitate straightforward $\sqrt{n}$-inference when measured confounding is a maximum, we invoke a separation condition.
\begin{assumption} \label{asmp:separation_max_fx}
	\emph{Separation of maximum in the effect differences model:} There exists $j^\prime \in [d]$ such that $\left| \psi - \psi_{-j^\prime}\right| > \left| \psi - \psi_{-j}\right| \ \forall \ j \in [d] \setminus j^\prime$. 
\end{assumption}
While this assumption is relatively mild, it could be relaxed by using a smooth approximation of the maximum, such as the LogSumExp function. Future work may also develop statistical methods that accommodate multiple maxima, drawing on theoretical results from, e.g., \citet{luedtke2016statistical}. With Assumption~\ref{asmp:separation_max_fx} ensuring separation, the next result establishes convergence guarantees for estimating the upper bound on the ATE within the effect differences model.

\begin{restatable}{theorem}{thmfxconv} \label{thm:fx_conv} \textbf{\emph{(Maximum LOO effect differences model)}} 
	Let $\widehat{\mathcal{U}}(\Gamma)$ and $\mathcal{U}(\Gamma)$ be as in Definition~\ref{def:effect}, and let $j^\prime = \argmax_{j \in [d]} \left|  \psi - \psi_{-j} \right|.$ Suppose Assumptions~\ref{asmp:positivity}-\ref{asmp:separation_max_fx} hold and 
	\begin{enumerate}
		\item for $j \in [d] \cup \emptyset$, $\widehat \pi_1(X_{-j})$ is bounded away from zero and one  \label{cond:bounded}
		\item for $j \in [d] \cup \emptyset$, $\widehat \phi$ is consistent in the sense that $\lVert \widehat \phi(Z_{-j}) - \phi(Z_{-j}) \rVert_2 \inprob 0$, and \label{cond:empirical_process}
		\item for $j \in \{ j^\prime, \emptyset\}$, $\lVert \widehat \pi_1(X_{-j}) - \pi_1(X_{-j}) \rVert_2 \left( \sum_{a \in \{0,1\}} \lVert \widehat \mu_a(X_{-j}) - \mu_a(X_{-j}) \rVert_2 \right) = o_{\bbP}(n^{-1/2})$. \label{cond:bias}
	\end{enumerate}
	Then,
	\begin{equation} \label{eq:fx_upper_lim}
		\widehat{\mathcal{U}}(\Gamma) - \mathcal{U}(\Gamma) = (\bbP_n - \bbE) \Big[ \phi(Z) +  \text{sign} (\psi- \psi_{-j^\prime}) \Gamma \big\{ \phi(Z) - \phi(Z_{-j^\prime}) \big\} \Big] + o_\bbP(n^{-1/2}).
	\end{equation} 
\end{restatable}

Theorem~\ref{thm:fx_conv} establishes that the error of the estimator for the upper bound behaves like a centered sample average plus asymptotically negligible error under doubly robust conditions on the nuisance function estimators.  Indeed, because the estimator uses the EIF of the adjusted mean difference, its error follows the form of the doubly robust estimator for the adjusted mean difference, requiring that the propensity score and outcome regression are estimated at a $\sqrt{n}$-rate in product.

\bigskip

Here, we give some intuition for the conditions of the result.  Condition~\ref{cond:bounded} is necessary so that the bias of each adjusted mean difference estimator is bounded, while condition~\ref{cond:empirical_process} is a weak consistency condition necessary for controlling the empirical process terms in the error of each adjusted mean difference estimator.  Crucially, condition~\ref{cond:bias} is necessary to control the bias of $\widehat{\mathcal{U}}(\Gamma)$. Interestingly, condition \ref{cond:bias} only requires that the bias of the estimator for maximum measured confounding ($|\psi - \psi_{-j^\prime}|$) converges at a $\sqrt{n}$-rate, while the non-maximums can be estimated consistently at any rate (which is guaranteed by condition~\ref{cond:empirical_process}). Intuitively, this occurs because the estimators for the non-maximums are only used to find the index of the true maximum, which is much easier statistically than estimating the value of the maximum. Indeed, access to merely consistent estimators of $| \psi - \psi_{-j}|$ for all $j \in [d]$, including the maximum, is enough to guarantee that the estimator for the maximum index, $\widehat j$, converges arbitrarily quickly to the true maximum index, $j^\prime$. A technical lemma establishing this is provided in Appendix~\ref{app:technical}.  

\bigskip

One can gain intuition for how the confidence interval for the ATE can be wider or narrower with a calibrated sensitivity model compared to a standard sensitivity model by examining~\eqref{eq:fx_upper_lim}. The limiting variance will determine the size of an asymptotically valid confidence interval. Without accounting for uncertainty in estimating measured confounding, as in a standard sensitivity analysis, the limiting variance is $\bbV \{ \phi(Z) \}$. By contrast, \eqref{eq:fx_upper_lim} yields the limiting variance $\bbV \left[ \phi(Z) +  \text{sign} (\psi- \psi_{-j^\prime}) \Gamma \big\{ \phi(Z) - \phi(Z_{-j^\prime}) \big\} \right]$.  This could be larger or smaller than $\bbV \{ \phi(Z) \}$, depending on the covariance between $\phi(Z)$ and $\text{sign} (\psi- \psi_{-j^\prime}) \Gamma \big\{ \phi(Z) - \phi(Z_{-j^\prime}) \big\}$.   We provide further analysis in the supplementary materials. Similar principles apply with other models, but with more complicated analysis. 

\subsection{Odds ratio}

Next, we consider estimating the upper bound on the ATE in the odds ratio model (calibrated sensitivity model~\ref{csm:max_loo_odds}). There are two extra nuances to this estimator. First, because the quantification of measured confounding is a supremum, it is not possible to estimate it at a $\sqrt{n}$-rate under nonparametric assumptions \citep{lepski1999estimation}.  Instead, we target the best projection of the propensity score onto a finite-dimensional model, and use that to estimate measured confounding. This approach has a long history in statistics (e.g., \citet{huber1967behavior}) and has been studied in many contexts in causal inference (e.g., \citet{semenova2021debiased, kennedy2023density}). Second, the nuisance functions depend on measured confounding (see, Proposition~\ref{prop:partial_id}).  Therefore, the estimator for the upper bound estimates measured confounding within the training sample, and uses it as an input to estimate the nuisance functions for the upper bound.  Finally, it estimates the bound in a separate estimation sample, as with the effect differences model.

\bigskip

For ease of exposition, we impose a further mild assumption on the covariates.
\begin{assumption} \label{asmp:covariates}
	\emph{Bounded covariates:} The support of the covariates is the d-dimensional unit cube.
\end{assumption}
This allows for construction of an estimator for measured confounding by finding the supremum over the unit cube.  This assumption could be relaxed to any known and bounded support.   Future work could relax this assumption entirely and incorporate an estimator of the support of the covariates. 

\bigskip

Definitions~\ref{def:odds} and \ref{def:odds_bound} below define the estimator for measured confounding and the estimator for the upper bound, respectively. We split the estimator into two parts to clearly illustrate how the estimator for the best projection of measured confounding is constructed.  Estimating the best projection of the propensity score corresponds to maximum likelihood estimation with binary regression (see, e.g., \citet{van2000asymptotic}, Example 5.40). Definition~\ref{def:odds} provides a specific example of binary regression --- logistic regression with no interactions --- but the ideas generalize to other link functions and to transformations of the covariates.
\begin{definition} \label{def:odds}
	Suppose Assumption~\ref{asmp:covariates} holds. Let $\Psi(x) = \frac{1}{1 + \exp(-x)}$ denote the logistic function. Define $\pi_1^\perp (X) = \Psi ( \beta^T X )$, where $\beta = \argmax_{b \in B} \bbE \left[\Psi ( b^T X )^A  \left\{ 1 - \Psi (b^T X) \right\}^{1-A} \right] $, where $B \subset \bbR^d$ is the parameter space. Construct an estimator for the propensity score as $\widehat \pi_1^\perp = \Psi( \widehat \beta^T X )$ 	where $\widehat \beta = \argmax_{b \in B} \bbP_n \left[ \Psi ( b^T X )^A  \left\{ 1- \Psi ( b^T X ) \right\}^{1-A} \right]$.  Because the model uses the logistic link, contains no interactions, and the support of the covariates is the unit cube, measured confounding $M$ (defined in \eqref{eq:M_odds}) corresponds to the largest absolute coefficient of $\beta$, i.e., $M = \max_{j \in [d]} | \beta_j|$, where $\beta_j$ is the $j^{th}$ coefficient of $\beta$. Therefore, let $\widehat M = \max_{j \in [d]} | \widehat \beta_j |.$
\end{definition}

As with the effect differences model, to facilitate $\sqrt{n}$-inference when measured confounding is a maximum, we invoke a separation condition.
\begin{assumption} \label{asmp:separation_max_odds}
	\emph{Separation of maximum in the odds ratio model:} There exists $j^\prime \in [d]$ such that $\left| \beta_{j^\prime} \right| > \left| \beta_{j} \right| \ \forall \ j \in [d] \setminus j^\prime$. 
\end{assumption}

\begin{definition} \label{def:odds_bound}
	Let $M$ and $\widehat M$ be as in Definition~\ref{def:odds} and $\mathcal{U}(\Gamma)$ as in \eqref{eq:odds_id_upper}. Construct an estimator for the upper bound as $\widehat{\mathcal{U}}(\Gamma) := \bbP_n \big[ \varphi_{U} \{ Z; \widehat \eta(\Gamma \widehat M) \} \big]$, where 
	\begin{itemize}
		\item $\varphi_{U}$ is the EIF for the upper bound on the ATE given in \citet{yadlowsky2022bounds}, also defined in \eqref{eq:yad_eif_upper} in the supplementary materials,
		\item the estimated nuisance functions $\widehat \eta = \{ \widehat \theta_a^\pm, \widehat \nu_a^\pm, \widehat \pi_a\}$ for $a \in \{0,1\}$ are constructed on a separate sample, where $\theta_a^\pm$ are defined in Proposition~\ref{prop:partial_id} and $\nu_a^\pm$ are defined in Lemma~\ref{lem:diff} in the supplementary materials, and 
		\item $\widehat M$ is constructed on a separate sample, according to Definition~\ref{def:odds}.
	\end{itemize}
\end{definition}

\noindent The next result provides a convergence guarantee for the estimator in Definition~\ref{def:odds_bound}.

\begin{restatable}{theorem}{thmoddsconv} \label{thm:odds_conv}
	\noindent \textbf{\emph{(Maximum LOO odds ratio model)}} Let $\mathcal{U}(\Gamma)$, $\widehat{\mathcal{U}}(\Gamma)$, $\varphi_{\mathcal{U}}$, and $\widehat \eta$ be as in Definition~\ref{def:odds_bound}, and let $j^\prime = \argmax_{j \in [d]} |\beta_j|$.  Suppose Assumptions~\ref{asmp:positivity}-\ref{asmp:bounded} and \ref{asmp:covariates}-\ref{asmp:continuity} hold, and
	\begin{enumerate}
		\item the distribution of $X$ is not concentrated on a $(d-1)$-dimensional affine subspace of its support, \label{cond:one_M}
		\item $\beta$ is at an inner point of $B$, the set of possible parameter values, \label{cond:two_M}
		\item $\widehat \eta(\Gamma \widehat M)$ is consistent for $\eta(\Gamma M)$ in the sense that $\left\lVert \varphi_U \{ Z; \widehat \eta(\Gamma \widehat M) \} - \varphi_U \{ Z; \eta(\Gamma M) \} \right\rVert_2 \inprob 0$, and \label{cond:one}
		\item the nuisance function estimators satisfy $\bbP \Big[ \varphi_U \{ Z; \widehat \eta(\Gamma \widehat M) \} - \varphi_U \{ Z; \eta(\Gamma \widehat M)\} \Big] = o_\bbP(n^{-1/2})$, where $\bbP$ denotes expectation conditional on the training data. \label{cond:two}
	\end{enumerate}
	Then,	
	\begin{equation} \label{eq:odds_conv}
		\widehat{\mathcal{U}}(\Gamma) - \mathcal{U}(\Gamma) = (\bbP_n - \bbE) \left[\varphi_{U} \{ Z; \eta(\Gamma M) \} +  \left\{ \frac{\partial}{\partial M} \mathcal{U}(\Gamma) \right\} \phi_M(Z) \right]  + o_{\bbP}(n^{-1/2}),
	\end{equation}
	where the derivative $\frac{\partial}{\partial M} \mathcal{U}(\Gamma)$ is defined in \eqref{eq:odds_deriv}, and $\phi_M(Z) =  e_{j^\prime}^T I_{\beta}^{-1}\ \big\{ \text{sign}(\beta_{j^\prime}) s(Z; \beta) \big\}$, where $e_j$ is the $j^{th}$ unit vector, $I$ is the Fisher information, and $s(\cdot)$ is the score function. With the model $\Psi ( \beta^T X )$ considered here, $I_{\beta} = \bbE \left[ \frac{\Psi^\prime (\beta^T X)^2}{\Psi(\beta^T X) \{ \beta^T X - \Psi(\beta^T X)\}} X X^T \right]$ and $s(Z; \beta) = \frac{A - \Psi(\beta^T X)}{\Psi(\beta^T X) \{ \beta^T X - \Psi(\beta^T X)\}} \Psi(\beta^T X) X.$
\end{restatable}

Theorem~\ref{thm:odds_conv} shows that the estimator for the upper bound in the odds ratio model behaves like a centered sample average plus asymptotically negligible error. Here, we give some intuition for the conditions of the result. Conditions \ref{cond:one_M} and \ref{cond:two_M} are necessary so that the estimator $\widehat M$ successfully targets $M$, the best projection of measured confounding onto the logistic model with no interactions. Condition~\ref{cond:one_M} is necessary so that the parameter $\beta$ is identifiable, in the statistical non-causal sense; see, e.g., pg. 62 in \citet{van2000asymptotic}.  Indeed, if this condition does not hold, then infinite parameter estimates can exist (see, e.g., Section 5.4.2., \citet{agresti2015foundations}).  Condition~\ref{cond:two_M} is a standard condition necessary for establishing the convergence guarantees of M-estimators (see \citet{van2000asymptotic}, Section 5), of which maximum likelihood estimation, used to construct $\widehat \beta$, is an example. Conditions~\ref{cond:one} and \ref{cond:two} control the error of the estimator of the upper bound. Condition~\ref{cond:one} is a mild consistency condition on the nuisance function estimators, which allows us to control the empirical process term. Meanwhile, condition~\ref{cond:two} asserts that the conditional bias due to estimating the nuisance functions --- with estimated measured confounding $\widehat M$ as an input --- is asymptotically negligible. This allows us to establish that the bias of the estimator $\widehat{\mathcal{U}}(\Gamma)$ is asymptotically negligible.  This is the same doubly robust-style condition from \citet{yadlowsky2022bounds} Assumption A.2, but with estimated measured confounding $\widehat M$ multiplied by the calibrated sensitivity parameter $\Gamma$ rather than a fixed sensitivity parameter $\gamma$. We anticipate this holds with many nonparametric estimators under assumptions like smoothness or sparsity.

\subsection{Inference} \label{sec:inference}

The results in the previous sections establish that estimators for the bounds on the ATE which incorporate an estimator for measured confounding can be $\sqrt{n}$-efficient and asymptotically linear under nonparametric doubly-robust style conditions. By the central limit theorem, Wald-type confidence intervals can be constructed using the normal distribution and the influence functions in Theorems~\ref{thm:fx_conv} \& \ref{thm:odds_conv} (here, we mean influence functions in the sense of asymptotically linear estimators, i.e., the terms $f(Z)$ in $(\bbP_n - \bbE) \{ f(Z) \}$). By combining one-sided confidence intervals for the upper and lower bounds on the ATE, one can conduct inference for the ATE itself. The following result gives the upper one-sided confidence interval for the ATE using the estimator for the upper bound $\widehat{\mathcal{U}}(\Gamma)$, while a lower one-sided interval can be constructed analogously using $\widehat{\mathcal{L}}(\Gamma)$. 

\begin{corollary} \label{cor:ci}
	Let $UB(\alpha) = \widehat{\mathcal{U}}(\Gamma) + \Phi^{-1}(1-\alpha) \sqrt{ \frac{\bbV\{ \phi_{\mathcal{U}}(Z) \} }{n}}$, where $\Phi(\cdot)$ is the cumulative distribution function of the standard normal.  Suppose Assumption~\ref{asmp:consistency} holds, and either the setup of Theorem~\ref{thm:fx_conv} holds and $\phi_{\mathcal{U}}$ is the influence function in \eqref{eq:fx_upper_lim}, or the setup of Theorem~\ref{thm:odds_conv} holds and $\phi_{\mathcal{U}}$ is the influence function in \eqref{eq:odds_conv}. Then, $\bbP \left\{ \psi_\ast \leq UB(\alpha) \right\} \geq 1-\alpha + o(1).$
\end{corollary}
Corollary~\ref{cor:ci} demonstrates how to construct an asymptotically valid one-sided confidence interval for the ATE: $(-\infty, UB(\alpha)]$. Analogously, one can construct $LB(\alpha)$ such that $\bbP \left\{ \psi_\ast \geq LB(\alpha) \right\} \geq 1-\alpha + o(1).$ Hence, one can construct an asymptotically valid two-sided $1-\alpha$ confidence interval for the ATE as $\left[ LB(\alpha / 2), UB(\alpha / 2) \right]$. 

\bigskip

Corollary~\ref{cor:ci} leverages the influence functions derived in Theorems~\ref{thm:fx_conv} \& \ref{thm:odds_conv} which depend on measured confounding and account for uncertainty in estimating $M$. In many cases, the limiting variance in the confidence intervals could be estimated with the sample variance of the relevant influence function $\phi_{\mathcal{U}}(Z)$. However, because the influence function could be quite complicated --- particularly in Theorem~\ref{thm:odds_conv}, depending on the choice of finite-dimensional model and link function defining $\pi_1^\perp$ --- one might prefer a resampling method such as the nonparametric bootstrap \citep{efron1992bootstrap}.  

\section{Illustrative Data Analysis} \label{sec:illustrations}

In this section, we illustrate our methods with a data analysis. We examine the effect of mothers' smoking on infant birth weight. We use a calibrated sensitivity analysis to assess the robustness of causal estimates to unmeasured confounding, and compare our results to what one would conclude from a standard sensitivity analysis with post hoc calibration.  In the supplementary materials, we also examine the effect of exposure to violence on attitudes towards peace in Darfur, and assess sensitivity to unmeasured confounding with both the effect differences model and the odds ratio model. Our code is available at \url{https://github.com/alecmcclean/Calibrated-sensitivity-models}.

\subsection{Data}

The data contains information about infant birth weights and mothers' health in Pennsylvania between 1989 and 1991 \citep{almond2005costs, cattaneo2010efficient}.  It has been used to estimate the causal effect of mothers' smoking on infant birth weight. We used a publicly available dataset consisting of 5,000 randomly subsampled observations from the original dataset (available at \url{https://github.com/mdcattaneo/replication-C_2010_JOE}). The treatment variables takes six values corresponding to ranges of cigarettes smoked daily ($\{0, 1-5, 6-10, 11-15, 16-20, >20\}$), which we dichotomize into an indicator for smoking.  The outcome variable is the birth weight of the infant in grams. Covariates include the race and age of the mother and father, the education level each attained, the mother’s marital status and foreign born status, the county of birth, and indicators for trimester of first prenatal care visit and mother’s alcohol use.

\subsection{Methods} \label{sec:illustrations_methods}

We analyzed the data assuming the maximum leave-one-out (LOO) effect differences model (calibrated sensitivity model~\ref{csm:max_loo_fx}).  We constructed estimators for the bounds on the ATE according to Definition~\ref{def:effect}, with $\Gamma \in \{0.5, 1, \dots, 5\}$.  We constructed estimators for the adjusted mean difference with different covariate subsets using the \texttt{npcausal} package in \texttt{R} \citep{kennedy2023npcausal, rcore2021language}.  We used 5 splits and estimated the propensity score and outcome regression functions with the \texttt{SuperLearner}, stacking the sample average and a random forest from the \texttt{ranger} package with default tuning parameters \citep{van2007superlearner, wright2017ranger}. We constructed a pointwise confidence band for the upper bound on the ATE according to \eqref{eq:fx_upper_lim} and a pointwise confidence band for the lower bound analogously.  As described in Section~\ref{sec:inference}, we constructed one-sided 97.5\% confidence bands for the upper and lower bounds, and took their intersection to construct a 95\% confidence band for the ATE. We also conducted a standard sensitivity analysis and post hoc calibration step, where we standardized the sensitivity parameter by estimated measured confounding.  The sensitivity analysis imposed the bound $\left| \psi_\ast - \psi \right| \leq \gamma$, while measured confounding was $M = \max_{j \in [d]} \left|\psi - \psi_{-j} \right|$. Finally, we also calculated a one-number summary of the robustness of the study, which captures how large the sensitivity parameter $\Gamma$ would need to be for the bounds on the ATE to include zero (i.e., to be uninformative as to the sign of the effect). Further formal details on the robustness value are in Appendix~\ref{app:robustness}.

\subsection{Results}

We first report the key estimates. The ATE estimate with all covariates includes was negative and significant: \textbf{--264g} (95\%CI: [--330g, --100g]), which, interpreted causally, says that a mother's smoking will lower their child's birth weight by 264g, on average.

\bigskip

Table~\ref{tab:measured} reports the five most important observed covariates, in terms of absolute change in adjusted mean difference. The most impactful confounder was \emph{the number of months since the last live birth for the mother}, which shifted the adjusted mean difference by $36.5$g when added as a covariate to the other observed covariates. There were several other important confounders, including the parents' race and age. There were large effects for geographic location and parents' education, but 95\% confidence intervals for these effects included zero.
\begin{table}[ht]
	\centering
	\begin{tabular}{L{0.3\textwidth} R{0.25\textwidth} R{0.35\textwidth}}
		\toprule
		\textbf{Variable} & \textbf{Confounding estimate} & \textbf{95\% confidence interval} \\
		\midrule
		Months since last live birth & 36.5 & [12.4,\,60.6] \\
		Parents' race & 31.1 & [8.1,\,54.1] \\
		Geographic location  & 28.5 & [0.0,\,59.3] \\
		Parents' age & 27.3 & [2.8,\,51.7] \\
		Parents' education & 25.7 & [0.0,\,62.9] \\
		\bottomrule
		\end{tabular}
		\caption{Top five most important confounders, in terms of estimated absolute change in the adjusted mean difference $\left| \widehat \psi - \widehat \psi_j \right|$, with 95\% confidence interval}
	\label{tab:measured}
\end{table}

\bigskip

Figure~\ref{fig:smoking} presents the ATE estimates, bounds on the ATE, and 95\% confidence bands for both sensitivity analysis approaches. The x-axis represents the level of the sensitivity parameter, $\Gamma$, in the calibrated sensitivity model, and the sensitivity parameter standardized by estimated measured confounding, $\gamma / \widehat M$, in the standard sensitivity model. The y-axis corresponds to the causal effect magnitude. The horizontal dashed line indicates the ATE estimate, while the horizontal dotted line marks zero. The dot-dash lines represent the estimated upper and lower bounds on the ATE. Point estimates for these bounds are identical across both sensitivity models, but their confidence intervals differ---solid lines for the calibrated sensitivity model and long dashes for the standard model.

\begin{remark}
The estimated bounds are identical in both analyses after standardizing the x-axis. This occurs because for the sensitivity model $\left| \psi_\ast - \psi \right| \leq \gamma$, standardizing the sensitivity parameter by estimated measured confounding (i.e., using $\gamma / \widehat M$ as the x-axis) yields the same point estimates for the bounds as the calibrated sensitivity analysis with the effect differences model. We provide further details in Appendix~\ref{app:invariance_intuition}. The key difference between approaches lies in their confidence intervals, which properly account for estimation uncertainty in the calibrated model.
\end{remark}

\bigskip

The calibrated sensitivity results in Figure~\ref{fig:smoking} demonstrate that the estimated bounds on the ATE remain significantly different from zero when unmeasured confounding is less than seven times the magnitude of measured confounding. The confidence interval for the ATE intersects zero when $\Gamma \approx 4$, indicating that the statistical significance of the estimated effect would only be nullified if unmeasured confounding were at least four times as large as measured confounding. Specifically, an unmeasured confounder, when
added to all observed covariates, would need to shift the adjusted mean difference by at least four
times as much as \emph{the number of months since last live birth} changed the adjusted mean difference when it  was added to the other observed covariates. Given that months since last live birth changed the adjusted mean difference by $36.5$g, an unmeasured confounder with a similar relationship to all observed covariates as the relationship between months since last live birth and all other covariates would need to induce a change of approximately $150$g to nullify our results.

\bigskip

The magnitude of effects from other measured confounders provides additional context for assessing robustness. Table~\ref{tab:measured} demonstrates that at least three of the five most important measured confounders had statistically significant effects, indicating that several observed covariates account for confounding.  With this rich covariate set that includes established predictors of birth weight such as birth spacing, parental demographics, and socioeconomic factors, one could argue that the causal effect is robust to unmeasured confounding. While subject-matter expertise is crucial to assess the plausibility of such a change in birth weight due to an unmeasured confounder, the implausibility that an unmeasured confounder would have four times the impact of these well-established confounders could strengthen confidence in our findings. 

\bigskip

While both the traditional and calibrated sensitivity analyses suggest that each effect is somewhat robust to confounding, they suggest it to differing degrees.  Indeed, these results illustrate how not accounting for uncertainty in estimating measured confounding under-estimating robustness to unmeasured confounding.  The calibrated sensitivity model confidence intervals always contain the standard sensitivity model confidence intervals.  Therefore, the standard analysis over-estimates robustness to unmeasured confounding. It is also possible for standard approaches to underestimate robustness to unmeasured confounding. In Appendix~\ref{app:invariance_intuition}, we explore this mathematically. We show that if the estimator for measured confounding and the estimator for the un-calibrated upper bound $u(\gamma)$ are negatively correlated, and the estimator for measured confounding is sufficiently poor relative to the estimator for the upper bound, then standard methods will underestimate robustness. This creates a scenario where the standard approach provides overly pessimistic conclusions about the study's robustness to unmeasured confounding.

\bigskip

Finally, the robustness value estimate was $7.24$ (95\% confidence interval: $[2.38, 12.1]$). In other words, we estimated that unmeasured confounding would need to be roughly seven times as strong as measured confounding for the bounds on the ATE to include zero, indicating the study's robustness threshold. This estimate had a wide confidence interval, however, ranging from two to twelve. This robustness threshold of $7.24$ can be visualized in Figure~\ref{fig:smoking} as the point where the purple dot-dash line (upper bound) would cross the x-axis.

\begin{figure}[htbp]
\centering
\includegraphics[height = 4in]{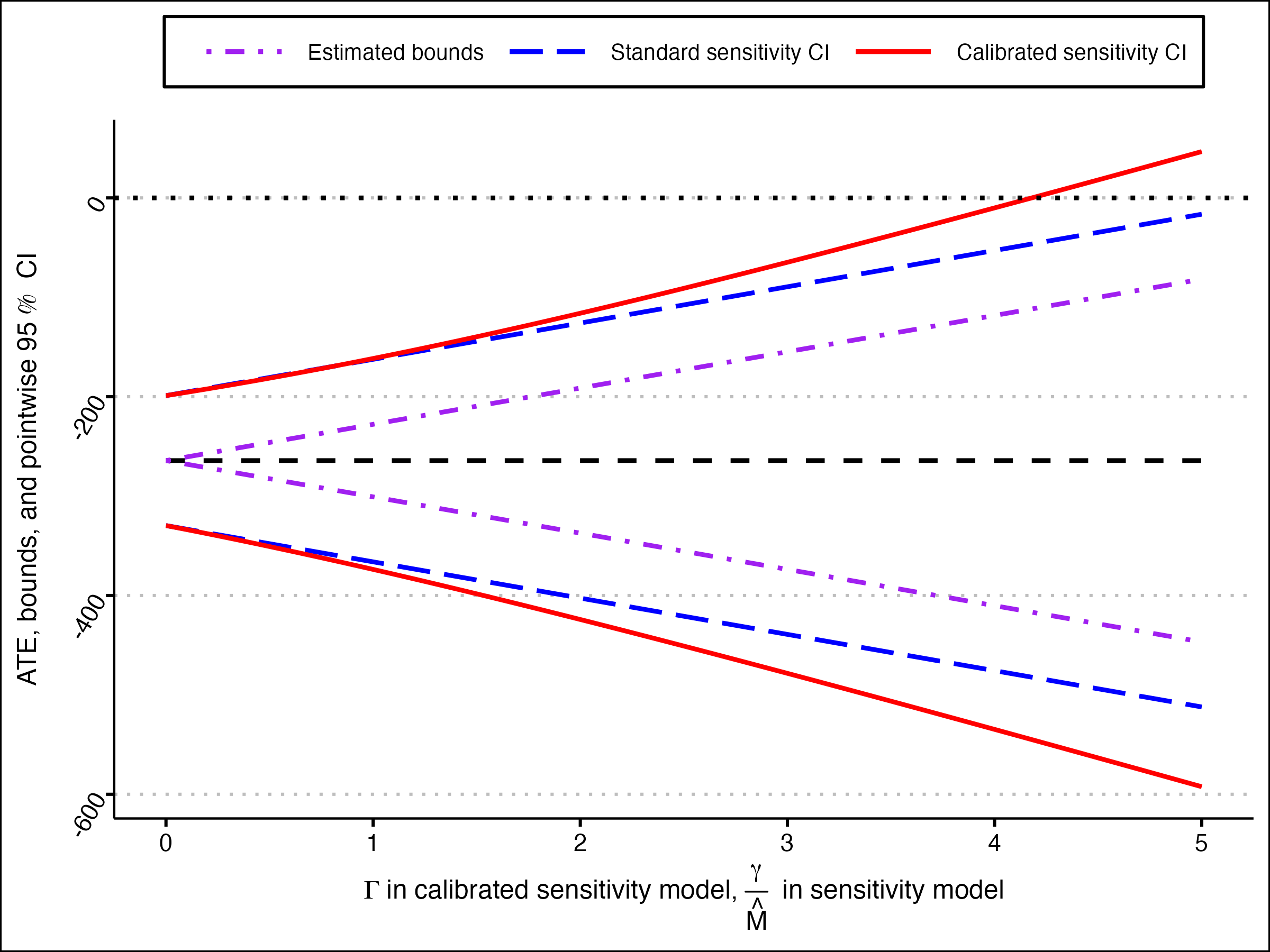}
\caption{\textbf{Calibrated sensitivity analysis for the effect of maternal smoking on infant birth weight.} The figure displays the estimated average treatment effect (ATE) and its sensitivity to unmeasured confounding. The x-axis represents the strength of unmeasured confounding: $\Gamma$ for the calibrated sensitivity model and $\gamma / \widehat M$ for the standard sensitivity model, where larger values indicate stronger unmeasured confounding relative to measured confounding. The y-axis shows the causal effect on birth weight in grams. The black horizontal dashed line shows the point estimate of the ATE (-264g), while the black horizontal dotted line marks zero effect. The purple dot-dash lines represent the upper and lower bounds on the ATE under different levels of unmeasured confounding. Confidence intervals are shown as solid red lines (calibrated model) and long blue dashes (standard model). The calibrated model accounts for uncertainty in estimating measured confounding, producing wider confidence intervals. Statistical significance is lost when $\Gamma \approx 4$, meaning an unmeasured confounder would need to be four times stronger than the strongest measured confounder to nullify the smoking effect.}
\label{fig:smoking}
\end{figure}

\section{Discussion: formal benchmarking} \label{sec:discussion}

In Section~\ref{sec:interpretation}, we emphasized that interpreting our calibrated sensitivity models depends on the observed covariates. Here, we revisit this point through the lens of "informal" versus "formal" benchmarking—a popular distinction in the sensitivity analysis literature. While our models appear to fail the specific formal benchmarking criterion proposed in the literature, we argue this criterion is not appropriate for these models. We illustrate this through two examples.

\bigskip

The core idea behind formal benchmarking is intuitive: if an observed covariate is a perfect proxy for an unmeasured confounder, then calibrated sensitivity bounds using that covariate should include the true causal effect. This follows because observing the unmeasured confounder would identify the causal effect. In other words, if we knew $M = U_\ast$ (where $U_\ast$ represents the true effect of unmeasured confounders), then imposing $U \leq \Gamma M = \Gamma U_*$ with $\Gamma=1$ will yield bounds containing the causal effect.

\bigskip

However, applying this criterion to our models raises a fundamental question: how should we formalize a ``perfect proxy"? \citet[Sections 4.4 and 6.2]{cinelli2020making} propose using marginal exchangeability:

\begin{property} \label{property:formal}
	Let $M(\cdot): \mathcal{X} \times \mathcal{P} \to \bbR$ be a notion of measured confounding that can be expressed as a direct function of a covariate or set of covariates. If there exists some covariate $X$ \textbf{that is exchangeable} with the unmeasured confounder, then a calibrated sensitivity model $U \leq \Gamma M(X)$ is \emph{\textbf{formal}} if it satisfies $\psi_\ast \in [\mathcal{L}(1), \mathcal{U}(1)]$.
\end{property}

Property~\ref{property:formal} formalizes a perfect proxy as a covariate that is maringally exchangeable with the unmeasured confounder (we give an example below). This criterion suits explained variance models that parameterize bias through the unmeasured confounder's direct explanatory power \citep{huang2025variance,chernozhukov2022long}. However, Property~\ref{property:formal} is inappropriate for our models because they capture the effect of an unmeasured confounder \emph{when added} to observed covariates. Thus, a perfect proxy is not simply an exchangeable covariate, but one whose \emph{additional} impact---when added to the other observed covariates---matches the additional impact of the unmeasured confounder \emph{when added to all observed covariates.}

\bigskip

This distinction is crucial. To clarify, we present two illustrative examples. Both involve two covariates: $X$ is observed and $W$ is unobserved, with $X$ serving as a proxy for $W$. Throughout, we use the simple leave-one-out calibrated sensitivity model which imposes the bound $| \psi_\ast - \psi_X  | \leq \Gamma \left| \psi_X - \psi_{\emptyset} \right|$, although similar principles apply to the other examples in this paper.

\bigskip

\noindent \textbf{Exchangeable covariates may not be perfect proxies} \\
The following example is adapted to binary treatment from \citet[Section 6.2]{cinelli2020making}):
\[
X \sim \text{Unif}(-1,1),\quad W \sim \text{Unif}(-1,1),\quad A \sim \text{Bernoulli} \left(\tfrac{X + W + 2}{4} \right),\quad Y = X + W + \varepsilon,
\]
where $\varepsilon \sim N(0,1)$ and all random variables are independent. There is no treatment effect, so $\psi_\ast = 0$. However, there is unmeasured confounding: $\psi_X = \bbE \{ \mu_1(X) - \mu_0(X) \} = \tfrac{\log 3}{3} \neq 0$.

\bigskip

In this example, the formal benchmarking criterion and would assert that $X$ is a perfect proxy for $W$ because they are exchangeable marginally. Therefore, Property~\ref{property:formal} should hold and bounds with $\Gamma = 1$ should include the true causal effect. However, setting $\Gamma = 1$ does not yield bounds that include the true causal effect: \( \psi_{\emptyset} \equiv \bbE(Y \mid A=1) - \bbE(Y \mid A=0) = \tfrac{2}{3} \) and therefore \( 0 \notin \psi_X \pm \left| \psi_X - \psi_\emptyset \right| \approx [0.066,0.667].\) 

\bigskip

That the bounds with $\Gamma =1$ fail to contain the true causal effect is not an issue with the calibrated sensitivity model. Rather, this example reveals the limitation of formal benchmarking as it is formalized in Property~\ref{property:formal}.  Specifically, adding $W$ to the set ${X}$ has a different impact than adding $X$ to the empty set, even though $X$ and $W$ are exchangeable marginally. This illustrates why Property~\ref{property:formal} is inappropriate for our approach: marginal exchangeability does not guarantee that a covariate provides perfect information about the effect of an unmeasured confounder if it were added to all the covariates already observed.

\bigskip

\noindent \textbf{Perfect proxies need not be exchangeable} 

We modify the setup to illustrate a case where $X$ is a perfect proxy for $W$ even though they are not exchangeable. The modified data generating process is 
\[
\bbP(A=1 \mid X, W) = \left(\tfrac{X + \theta W + 2}{4} \right) \text{ and } \bbE(Y \mid A, X, W) = X + \theta W,
\]
where $\theta = \sqrt{1 / (2 \log 3 - 1)} \approx 0.835$. Here, the bias induced by omitting $W$ from the set $\{X, W\}$ is equal to the bias induced by omitting $X$ from the set $\{X\}$. That is: \( \left| \psi_\ast - \psi_X \right| = \left| \psi_X - \psi_\emptyset \right| = \frac{\log 3}{6 \log 3 - 3} \approx 0.306. \) Therefore, \( 0 \in \psi_X \pm \left| \psi_X - \psi_\emptyset \right| = \left[ 0, \frac{2}{6 \log 3 - 3} \right]. \)

\bigskip

In this scenario, the key idea underpinning formal benchmarking is preserved---not due to marginal exchangeability---but because the effect of $W$ when added to $X$ is the same as the effect of $X$ when added to the empty set.  

\bigskip

These examples highlight that the appropriate notion of a perfect proxy for unmeasured confounding depends critically on the specific calibrated sensitivity model employed, and that Property~\ref{property:formal} is inappropriate for the models in this paper. For our calibrated models, which focus on the additional effect of an unmeasured confounder, perfect proxies are characterized by matching additional impacts rather than marginal exchangeability.

\section*{Acknowledgments}

The authors thank Melody Huang, Carlos Cinelli, several anonymous reviewers of a prior version of the paper, and participants in the Carnegie Mellon University causal inference reading group, CMStatistics 2023, and the 2024 American Causal Inference Conference for helpful comments and feedback.

\section*{References}

\vspace{-0.3in}
\bibliographystyle{plainnat}
\bibliography{references}

\begin{thebibliography}{61}
\providecommand{\natexlab}[1]{#1}
\providecommand{\url}[1]{\texttt{#1}}
\expandafter\ifx\csname urlstyle\endcsname\relax
  \providecommand{\doi}[1]{doi: #1}\else
  \providecommand{\doi}{doi: \begingroup \urlstyle{rm}\Url}\fi

\bibitem[Agresti(2015)]{agresti2015foundations}
Alan Agresti.
\newblock \emph{Foundations of linear and generalized linear models}.
\newblock John Wiley \& Sons, 2015.

\bibitem[Almond et~al.(2005)Almond, Chay, and Lee]{almond2005costs}
Douglas Almond, Kenneth~Y Chay, and David~S Lee.
\newblock The costs of low birth weight.
\newblock \emph{The Quarterly Journal of Economics}, 120\penalty0 (3):\penalty0 1031--1083, 2005.

\bibitem[Bickel et~al.(1993)Bickel, Klaassen, Ritov, and Wellner]{bickel1993efficient}
Peter~J Bickel, Chris~AJ Klaassen, Ya'acov Ritov, and Jon~A Wellner.
\newblock \emph{Efficient and Adaptive Estimation for Semiparametric Models}.
\newblock Baltimore: Johns Hopkins University Press, 1993.

\bibitem[Birg{\'e} and Massart(1995)]{birge1995estimation}
Lucien Birg{\'e} and Pascal Massart.
\newblock Estimation of integral functionals of a density.
\newblock \emph{The Annals of Statistics}, 23\penalty0 (1):\penalty0 11--29, 1995.

\bibitem[Bonvini and Kennedy(2022)]{bonvini2022sensitivitya}
Matteo Bonvini and Edward~H Kennedy.
\newblock Sensitivity analysis via the proportion of unmeasured confounding.
\newblock \emph{Journal of the American Statistical Association}, 117\penalty0 (539):\penalty0 1540--1550, 2022.

\bibitem[Brumback et~al.(2004)Brumback, Hern{\'a}n, Haneuse, and Robins]{brumback2004sensitivity}
Babette~A Brumback, Miguel~A Hern{\'a}n, Sebastien~JPA Haneuse, and James~M Robins.
\newblock Sensitivity analyses for unmeasured confounding assuming a marginal structural model for repeated measures.
\newblock \emph{Statistics in medicine}, 23\penalty0 (5):\penalty0 749--767, 2004.

\bibitem[Cattaneo(2010)]{cattaneo2010efficient}
Matias~D Cattaneo.
\newblock Efficient semiparametric estimation of multi-valued treatment effects under ignorability.
\newblock \emph{Journal of Econometrics}, 155\penalty0 (2):\penalty0 138--154, 2010.

\bibitem[Chen et~al.(2022)Chen, Syrgkanis, and Austern]{chen2022debiased}
Qizhao Chen, Vasilis Syrgkanis, and Morgane Austern.
\newblock Debiased machine learning without sample-splitting for stable estimators.
\newblock \emph{Advances in Neural Information Processing Systems}, 35:\penalty0 3096--3109, 2022.

\bibitem[Chernozhukov et~al.(2018)Chernozhukov, Chetverikov, Demirer, Duflo, Hansen, Newey, and Robins]{chernozhukov2018double}
Victor Chernozhukov, Denis Chetverikov, Mert Demirer, Esther Duflo, Christian Hansen, Whitney Newey, and James Robins.
\newblock Double/debiased machine learning for treatment and structural parameters.
\newblock \emph{The Econometrics Journal}, 21\penalty0 (1):\penalty0 C1--C68, 2018.

\bibitem[Chernozhukov et~al.(2022)Chernozhukov, Cinelli, Newey, Sharma, and Syrgkanis]{chernozhukov2022long}
Victor Chernozhukov, Carlos Cinelli, Whitney Newey, Amit Sharma, and Vasilis Syrgkanis.
\newblock Long story short: Omitted variable bias in causal machine learning.
\newblock Technical report, National Bureau of Economic Research, 2022.

\bibitem[Cinelli and Hazlett(2020)]{cinelli2020making}
Carlos Cinelli and Chad Hazlett.
\newblock Making sense of sensitivity: Extending omitted variable bias.
\newblock \emph{Journal of the Royal Statistical Society Series B-Statistical Methodology}, 82\penalty0 (1):\penalty0 39--67, 2020.

\bibitem[Cinelli et~al.(2020)Cinelli, Ferwerda, and Hazlett]{cinelli2020sensemakr}
Carlos Cinelli, Jeremy Ferwerda, and Chad Hazlett.
\newblock sensemakr: Sensitivity analysis tools for ols in r and stata.
\newblock \emph{Available at SSRN 3588978}, 2020.

\bibitem[D{\'\i}az and van~der Laan(2013)]{diaz2013sensitivity}
Iv{\'a}n D{\'\i}az and Mark~J van~der Laan.
\newblock Sensitivity analysis for causal inference under unmeasured confounding and measurement error problems.
\newblock \emph{The international journal of biostatistics}, 9\penalty0 (2):\penalty0 149--160, 2013.

\bibitem[D{\'\i}az et~al.(2018)D{\'\i}az, Savenkov, and Ballman]{diaz2018targeted}
Iv{\'a}n D{\'\i}az, Oleksandr Savenkov, and Karla Ballman.
\newblock Targeted learning ensembles for optimal individualized treatment rules with time-to-event outcomes.
\newblock \emph{Biometrika}, 105\penalty0 (3):\penalty0 723--738, 2018.

\bibitem[Efron(1992)]{efron1992bootstrap}
Bradley Efron.
\newblock Bootstrap methods: another look at the jackknife.
\newblock In \emph{Breakthroughs in statistics: Methodology and distribution}, pages 569--593. Springer, 1992.

\bibitem[Fisher and Fisher(2023)]{fisher2023three}
Aaron Fisher and Virginia Fisher.
\newblock Three-way cross-fitting and pseudo-outcome regression for estimation of conditional effects and other linear functionals.
\newblock \emph{arXiv preprint arXiv:2306.07230}, 2023.

\bibitem[Franks et~al.(2020)Franks, D’Amour, and Feller]{franks2020flexible}
Alexander~M. Franks, Alexander D’Amour, and Avi Feller.
\newblock Flexible sensitivity analysis for observational studies without observable implications.
\newblock \emph{Journal of the American Statistical Association}, 115\penalty0 (532):\penalty0 1730--1746, 2020.

\bibitem[Gy{\"o}rfi et~al.(2002)Gy{\"o}rfi, Kohler, Walk, et~al.]{gyorfi2002distribution}
L{\'a}szl{\'o} Gy{\"o}rfi, Michael Kohler, Harro Walk, et~al.
\newblock \emph{A distribution-free theory of nonparametric regression}, volume~1.
\newblock Springer, 2002.

\bibitem[Hazlett(2020)]{hazlett2020angry}
Chad Hazlett.
\newblock Angry or weary? how violence impacts attitudes toward peace among darfurian refugees.
\newblock \emph{Journal of Conflict Resolution}, 64\penalty0 (5):\penalty0 844--870, 2020.

\bibitem[Hines et~al.(2022)Hines, Dukes, Diaz-Ordaz, and Vansteelandt]{hines2022demystifying}
Oliver Hines, Oliver Dukes, Karla Diaz-Ordaz, and Stijn Vansteelandt.
\newblock Demystifying statistical learning based on efficient influence functions.
\newblock \emph{The American Statistician}, 76\penalty0 (3):\penalty0 292--304, 2022.

\bibitem[Hong et~al.(2021)Hong, Yang, and Qin]{hong2021did}
Guanglei Hong, Fan Yang, and Xu~Qin.
\newblock Did you conduct a sensitivity analysis? a new weighting-based approach for evaluations of the average treatment effect for the treated.
\newblock \emph{Journal of the Royal Statistical Society Series A: Statistics in Society}, 184\penalty0 (1):\penalty0 227--254, 2021.

\bibitem[Huang and Pimentel(2025)]{huang2025variance}
Melody Huang and Samuel~D Pimentel.
\newblock Variance-based sensitivity analysis for weighting estimators results in more informative bounds.
\newblock \emph{Biometrika}, 112\penalty0 (1):\penalty0 asae040, 2025.

\bibitem[Huber et~al.(1967)]{huber1967behavior}
Peter~J Huber et~al.
\newblock The behavior of maximum likelihood estimates under nonstandard conditions.
\newblock In \emph{Proceedings of the fifth Berkeley symposium on mathematical statistics and probability}, volume~1, pages 221--233. Berkeley, CA: University of California Press, 1967.

\bibitem[Kennedy(2022)]{kennedy2022semiparametric}
Edward~H Kennedy.
\newblock Semiparametric doubly robust targeted double machine learning: a review.
\newblock \emph{arXiv preprint arXiv:2203.06469}, 2022.

\bibitem[Kennedy(2023{\natexlab{a}})]{kennedy2023npcausal}
Edward~H. Kennedy.
\newblock \emph{npcausal: Nonparametric causal inference methods}, 2023{\natexlab{a}}.
\newblock R package version 0.1.0.

\bibitem[Kennedy(2023{\natexlab{b}})]{kennedy2023towards}
Edward~H Kennedy.
\newblock Towards optimal doubly robust estimation of heterogeneous causal effects.
\newblock \emph{Electronic Journal of Statistics}, 17\penalty0 (2):\penalty0 3008--3049, 2023{\natexlab{b}}.

\bibitem[Kennedy et~al.(2020)Kennedy, Balakrishnan, and G’Sell]{kennedy2020sharp}
Edward~H. Kennedy, Sivaraman Balakrishnan, and Max G’Sell.
\newblock {Sharp instruments for classifying compliers and generalizing causal effects}.
\newblock \emph{The Annals of Statistics}, 48\penalty0 (4):\penalty0 2008 -- 2030, 2020.

\bibitem[Kennedy et~al.(2023)Kennedy, Balakrishnan, and Wasserman]{kennedy2023density}
EH~Kennedy, S~Balakrishnan, and LA~Wasserman.
\newblock Semiparametric counterfactual density estimation.
\newblock \emph{Biometrika}, page asad017, 2023.

\bibitem[Kosorok(2008)]{kosorok2008introduction}
Michael~R Kosorok.
\newblock Introduction to empirical processes.
\newblock \emph{Introduction to Empirical Processes and Semiparametric Inference}, 2008.

\bibitem[Laurent(1997)]{laurent1997estimation}
B{\'e}atrice Laurent.
\newblock Estimation of integral functionals of a density and its derivatives.
\newblock \emph{Bernoulli}, pages 181--211, 1997.

\bibitem[Lepski et~al.(1999)Lepski, Nemirovski, and Spokoiny]{lepski1999estimation}
Oleg Lepski, Arkady Nemirovski, and Vladimir Spokoiny.
\newblock On estimation of the l r norm of a regression function.
\newblock \emph{Probability theory and related fields}, 113:\penalty0 221--253, 1999.

\bibitem[Luedtke and van Der~Laan(2016)]{luedtke2016statistical}
Alexander~R Luedtke and Mark~J van Der~Laan.
\newblock Statistical inference for the mean outcome under a possibly non-unique optimal treatment strategy.
\newblock \emph{Annals of statistics}, 44\penalty0 (2):\penalty0 713, 2016.

\bibitem[Luedtke et~al.(2015)Luedtke, Diaz, and van~der Laan]{luedtke2015statistics}
Alexander~R Luedtke, Ivan Diaz, and Mark~J van~der Laan.
\newblock The statistics of sensitivity analyses.
\newblock \emph{U.C. Berkeley Division of Biostatistics Working Paper Series}, 341, 2015.

\bibitem[Manski(1990)]{manski1990bounds}
Charles~F. Manski.
\newblock Nonparametric bounds on treatment effects.
\newblock \emph{The American Economic Review}, 80\penalty0 (2):\penalty0 319--323, 1990.

\bibitem[Masten and Poirier(2018)]{masten2018identification}
Matthew~A. Masten and Alexandre Poirier.
\newblock Identification of treatment effects under conditional partial independence.
\newblock \emph{Econometrica}, 86\penalty0 (1):\penalty0 317--351, 2018.

\bibitem[McClean et~al.(2024)McClean, Branson, and Kennedy]{mcclean2024nonparametric}
Alec McClean, Zach Branson, and Edward~H Kennedy.
\newblock Nonparametric estimation of conditional incremental effects.
\newblock \emph{Journal of Causal Inference}, 12\penalty0 (1):\penalty0 20230024, 2024.

\bibitem[Mises(1947)]{von1947asymptotic}
RV~Mises.
\newblock On the asymptotic distribution of differentiable statistical functions.
\newblock \emph{The annals of mathematical statistics}, 18\penalty0 (3):\penalty0 309--348, 1947.

\bibitem[Nabi et~al.(2024)Nabi, Bonvini, Kennedy, Huang, Smid, and Scharfstein]{nabi2024semiparametric}
Razieh Nabi, Matteo Bonvini, Edward~H Kennedy, Ming-Yueh Huang, Marcela Smid, and Daniel~O Scharfstein.
\newblock Semiparametric sensitivity analysis: unmeasured confounding in observational studies.
\newblock \emph{Biometrics}, 80\penalty0 (4):\penalty0 ujae106, 2024.

\bibitem[{R Core Team}(2021)]{rcore2021language}
{R Core Team}.
\newblock \emph{R: A Language and Environment for Statistical Computing}.
\newblock R Foundation for Statistical Computing, Vienna, Austria, 2021.
\newblock URL \url{https://www.R-project.org/}.

\bibitem[Robins(1989)]{robins1989analysis}
James~M Robins.
\newblock The analysis of randomized and non-randomized aids treatment trials using a new approach to causal inference in longitudinal studies.
\newblock \emph{Health service research methodology: a focus on AIDS}, pages 113--159, 1989.

\bibitem[Robins(1999)]{robins1999association}
James~M Robins.
\newblock Association, causation, and marginal structural models.
\newblock \emph{Synthese}, 121\penalty0 (1/2):\penalty0 151--179, 1999.

\bibitem[Robins(2002)]{robins2002covariance}
James~M Robins.
\newblock [covariance adjustment in randomized experiments and observational studies]: Comment.
\newblock \emph{Statistical Science}, 17\penalty0 (3):\penalty0 309--321, 2002.

\bibitem[Robins et~al.(1994)Robins, Rotnitzky, and Zhao]{robins1994estimation}
James~M Robins, Andrea Rotnitzky, and Lue~Ping Zhao.
\newblock Estimation of regression coefficients when some regressors are not always observed.
\newblock \emph{Journal of the American statistical Association}, 89\penalty0 (427):\penalty0 846--866, 1994.

\bibitem[Rosenbaum(2002{\natexlab{a}})]{rosenbaum2002covariance}
Paul~R Rosenbaum.
\newblock Covariance adjustment in randomized experiments and observational studies.
\newblock \emph{Statistical Science}, 17\penalty0 (3):\penalty0 286--327, 2002{\natexlab{a}}.

\bibitem[Rosenbaum(2002{\natexlab{b}})]{rosenbaum2002rejoinder}
Paul~R Rosenbaum.
\newblock [covariance adjustment in randomized experiments and observational studies]: Rejoinder.
\newblock \emph{Statistical Science}, 17\penalty0 (3):\penalty0 321--327, 2002{\natexlab{b}}.

\bibitem[Rosenbaum(2002{\natexlab{c}})]{rosenbaum2002sensitivity}
Paul~R. Rosenbaum.
\newblock \emph{Sensitivity to Hidden Bias}.
\newblock Springer New York, New York, NY, 2002{\natexlab{c}}.

\bibitem[Semenova and Chernozhukov(2021)]{semenova2021debiased}
Vira Semenova and Victor Chernozhukov.
\newblock Debiased machine learning of conditional average treatment effects and other causal functions.
\newblock \emph{The Econometrics Journal}, 24\penalty0 (2):\penalty0 264--289, 2021.

\bibitem[Sj{\"o}lander et~al.(2022)Sj{\"o}lander, Gabriel, and Cioc{\u{a}}nea-Teodorescu]{sjolander2022sensitivity}
Arvid Sj{\"o}lander, Erin~E Gabriel, and Iuliana Cioc{\u{a}}nea-Teodorescu.
\newblock Sensitivity analysis for causal effects with generalized linear models.
\newblock \emph{Journal of Causal Inference}, 10\penalty0 (1):\penalty0 441--479, 2022.

\bibitem[Tan(2006)]{tan2006distributional}
Zhiqiang Tan.
\newblock A distributional approach for causal inference using propensity scores.
\newblock \emph{Journal of the American Statistical Association}, 101\penalty0 (476):\penalty0 1619--1637, 2006.

\bibitem[Tchetgen and VanderWeele(2012)]{tchetgen2012causal}
Eric J~Tchetgen Tchetgen and Tyler~J VanderWeele.
\newblock On causal inference in the presence of interference.
\newblock \emph{Statistical methods in medical research}, 21\penalty0 (1):\penalty0 55--75, 2012.

\bibitem[Tsiatis(2006)]{tsiatis2006semiparametric}
Anastasios~A Tsiatis.
\newblock \emph{Semiparametric Theory and Missing Data}.
\newblock New York: Springer, 2006.

\bibitem[Van~der Laan and Robins(2003)]{van2003unified}
Mark~J Van~der Laan and James~M Robins.
\newblock \emph{Unified methods for censored longitudinal data and causality}, volume~5.
\newblock Springer, 2003.

\bibitem[van~der Laan et~al.(2007)van~der Laan, Polley, and Hubbard]{van2007superlearner}
Mark~J. van~der Laan, Eric~C Polley, and Alan~E. Hubbard.
\newblock Super learner.
\newblock \emph{Statistical Applications in Genetics and Molecular Biology}, 6\penalty0 (1), 2007.

\bibitem[{van der Vaart}(2000)]{van2000asymptotic}
Aad~W {van der Vaart}.
\newblock \emph{Asymptotic Statistics}.
\newblock Cambridge: Cambridge University Press, 2000.

\bibitem[{van der Vaart} and Wellner(1996)]{van1996weak}
Aad~W {van der Vaart} and Jon~A Wellner.
\newblock \emph{Weak Convergence and Empirical Processes}.
\newblock Springer, 1996.

\bibitem[van~der Vaart(2002)]{van2002semiparametric}
AW~van~der Vaart.
\newblock \emph{Semiparametric Statistics}.
\newblock Springer, 2002.

\bibitem[Veitch and Zaveri(2020)]{veitch2020sense}
Victor Veitch and Anisha Zaveri.
\newblock Sense and sensitivity analysis: Simple post-hoc analysis of bias due to unobserved confounding.
\newblock In H.~Larochelle, M.~Ranzato, R.~Hadsell, M.F. Balcan, and H.~Lin, editors, \emph{Advances in Neural Information Processing Systems}, volume~33, pages 10999--11009. Curran Associates, Inc., 2020.

\bibitem[Westreich and Cole(2010)]{westreich2010invited}
Daniel Westreich and Stephen~R Cole.
\newblock Invited commentary: positivity in practice.
\newblock \emph{American journal of epidemiology}, 171\penalty0 (6):\penalty0 674--677, 2010.

\bibitem[Wright and Ziegler(2017)]{wright2017ranger}
Marvin~N. Wright and Andreas Ziegler.
\newblock {ranger}: A fast implementation of random forests for high dimensional data in {C++} and {R}.
\newblock \emph{Journal of Statistical Software}, 77\penalty0 (1):\penalty0 1--17, 2017.

\bibitem[Yadlowsky et~al.(2022)Yadlowsky, Namkoong, Basu, Duchi, and Tian]{yadlowsky2022bounds}
Steve Yadlowsky, Hongseok Namkoong, Sanjay Basu, John Duchi, and Lu~Tian.
\newblock Bounds on the conditional and average treatment effect with unobserved confounding factors.
\newblock \emph{The Annals of Statistics}, 50\penalty0 (5):\penalty0 2587--2615, 2022.

\bibitem[Zheng et~al.(2024)Zheng, Wu, D’Amour, and Franks]{zheng2024sensitivity}
Jiajing Zheng, Jiaxi Wu, Alexander D’Amour, and Alexander Franks.
\newblock Sensitivity to unobserved confounding in studies with factor-structured outcomes.
\newblock \emph{Journal of the American Statistical Association}, 119\penalty0 (547):\penalty0 2026--2037, 2024.

\end{thebibliography}

\newpage
\appendix
{\huge \noindent \textbf{Supplementary Materials}}

\bigskip

\noindent These supplementary materials are split into nine sections:
\begin{enumerate}
\item[Appendix \ref{app:invariance_intuition}] discusses an invariance property of the bounds and provides an illustrative analysis demonstrating when standard sensitivity analyses with post hoc calibration over- or under-estimate robustness to unmeasured confounding.
\item[Appendix~\ref{app:robustness}] describes the robustness value, a one-number summary of the robustness of a study to unmeasured confounding, and how to estimate it with calibrated sensitivity models.
\item[Appendix \ref{app:outcome_est}] provides estimation and inference results with the outcome model (calibrated sensitivity model~\ref{csm:avg_lso_out}), which complement those for the effect differences and odds ratio models in Section~\ref{sec:illustrations}.
\item[Appendix \ref{app:illustrations}] describes two data analyses we conducted with the effect differences and odds ratio models, assessing the effect of exposure to violence on attitudes towards peace in Darfur, and the sensitivity of that causal effect to unmeasured confounding.
\item[Appendix \ref{app:id_results}] provides results for Section~\ref{sec:partial_id}, Proposition~\ref{prop:partial_id} and Lemma~\ref{lem:show_diff_monotone}. Lemma~\ref{lem:show_diff_monotone} is proven by stating and proving a more detailed result, Lemma~\ref{lem:diff}.
\item[Appendix \ref{app:eifs}] provides the relevant efficient influence functions for estimating measured confounding and the bounds, deriving new efficient influence functions for the outcome model.
\item[Appendix \ref{app:est}] provides results for Section~\ref{sec:est} and Appendix~\ref{app:outcome_est}; in particular, for the convergence results in Theorems~\ref{thm:fx_conv}-\ref{thm:out_conv}.
\item[Appendix \ref{app:technical}] provides a technical result for estimating the maximum of a finite number of functionals, which is used in the proofs in Appendix~\ref{app:est}.
\item[Appendix \ref{app:proofs-robustness}] provides proofs for the results in Appendix~\ref{app:robustness}.
\end{enumerate}
\color{black}

\section{Invariance of the bounds and intuition for when standard approaches under- or over-estimate robustness to unmeasured confounding} \label{app:invariance_intuition}

This section builds further intuition for understanding the behavior of bounds and estimators for them with calibrated sensitivity models.

\subsection{Properties of the bounds} \label{sec:bounds_properties}

We begin with a simple observation, that the bounds in a sensitivity model and calibrated sensitivity model will be the same when $\gamma = \Gamma M$. 
\begin{proposition} \label{prop:invariance}
\textbf{(Invariance)} Suppose there is a sensitivity model that imposes a one-dimensional bound of the form $U \leq \gamma \in (0, \infty)$, where $U$ is some quantification of unmeasured confounding, which implies a deterministic upper bound on the causal effect of interest, $u(\gamma): (0, \infty) \to \bbR$.  Moreover, suppose there is a related calibrated sensitivity model of the form $U \leq \Gamma M$, where $U$ is the same as before, $\Gamma \in (0, \infty)$, and $0 < M < \infty$ (as in Assumption~\ref{asmp:bounded}) is some quantification of measured confounding. The calibrated sensitivity model then implies a deterministic bound $\mathcal{U}(\Gamma) = u(\Gamma M)$ and, therefore,
\begin{equation} \label{eq:invariance}
	\Gamma M = \gamma \implies \mathcal{U}(\Gamma) = u(\gamma).
\end{equation}
\end{proposition}

The relationship in \eqref{eq:invariance} is true under the weak condition that the bound $u(\gamma)$ is a deterministic function, so it maps the same input to the same output. It also requires that measured confounding is non-zero and bounded, so that the calibrated sensitivity bound is non-zero and non-infinite. This result holds for the effect differences and odds ratio models. For the outcome model, which implicitly imposes two bounds --- for $a \in \{0,1\}$ --- the result in \eqref{eq:invariance} does not hold, but a similar two-dimensional property applies.  

\smallskip

The observation in Proposition~\ref{prop:invariance} is important, despite its simplicity, because it illustrates how calibrated sensitivity analyses can be compared with standard methods. Standard approaches compare $\gamma$ values to estimates $\widehat M$ in terms of their relative magnitude. Thus, standard approaches implicitly consider the ratio $\gamma/M$, and \eqref{eq:invariance} confirms that if $\frac{\gamma}{M} = \Gamma$ then $\left\{ \frac{\gamma}{M}, u(\gamma) \right\} = \{ \Gamma, \mathcal{U}(\Gamma) \}$. In other words, standard sensitivity models with post hoc calibration identify the same tuple as calibrated sensitivity models.  Therefore, we can compare uncertainty quantification in each approach by dividing $\gamma$ by $\widehat M$, and comparing the confidence intervals for the second elements in $\left\{ \gamma / \widehat M, \widehat{u}(\gamma) \right\}$ and $\{ \Gamma, \widehat{\mathcal{U}}(\Gamma) \}$. This comparison enables understanding when standard methods over- or under-estimate robustness to unmeasured confounding. 

\subsection{Understanding the change in robustness to unmeasured confounding} \label{sec:intuition}

For simplicity, we focus on the effect differences model, although the conclusions generalize immediately to other models with bounds that are linear in the sensitivity parameter, such as the outcome model, and may generalize to other models.

\smallskip

With the effect differences sensitivity model, the upper bound takes the form $u(\gamma) = \psi + \gamma$, where $\psi$ is some functional (in this paper, the adjusted mean difference) that depends on the data.  The calibrated upper bound is $\mathcal{U}(\Gamma) = \psi + \Gamma M.$ Inspired by the invariance property in Proposition~\ref{prop:invariance}, we compare confidence intervals of stylized estimators that point estimate the same tuple.  Suppose $\gamma = \Gamma M$. We compare confidence intervals for the second element in the following tuples:
\begin{enumerate}
\item Sensitivity model: $\left\{ \gamma / \widehat M, \widehat{u}(\gamma) \right\}$ 
\item Calibrated sensitivity model: $\left\{ \Gamma, \widehat{\mathcal{U}}(\Gamma) \right\}$
\end{enumerate}
Assuming appropriately accurate estimators, both approaches estimate the same tuple because $\gamma = \Gamma M$.  The difference arises in how confidence intervals are constructed for the second element --- the calibrated sensitivity approach accounts for uncertainty in estimating $\widehat M$, while the standard post hoc calibration approach does not.

\smallskip

Suppose access to estimators such that Wald-type confidence intervals can be constructed using the variance of each estimator:
\begin{align}
\bbV \left\{ \widehat{u}(\gamma) \right\} &= \bbV(\widehat \psi) \text{ and } \label{eq:standard} \\
\bbV\left\{ \widehat{\mathcal{U}}(\Gamma) \right\} &= \bbV(\widehat \psi + \Gamma \widehat M) = \bbV \left(\widehat \psi + \gamma \frac{\widehat M}{M} \right) \equiv   \bbV \left\{  \widehat{u}(\gamma) + \gamma \frac{\widehat M}{M} \right\}  \label{eq:new}
\end{align}
We can compare confidence intervals by comparing the sizes of the variances in \eqref{eq:standard} and \eqref{eq:new}. If the variance of $\widehat u(\gamma)$ is larger than the variance of $\widehat{\mathcal{U}}(\Gamma)$, it implies that standard methods are under-estimating robustness to unmeasured confounding, in the sense that a seemingly smaller amount of unmeasured confounding would be needed to nullify a significant causal effect compared to what one would conclude from a calibrated sensitivity analysis that incorporates uncertainty in $\widehat{M}$. Similarly, if the variance of $\widehat{u}(\gamma)$ is smaller, it implies that standard methods are over-estimating robustness to unmeasured confounding. Notice that
\begin{equation} 
\frac{\bbV \{ \widehat{\mathcal{U}}(\Gamma) \}}{\bbV \{ \widehat{u}(\gamma) \}} = 1 + \frac{\bbV(\widehat M) / M^2}{\bbV \{ \widehat{u}(\gamma) \} / \gamma^2} + 2\rho \left( \widehat{u}(\gamma), \widehat M \right) \sqrt{\frac{\bbV(\widehat M) / M^2}{\bbV \{ \widehat{u}(\gamma) \} / \gamma^2 }}, \label{eq:identity}
\end{equation}
where $\rho(X, Y) = \frac{\cov(X, Y)}{\sqrt{\bbV(X)\bbV(Y)}}$ is the correlation between $X$ and $Y$. The other key quantity on the right-hand side is $\frac{\sqrt{\bbV(\widehat M)} / M}{\sqrt{\bbV\{ \widehat u(\gamma) \}} / \gamma }$. The numerator is the relative standard error of $\widehat M$, while the denominator is a similar relative error quantity, but standardizes by $\gamma$ instead of $u(\gamma)$. We will refer to this as the ratio of relative standard errors (RRSE) of the estimators for measured confounding and the upper bound, $\widehat M$ and $\widehat u(\gamma)$.  Rearranging \eqref{eq:identity}, we can characterize when standard methods over- or under-estimate robustness to unmeasured confounding by the following inequalities:
\begin{align} 
\frac{1}{2} \frac{\sqrt{\bbV(\widehat M)} / M}{\sqrt{\bbV \{ \widehat{u}(\gamma) \}} / \gamma } < -\rho \{ \widehat{u}(\gamma), \widehat M \} \implies \bbV \{ \widehat{\mathcal{U}}(\Gamma) \} < \bbV \{ \widehat{u}(\gamma) \} \nonumber \\
\frac{1}{2} \frac{\sqrt{\bbV(\widehat M)} / M}{\sqrt{\bbV \{ \widehat{u}(\gamma) \}} / \gamma }  > -\rho \{ \widehat{u}(\gamma), \widehat M \} \implies \bbV \{ \widehat{\mathcal{U}}(\Gamma) \} > \bbV \{ \widehat{u}(\gamma) \} \label{eq:relationship}	
\end{align}

\begin{figure}[h]
\centering
\includegraphics[width = 0.9\textwidth]{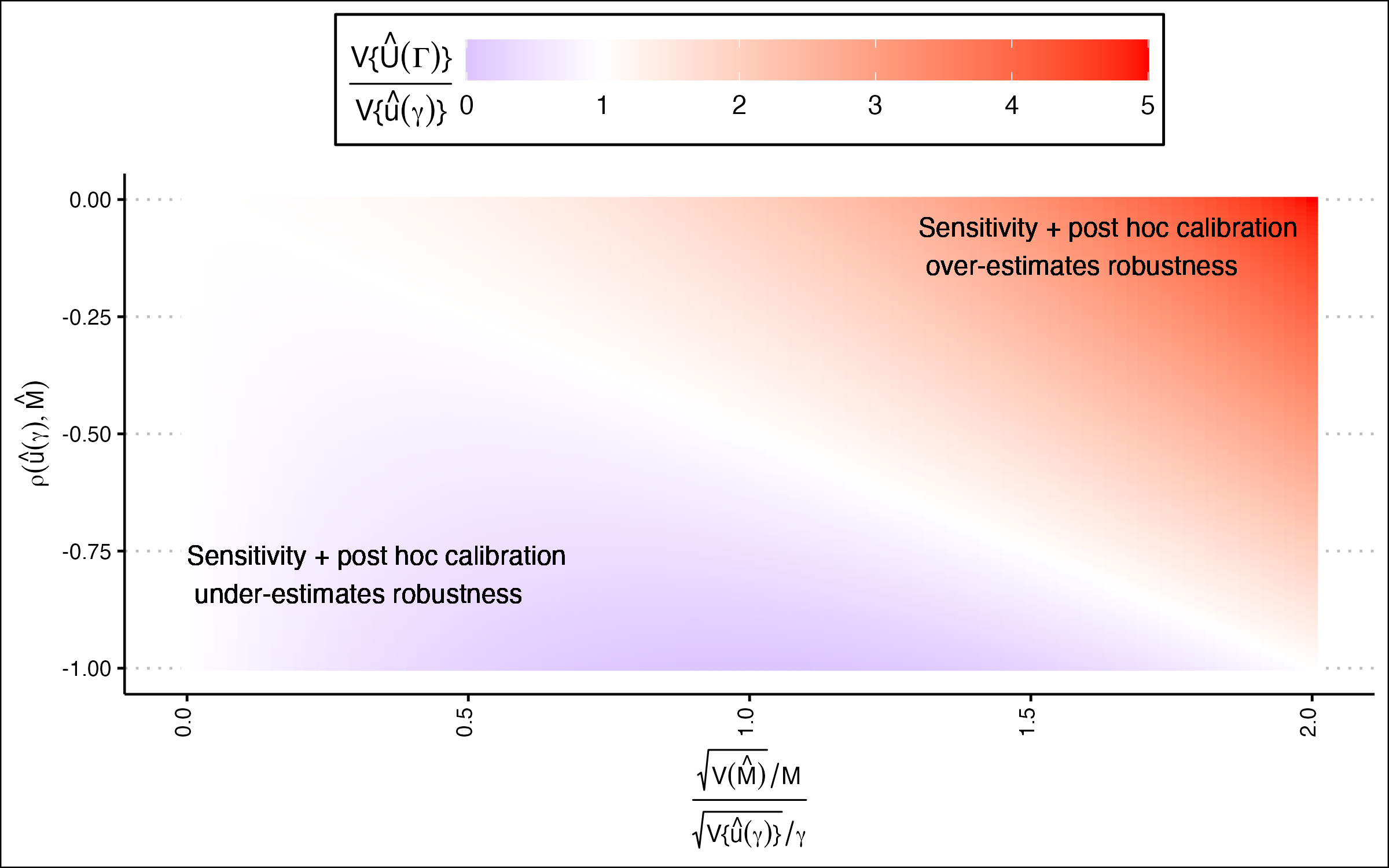}
\caption{Visualization of \eqref{eq:relationship}, showing when sensitivity analyses with post hoc calibration over- or under-estimate robustness to unmeasured confounding, compared to calibrated sensitivity models.}
\label{fig:intuition}
\end{figure}

The relationships in \eqref{eq:relationship} can also be visualized, as in Figure~\ref{fig:intuition}. In the red area in the top-right of Figure~\ref{fig:intuition}, standard approaches over-estimate robustness to unmeasured confounding.  In the blue area in the bottom-left, standard approaches under-estimate robustness to unmeasured confounding. Along the white diagonal the approaches yield the same conclusions.  This plot does not characterize the whole space because it is red for $> 2$ on the x-axis and $> 0$ on the y-axis (i.e., for most of the space, standard approaches over-estimate robustness to unmeasured confounding). 

\smallskip

\noindent Figure~\ref{fig:intuition} and \eqref{eq:relationship} can be summarized as follows:
\begin{itemize}
	\item \textbf{Standard sensitivity analysis with post hoc calibration overestimates robustness to confounding} when the estimators for the upper bound in the sensitivity model and for measured confounding are positively correlated. Alternatively, if the estimates are negatively correlated, this occurs when the RRSE is more than twice the absolute correlation---when the estimator for measured confounding is relatively accurate compared to the estimator for the upper bound.
	\item \textbf{Standard methods underestimate robustness to confounding} only when the estimators for measured confounding and the upper bound are negatively correlated and the RRSE is less than twice the absolute correlation---when the estimator for measured confounding is relatively inaccurate compared to the estimator for the upper bound.
\end{itemize}
This analysis provides intuition for when standard methods might over- or underestimate robustness to unmeasured confounding by failing to appropriately account for uncertainty in estimates of measured confounding.

\section{Estimating a robustness value with a calibrated sensitivity model} \label{app:robustness}

As alluded to in the main paper, researchers might be interested in estimating a one-number summary of the robustness of the calibrated sensitivity analysis. In this section, we focus on the level of $\Gamma$ at which bounds for the causal effect include zero. That is:
\begin{equation}
	\Gamma_0 := \inf_{\Gamma \in [0, \infty)} \left\{ \Gamma: \text{sign} \{ \mathcal{L}(\Gamma) \} \neq \text{sign} \{ \mathcal{U}(\Gamma)\} \right\}.
\end{equation}
This communicates the minimum level of $\Gamma$ at which the upper or lower bound on the causal effect crosses zero. In the next two sections we establish general conditions under which $\Gamma_0$ exists and when a Z-estimator satisfies a Gaussian limiting distribution. In the final section, we revisit the maximum leave-one-out effect differences model from the main paper and show how these results apply.

\subsection{Existence of $\Gamma_0$}

\noindent The existence of $\Gamma_0$ is guaranteed under mild assumptions. 
\begin{assumption} \label{asmp:existence}	
	$\mathcal{U}(0) = \mathcal{L}(0) = \psi_\ast$.  Moreover, there exists $\mathcal{G} = [0, G)$ with $G < \infty$, on which $\mathcal{U}(\cdot)$ and $\mathcal{L}(\cdot)$ are strictly monotone (increasing and decreasing, respectively), continuous, and bounded. Additionally, there exists $\Gamma \in \mathcal{G}$ such that $\mathcal{U}(\Gamma) = 0$ when $\psi_\ast < 0$, or $\mathcal{L}(\Gamma) = 0$ when $\psi_\ast > 0$.	 
\end{assumption}
Assumption~\ref{asmp:existence} asserts that the bounds equal the true causal effect when the sensitivity model assumes no unmeasured confounding. It also ensures that the bounds are monotone, continuous, and bounded for $\Gamma \in \mathcal{G} \subset [0, \infty)$, and that one of the bounds intersects zero for some $\Gamma \in \mathcal{G}$. When combined with Assumption~\ref{asmp:bounded}, it is sufficient to define $\Gamma_0$ as the solution to a population moment condition.

\begin{proposition} \label{prop:gamma_zero_def}
    Under Assumptions~\ref{asmp:bounded} and~\ref{asmp:existence}, there exists a unique $\Gamma_0 \in \mathcal{G}$ such that $\Gamma_0$ is the solution to $\mathcal{L}(\Gamma_0) \mathcal{U}(\Gamma_0) = 0$.
\end{proposition}
The proofs are delayed to the Appendix~\ref{app:proofs-robustness} for brevity.

\subsection{Estimation}

In this subsection, we construct an estimator for $\Gamma_0$, based on estimators for the bounds in a calibrated sensitivity model.  We use the following notation: Let $\Psi(\cdot) := \mathcal{L}(\cdot) \mathcal{U}(\cdot)$, and assuming they exist, denote its derivative as $\Psi^\prime(\cdot)$ and the inverse of its derivative as $\Psi^\prime(\cdot)^{-1}$. Meanwhile, let $\widehat \Psi(\Gamma) = \widehat{\mathcal{L}}(\Gamma) \widehat{\mathcal{U}}(\Gamma)$ be an estimator for $\Psi(\Gamma)$, where $\{ \widehat{\mathcal{L}}, \widehat{\mathcal{U}} \}$ are estimators for $\{\mathcal{L}, \mathcal{U}\}$. For a function $f(\Gamma)$, define the norm $\left\| f \right\|_{\mathcal{G}} = \sup_{\Gamma \in \mathcal{G}} |f(\Gamma)|.$

\smallskip

Proposition~\ref{prop:gamma_zero_def} ensures that $\Psi(\Gamma_0) = 0$. This naturally suggests a Z-estimator $\widehat \Gamma_0$, which solves the moment condition $\widehat \Psi(\widehat \Gamma_0) = o_{\bbP}(n^{-1/2}).$ To establish the limiting distribution of $\widehat \Gamma_0$ converging to $\Gamma_0$, we require the following assumptions. 

\begin{assumption} \label{asmp:est-unif-conv}
    We have access to estimators $\widehat{\mathcal{L}}$ and $\widehat{\mathcal{U}}$ for the lower and upper bounds that satisfy uniform convergence guarantees across $\mathcal{G}$:
    $$
    \left\| \sqrt{n} \left( \widehat{\mathcal{L}} - \mathcal{L} \right) - \sqrt{n} (\bbP_n - \bbE) \varphi_l \right\|_{\mathcal{G}} = o_{\bbP}(1) \quad \text{and} \quad \left\| \sqrt{n} \left( \widehat{\mathcal{U}} - \mathcal{U} \right) - \sqrt{n} (\bbP_n - \bbE) \varphi_u \right\|_{\mathcal{G}} = o_{\bbP}(1)
    $$
    for $\varphi_l(Z; \Gamma), \varphi_u(Z; \Gamma)$ with bounded mean and non-zero and bounded variance. Additionally, $\left\| \widehat{\mathcal{L}} \right\|_{\mathcal{G}} = O_{\bbP}(1)$, $\left\| \widehat{\mathcal{U}} \right\|_{\mathcal{G}} = O_{\bbP}(1)$.
\end{assumption}

\begin{assumption} \label{asmp:g0-diff} 
    $\Psi(\cdot)$ has a continuously invertible non-zero derivative at $\Gamma_0$.
\end{assumption}

\begin{assumption} \label{asmp:donsker}
    The function classes $\{ \varphi_u(\cdot; \Gamma): \Gamma \in \mathcal{G} \}$ and $\{ \varphi_l(\cdot; \Gamma) : \Gamma \in \mathcal{G} \}$ are Donsker.
\end{assumption} 

\begin{assumption} \label{asmp:smooth-ifs}
    The functions $\varphi_u$ and $\varphi_l$ are smooth at $\Gamma_0$ in the sense that $\bbE \{ \varphi_u(Z; \Gamma) - \varphi_u(Z; \Gamma_0) \}^2 \to 0$ and $\bbE \{ \varphi_l(Z; \Gamma) - \varphi_l(Z; \Gamma_0) \}^2 \to 0$ as $\Gamma \to \Gamma_0$.
\end{assumption}

While Assumption~\ref{asmp:est-unif-conv} provides strong convergence guarantees for $\widehat \Psi$, it is not sufficient by itself to establish a limiting distribution for $\widehat \Gamma_0$. Assumptions~\ref{asmp:g0-diff}-\ref{asmp:smooth-ifs} impose smoothness conditions to establish the normal limiting distribution of $\widehat \Gamma_0$. Assumption~\ref{asmp:g0-diff} requires that the product of the upper and lower bounds has a continuously invertible derivative at $\Gamma_0$. Assumptions~\ref{asmp:donsker} and \ref{asmp:smooth-ifs} ensure that the influence functions of $\widehat{\mathcal{L}}$ and $\widehat{\mathcal{U}}$ are Donsker and appropriately smooth in the calibrated sensitivity parameter. Donsker classes are a broad category of functions often encountered in semiparametric efficiency theory. For further details, see \citet{van1996weak}. Assumption~\ref{asmp:smooth-ifs} ensures that $\varphi_u(Z; \Gamma)$ and $\varphi_l(Z; \Gamma)$ converge to their respective values at $\Gamma_0$ in squared mean as $\Gamma$ approaches $\Gamma_0$. 

\smallskip

Together with the conditions of Proposition~\ref{prop:gamma_zero_def}, Assumptions~\ref{asmp:est-unif-conv}-\ref{asmp:smooth-ifs} are sufficient to derive the limiting distribution for $\widehat \Gamma_0$.
\begin{theorem} \label{thm:gamma_zero}
    Let $\varphi_\Psi(Z; \Gamma) = \varphi_u(Z; \Gamma) \mathcal{L}(\Gamma) + \varphi_l(Z; \Gamma) \mathcal{U}(\Gamma)$. Suppose the conditions of Proposition~\ref{prop:gamma_zero_def} and Assumptions~\ref{asmp:est-unif-conv}--\ref{asmp:smooth-ifs} hold and $\widehat \Gamma_0$ solves $\widehat \Psi(\widehat \Gamma_0) = o_{\bbP}(n^{-1/2})$. Then,
    $$
    \sqrt{n}( \widehat \Gamma_0 - \Gamma_0 ) \indist  N \left( 0, \bbV\left\{ \Psi^\prime(\Gamma_0)^{-1} \varphi_\Psi(Z; \Gamma_0 ) \right\} \right).
    $$
\end{theorem}
Theorem~\ref{thm:gamma_zero} shows that $\widehat \Gamma_0$ behaves like a sample average centered at $\Gamma_0$, with asymptotically negligible error. Thus, it demonstrates that inference is feasible using a Wald-type interval
$$
\widehat \Gamma_0 \pm z_{\alpha/2} \sqrt{ \bbV\left\{ \Psi^\prime(\Gamma_0)^{-1} \varphi_\Psi(Z; \Gamma_0 ) \right\} / n }.
$$
The result follows from standard Z-estimator analysis, particularly \citet[Theorem 3.3.1]{van1996weak}. \citet{bonvini2022sensitivitya} proved a specific version of this result for their sensitivity model. 

\smallskip

One can also gain further intuition for the limiting distribution expression in Theorem~\ref{thm:gamma_zero}, in the following result.
\begin{corollary} \label{cor:gamma_zero}
	Let $\mathcal{L}^\prime(\cdot)$ and $\mathcal{U}^\prime(\cdot)$ denote the derivatives of $\mathcal{L}(\cdot)$ and $\mathcal{U}(\cdot)$, respectively, and suppose the conditions of Theorem~\ref{thm:gamma_zero} hold.  If $\mathcal{U}(\Gamma_0) = 0$, then
	$$
	\sqrt{n} ( \widehat \Gamma_0 - \Gamma_0 ) \indist N \left( 0, \mathcal{U}^\prime(\Gamma_0)^{-2} \bbV\left\{ \varphi_u(Z; \Gamma_0) \right\} \right).
	$$
	If $\mathcal{L}(\Gamma_0) = 0$, then
	$$
	\sqrt{n} ( \widehat \Gamma_0 - \Gamma_0 ) \indist N \left( 0, \mathcal{L}^\prime(\Gamma_0)^{-2} \bbV\left\{ \varphi_l(Z; \Gamma_0) \right\} \right).
	$$
\end{corollary}
In other words, the limiting distribution of $\widehat \Gamma_0$ depends only on the limiting distribution of the estimator for the bound that crosses zero, scaled by the derivative of that bound at $\Gamma_0$.

\subsection{Application to the effect differences model} \label{app:effect}

The results from the previous sections can be applied to the maximum leave-one-out effect differences model from the main paper. In that case, the bounds are 
\[
\left\{ \mathcal{L}(\Gamma), \mathcal{U}(\Gamma) \right\} = \left\{ \psi - \Gamma \max_{j=1}^{d} \left| \psi - \psi_{-j} \right|, \psi + \Gamma \max_{j=1}^{d} \left| \psi - \psi_{-j} \right|\right\}
\]
and their derivatives with respect to $\Gamma$ are 
\[
\left\{ \mathcal{L}^\prime(\Gamma), \mathcal{U}^\prime(\Gamma) \right\} = \left\{ - \max_{j=1}^{d} \left| \psi - \psi_{-j} \right|,  \max_{j=1}^{d} \left| \psi - \psi_{-j} \right|\right\}.
\]
Under Assumption~\ref{asmp:bounded}, Assumption~\ref{asmp:existence} then immediately follows and $\Gamma_0$ is the unique solution to the moment condition
\[
\Psi(\Gamma_0) = 0 \text{ where } \Psi(\Gamma) = \left( \psi - \Gamma \max_{j=1}^{d} \left| \psi - \psi_{-j} \right| \right) \left( \psi + \Gamma \max_{j=1}^{d} \left| \psi - \psi_{-j} \right| \right).
\]
The bounds are linear in $\Gamma$ and therefore are smooth and satisfy Assumptions~\ref{asmp:g0-diff}-\ref{asmp:smooth-ifs}. It only remains to establish that the estimators $\widehat{\mathcal{L}}$ and $\widehat{\mathcal{U}}$ satisfy Assumption~\ref{asmp:est-unif-conv}. These estimators can be constructed using the estimators from Theorem~\ref{thm:fx_conv}. Under the conditions of Theorem~\ref{thm:fx_conv}, the estimators $\widehat{\mathcal{L}}$ and $\widehat{\mathcal{U}}$ satisfy uniform convergence guarantees across $\mathcal{G}$, as in Assumption~\ref{asmp:est-unif-conv}.  Therefore, by Theorem~\ref{thm:gamma_zero}, we can construct an estimator $\widehat \Gamma_0$ for the minimum level of $\Gamma$ at which one of the bounds crosses zero, and it satisfies a Gaussian limiting distribution. As in Corollary~\ref{cor:gamma_zero}, we'll consider two cases, depending on whether $\mathcal{U}(\Gamma_0) = 0$ or $\mathcal{L}(\Gamma_0) = 0$.  When $\mathcal{U}(\Gamma_0) = 0$, the limiting distribution simplifies to
\begin{align*}
	\sqrt{n} ( \widehat \Gamma_0 - \Gamma_0 ) &\indist N \left(0, \left( \max_{j=1}^{d} \left| \psi - \psi_{-j} \right| \right)^{-2} \bbV\left\{ \varphi_u(Z; \Gamma_0) \right\} \right) \text{where} \\
	\varphi_u(Z; \Gamma) &= \phi(Z) +  \text{sign} (\psi- \psi_{-j^\prime}) \Gamma \big\{ \phi(Z) - \phi(Z_{-j^\prime}) \big\}.
\end{align*}
Meanwhile, when $\mathcal{L}(\Gamma_0) = 0$, the limiting distribution simplifies to
\begin{align*}
	\sqrt{n} ( \widehat \Gamma_0 - \Gamma_0 ) &\indist N \left(0, \left( \max_{j=1}^{d} \left| \psi - \psi_{-j} \right| \right)^{-2} \bbV\left\{ \varphi_l(Z; \Gamma_0) \right\} \right) \text{ where } \\
	\varphi_l(Z; \Gamma) &= \phi(Z) -  \text{sign} (\psi- \psi_{-j^\prime}) \Gamma \big\{ \phi(Z) - \phi(Z_{-j^\prime}) \big\}.
\end{align*}
Because of the simple construction of the bounds, the limiting distribution of the estimator for $\Gamma_0$ is scaled by the estimate of measured confounding directly.

\section{Estimation and inference with the outcome model} \label{app:outcome_est}

Here, we consider estimating the upper bound on the ATE in the outcome model (calibrated sensitivity model~\ref{csm:avg_lso_out}).  These are complementary results to those for the effect differences and odds ratio models in Section~\ref{sec:est}.
\begin{definition} \label{def:outcome} 
Let $\mathcal{U}(\Gamma)$ be as in \eqref{eq:out_id}. Construct an estimator for the upper bound as
\begin{equation} \label{eq:bound_estimator}
	\widehat{\mathcal{U}}(\Gamma) := \bbP_n \{ \widehat \phi(Z) \} + \Gamma \sum_{a \in \{0,1\}} \sqrt{\bbP_n \{ \widehat \xi_{1-a}(X) \}} \sqrt{\frac{1}{|\mathcal{S}|} \sum_{S\in \mathcal{S}} \bbP_n \{ \widehat \lambda_a(Z; S) \}},
\end{equation}
where 
\begin{itemize}
	\item 	where $\phi$ is the EIF of the adjusted mean difference, defined in \eqref{eq:mde_if},
	\item $\xi_a$ is the EIF of $\lVert \pi_{a}(X) \rVert_2^2$, defined in \eqref{eq:xi} in the supplementary materials, 
	\item $\lambda_a(Z; S)$ is the EIF of $\lVert \mu_a(X_{-S}) - \bbE \big\{ \mu_a(X) \mid A = 1-a, X_{-S} \} \rVert_2^2$, defined in \eqref{eq:lambda} in the supplementary materials, and 
	\item the estimated nuisance functions constituting $\widehat \phi, \widehat \xi$, and $\widehat \lambda$ are constructed on a separate sample.
\end{itemize}
\end{definition}
The estimator in \eqref{eq:bound_estimator} has the same structure as the estimator in \eqref{eq:fx_estimator} for the effect differences model --- the first term comes from estimating the adjusted mean difference while the other terms come from estimating the bound on $\left|\psi_\ast - \psi \right|$ implied by the model, i.e., the second summand in the partial identification result in \eqref{eq:out_id} in Proposition~\ref{prop:partial_id}. The form of the other terms follows from estimating the two functionals in the bound, $\lVert \pi_{1-a}(X) \rVert_2$ and $ \sqrt{\frac{1}{|\mathcal{S}|} \sum_{S \in \mathcal{S}} \left\lVert \mu_a(X_{-S}) - \bbE \big\{ \mu_a(X) \mid A = 1-a, X_{-S} \} \right\rVert_2^2}$.  

\bigskip

A crucial step in constructing the estimator for measured confounding is to construct an estimator for $\bbE \{ \mu_a(X) \mid A=1-a, X_{-S} \}$, which appears in the definition of the bound in \eqref{eq:out_id}. To construct an estimator for $\bbE \{ \mu_a(X) \mid A = 1-a, X_{-S}\}$, one could use a two-stage estimator where the EIF of $\bbE\{ \mu_a(X)\}$ is estimated and regressed against $A = 1-a, X_{-S}$ as a pseudo-outcome in a second-stage regression \citep{kennedy2023towards, mcclean2024nonparametric, fisher2023three, semenova2021debiased, diaz2018targeted}. Therefore, in the next result, we let $\widehat{\bbE}_n \{ \widehat \mu_a(X) \mid A=1-a, X_{-S} \}$ denote a generic regression estimator that regresses estimated EIF values of $\bbE\{ \mu_a(X) \}$ against $1-a$ and $X_{-S}$. Linear smoothers as second-stage regressions have been studied in this context, and could achieve $n^{-1/4}$ rates when $\bbE \{ \mu_a(X) \mid A = 1-a, X_{-S} \}$ is appropriately smooth as a function of $X_{-S}$ \citep{kennedy2023towards}.  Some analyses of two-stage estimators require further sample splitting, so that the nuisance functions constituting the EIF of $\bbE \{ \mu_a(X)\}$ are estimated on a separate sample from the second-stage regression. We will ignore that complication here, for ease of exposition.

\begin{restatable}{theorem}{thmoutconv} \label{thm:out_conv} \textbf{\emph{(Average LSO outcome model)}} 
Let $\mathcal{U}(\Gamma)$ and $\widehat{\mathcal{U}}(\Gamma)$ be as in Definition~\ref{def:outcome}.  Suppose Assumptions~\ref{asmp:positivity}-\ref{asmp:bounded} hold, and
\begin{enumerate}
	\item for $S \in \mathcal{S} \cup \emptyset$, $\widehat \pi_1(X_{-S})$ is bounded away from zero and one, \label{cond:out_one}
	\item $\widehat \phi$ is consistent in the sense that $\lVert \widehat \phi(Z) - \phi(Z) \rVert_2 \inprob 0$, \label{cond:out_two}
	\item for $a \in \{0,1\}$, $\widehat \xi_a$ is consistent in the sense that $\lVert \widehat \xi_a(X) - \xi_a(X) \rVert_2 \inprob 0$, \label{cond:out_three}
	\item for $a \in \{0,1\}$ and $S \in \mathcal{S}$, $\widehat \lambda_a(Z; S)$ is consistent in the sense that $\lVert \widehat \lambda_a(Z; S) - \lambda_a(Z; S) \rVert_2 \inprob 0$, \label{cond:out_four}
	\item for $a \in \{0,1\}$, $\lVert \widehat \pi_1(X) - \pi_1(X) \rVert_2 \lVert \mu_a(X) - \mu_a(X) \rVert_2 = o_\bbP(n^{-1/2})$, \label{cond:out_five}
	\item $\lVert \widehat \pi_1(X) - \pi_1(X) \rVert_2 = o_{\bbP}(n^{-1/4})$, and \label{cond:out_six}
	\item for $a \in \{0,1\}$ and $S \in \mathcal{S}$, \label{cond:out_seven}
	\begin{align}
		\hspace{-1in}o_\bbP(n^{-1/2}) &= \left( \left\lVert \widehat \mu_a(X_{-S}) - \mu_a(X_{-S}) \right\rVert_2 + \left\lVert \widehat \bbE_n \{ \widehat \mu_a(X) \mid A =1-a, X_{-S} \} - \bbE\{ \mu_a(X) \mid A =1-a, X_{-S} \} \right\rVert_2 \right)^2 \nonumber \\
		&+ \lVert \widehat \pi_{1-a}(X_{-S}) - \pi_{1-a}(X_{-S}) \rVert_2 \lVert \widehat \mu_a(X_{-S}) - \mu_a(X_{-S}) \rVert_{2} \nonumber \\
		&+\lVert \widehat \pi_{1-a}(X_{-S}) - \pi_{1-a}(X_{-S}) \rVert_2 \left\lVert \bbE \{ \widehat \mu_a(X) \mid A = 1-a, X_{-S} \} - \widehat \bbE_n  \{ \widehat \mu_a(X) \mid A = 1-a, X_{-S} \} \right\rVert_2 \nonumber  \\
		&+ \lVert \widehat \mu_a(X) - \mu_a(X) \rVert_2  \left( \left\lVert \widehat \pi_{1-a}(X) - \pi_{1-a}(X) \right\rVert_2 + \left\lVert \widehat \pi_{1-a}(X_{-S}) - \pi_{1-a}(X_{-S}) \right\rVert_2 \right). \label{eq:dr_cond2}
	\end{align} 
\end{enumerate}
Then,
\begin{align}
	\hspace{-0.5in} \widehat{\mathcal{U}}(\Gamma) - \mathcal{U}(\Gamma) = (\bbP_n - \bbE) \Bigg( \phi(Z) + \Gamma \Bigg[ \sum_{a \in \{0,1\}} &\left\lVert \pi_{1-a}(X) \right\rVert_2 \left\{ \frac{1}{|\mathcal{S}|} \sum_{S\in \mathcal{S}} \frac{\lambda_a(Z; S)}{2M_a}\right\} \nonumber \\
	&+ M_a \left\{ \frac{\xi_{1-a}(X) }{2 \lVert \pi_{1-a}(X) \rVert_2} \right\} \Bigg]  \Bigg)  + o_\bbP(n^{-1/2}). \label{eq:outcome_ral}
\end{align}
where
$$
M_a = \sqrt{\frac{1}{|\mathcal{S}|} \sum_{S \in \mathcal{S}} \lVert \mu_a(X_{-S}) - \bbE \{ \mu_a(X) \mid A =1-a, X_{-S} \} \rVert_2^2}.
$$
\end{restatable}

Theorem~\ref{thm:out_conv} shows that the error of the estimator for the upper bound in the outcome model behaves like a centered sample average plus asymptotically negligible error under doubly robust conditions. Here, we give some intuition for the conditions.  Condition~\ref{cond:out_one} ensures that the bias of $\widehat{\mathcal{U}}(\Gamma)$ is bounded. Conditions~\ref{cond:out_two}, \ref{cond:out_three}, and \ref{cond:out_four} are weak consistency conditions for controlling the empirical process terms that arise in estimating the three components of the bound: $\psi$, $\lVert \pi_{a}(X) \rVert_2$, and $\sqrt{\frac{1}{|\mathcal{S}|} \sum_{S \in \mathcal{S}} \lVert \mu_a(X_{-S}) - \bbE \{ \mu_a(X) \mid A =1-a, X_{-S} \} \rVert_2^2}$.  These conditions would hold under weak nonparametric assumptions.  Condition~\ref{cond:out_five} is the usual doubly robust condition for controlling the conditional bias for estimating the adjusted mean difference $\psi$.  Conditions \ref{cond:out_six} and \ref{cond:out_seven} are stronger conditions, for controlling the other terms in the bias of $\widehat{\mathcal{U}}(\Gamma)$. Condition~\ref{cond:out_six} controls the conditional bias of the estimator for $\lVert \pi_{a}(X) \rVert_2$, while condition \ref{cond:out_seven} controls the conditional bias for estimating $\sqrt{\frac{1}{|\mathcal{S}|} \sum_{S \in \mathcal{S}} \lVert \mu_a(X_{-S}) - \bbE \{ \mu_a(X) \mid A =1-a, X_{-S} \} \rVert_2^2}$. Both are doubly-robust style conditions --- they consist of products of errors, and thus a sufficient condition for each to hold is that each nuisance function is estimated at only an $n^{-1/4}$ rate.

\begin{remark}
When the measured confounding is an average, as in the average LSO outcome model, $\sqrt{n}$ convergence guarantees on the nuisance function estimators are required across all covariate sets under consideration (in condition~\ref{cond:out_seven}).  This is necessary so that none of the biases from estimating each sub-quantification of measured confounding dominate asymptotically.  This is different from the maximum leave-one-out confounding, where $\sqrt{n}$ convergence guarantees are only required for the sub-quantification of measured confounding corresponding to the maximum (e.g., condition~\ref{cond:bias} in Theorem~\ref{thm:fx_conv}). In this sense, it is easier to estimate maximum measured confounding than average measured confounding, conditional on there being a unique maximum. However, the estimation procedures are the same whether measured confounding is a maximum or an average, in the sense that one would still aim to estimate the nuisance functions as accurately as possible with each covariate subset even when measured confounding is a maximum. This is because, a priori, one cannot know which covariate subset corresponds to the true maximum, and therefore it is still necessary to estimate all nuisance functions well in both cases.
\end{remark}

\section{Additional data analyses} \label{app:illustrations}

Here, we provide two additional data analyses, assessing the effect of exposure to violence on attitudes towards peace in Darfur.

\subsection{Data: peace attitudes in Darfur}

The dataset contains information about the attitudes towards peace of Darfurian villagers exposed to violence. In 2003 and 2004, the Darfurian government led a series of violent attacks against civilians, killing an estimated two hundred thousand people. Due to the quasi-randomness of the attacks, with this dataset one can attempt to answer the question of whether being directly injured or maimed in such attacks made people more or less likely to accept peace \citep{hazlett2020angry}. We use an example version of this dataset, which is publicly available in the \texttt{sensemakr} package in \texttt{R} \citep{cinelli2020sensemakr}. 

\bigskip

The data contains information on 1,276 Darfurian villagers. The treatment is a binary variable indicating whether the individual was physically injured during an attack, and the outcome is an index (continuous from 0 to 1) measuring pro-peace attitudes (with higher being more pro-peace). The covariates include the original village of the respondent (we grouped every village with under 10 respondents into ``Other''), gender, age (in years), whether they were a herder or a farmer, whether they voted in an earlier election before the conflict, and the size of their household.

\subsection{Methods}

We estimated bounds on the ATE using the maximum leave-one-out effect differences model (calibrated sensitivity model~\ref{csm:max_loo_fx}) and the maximum leave-one-out odds ratio model (calibrated sensitivity model~\ref{csm:max_loo_odds}).  For the effect differences model, we used the same methods as in Section~\ref{sec:illustrations_methods} of the body of the paper. Specifically, we constructed estimators for the bounds on the ATE according to Definition~\ref{def:effect}, with $\Gamma \in \{0.5, 1, \dots, 5\}$.  We constructed estimators for the adjusted mean difference with different covariate subsets using the \texttt{npcausal} package in \texttt{R} \citep{kennedy2023npcausal, rcore2021language}.  We used 5 splits and estimated the propensity score and outcome regression functions with the \texttt{SuperLearner}, stacking the sample average and a random forest from the \texttt{ranger} package with default tuning parameters \citep{van2007superlearner, wright2017ranger}. We constructed a pointwise confidence band for the upper bound on the ATE according to \eqref{eq:fx_upper_lim} and a pointwise confidence band for the lower bound analogously.  As described in Section~\ref{sec:inference}, we constructed one-sided 97.5\% confidence bands for the upper and lower bounds, and took their intersection to construct a 95\% confidence interval for the ATE. We also conducted a standard sensitivity analysis and post hoc calibration step, where we standardized the sensitivity parameter by estimated measured confounding.  

\bigskip

In the odds ratio model, we constructed estimators for the bounds on the ATE according to Definitions~\ref{def:odds} and \ref{def:odds_bound}, with $\exp(\Gamma) \in \{ 1.25, 1.5, 1.75, 2 \}$.  We estimated the best projection of measured confounding using a logistic regression with no interactions (see Definition~\ref{def:odds}). We also assumed that the observed maxima and minima in the covariate data were the true bounds of the covariate support, which allowed us to extend Theorem~\ref{thm:odds_conv} to a known and bounded support --- in this case, we estimated measured confounding as the maximum absolute value across covariates of each covariate's estimated coefficient multiplied by its range.  We then estimated the upper bound on the ATE using the algorithm in \citet{yadlowsky2022bounds}, with sieve estimators for the nuisance functions.  We constructed one-sided 97.5\% confidence bands for the upper and lower bounds and a 95\% confidence band for the ATE using Wald-type confidence intervals. We estimated the limiting variance with the nonparametric bootstrap, with $100$ resampling iterations of size $1,000$.  

\medskip

In both analyses, we also conducted a standard sensitivity analysis and post hoc calibration step, where we standardized the sensitivity parameter by estimated measured confounding.  

\subsection{Results}

Here, we report the results from each analysis.  The point estimate and confidence interval for the ATE assuming no unmeasured confounding are the same. We estimated a positive and significant ATE of $0.06$ (95\%: CI $[0.02, 0.1]$) which, interpreted causally, says that exposure to violence increased people's preference of peace by 0.06 on average (0.06 is a unit-less value --- the outcome is an index continuous on $[0,1]$).

\bigskip

With the effect differences model, we estimated that the respondent's original village was the most impactful confounder, and that the absolute change in the adjusted mean difference without it included as a covariate was $0.01$ with 95\% CI $[0.002, 0.02]$.  Meanwhile, with the odds ratio model, we again estimated that the respondent's original village was the most impactful confounder.  The maximum estimated confounding is $3.34$ (95\% CI: $[0,20]$) (this is a Wald-type confidence interval truncated below at $0$ where the variance was calculated across $100$ bootstrap resamples).  Note also that the respondent's village was the most impactful measured confounder in 99 out of the 100 bootstrap resamples.

\bigskip

Figure~\ref{fig:darfur} shows the ATE estimates, estimates for bounds on the ATE, and 95\% confidence bands for the ATE --- Figure~\ref{fig:darfur_fx} is for the effect differences model while Figure~\ref{fig:darfur_odds} is for the odds ratio model.  The x-axis is the level of the sensitivity parameter, $\Gamma$, in the calibrated sensitivity model, and the sensitivity parameter standardized by estimated measured confounding, $\gamma / \widehat M$, in the standard sensitivity model. The y-axis is at the level of the causal effect.  The horizontal dashed lines are the ATE estimates, and the horizontal dotted lines are at zero.  The dot-dash lines are the estimates of the upper and lower bounds on the ATE.  By the invariance property discussed in Proposition~\ref{prop:invariance} in Section~\ref{sec:partial_id}, the point estimates for these bounds are the same in the calibrated sensitivity model and sensitivity model.  However, the confidence intervals for the bounds differ.  The confidence intervals for the calibrated sensitivity model are solid, while the confidence intervals for the sensitivity model are long dashes.   

\bigskip

The calibrated sensitivity results in Figure~\ref{fig:darfur_fx} can be interpreted in the following way. The estimated bounds on the ATE do not intersect zero for unmeasured confounding less than five times measured confounding, and the confidence interval for the ATE intersects zero when $\Gamma \approx 2$, so the statistical significance of the estimated effect would only be nullified for unmeasured confounding at least twice as big as measured confounding. In this case, the measured confounding is the maximum change in the adjusted mean difference from leaving out one covariate. The most impactful observed confounder was the respondent's home village, and it changed the adjusted mean difference by 0.01.  The results in Figure~\ref{fig:darfur_odds} can be interpreted in a similar way.

\bigskip

The relationship between the calibrated sensitivity results and standard sensitivity results are similar in each analysis, in the sense that incorporating uncertainty due to estimating measured confounding demonstrates that the causal effect estimate is less robust to unmeasured confounding (i.e., the solid confidence intervals encompass the dashed confidence intervals).  However, with both the standard and calibrated sensitivity models, the odds ratio model suggests the estimates are far less robust to unmeasured confounding than what is suggested by the effect differences model.

\begin{figure}[h]
\centering
\begin{subfigure}{.49\textwidth} 
	\centering
	\includegraphics[width=\linewidth]{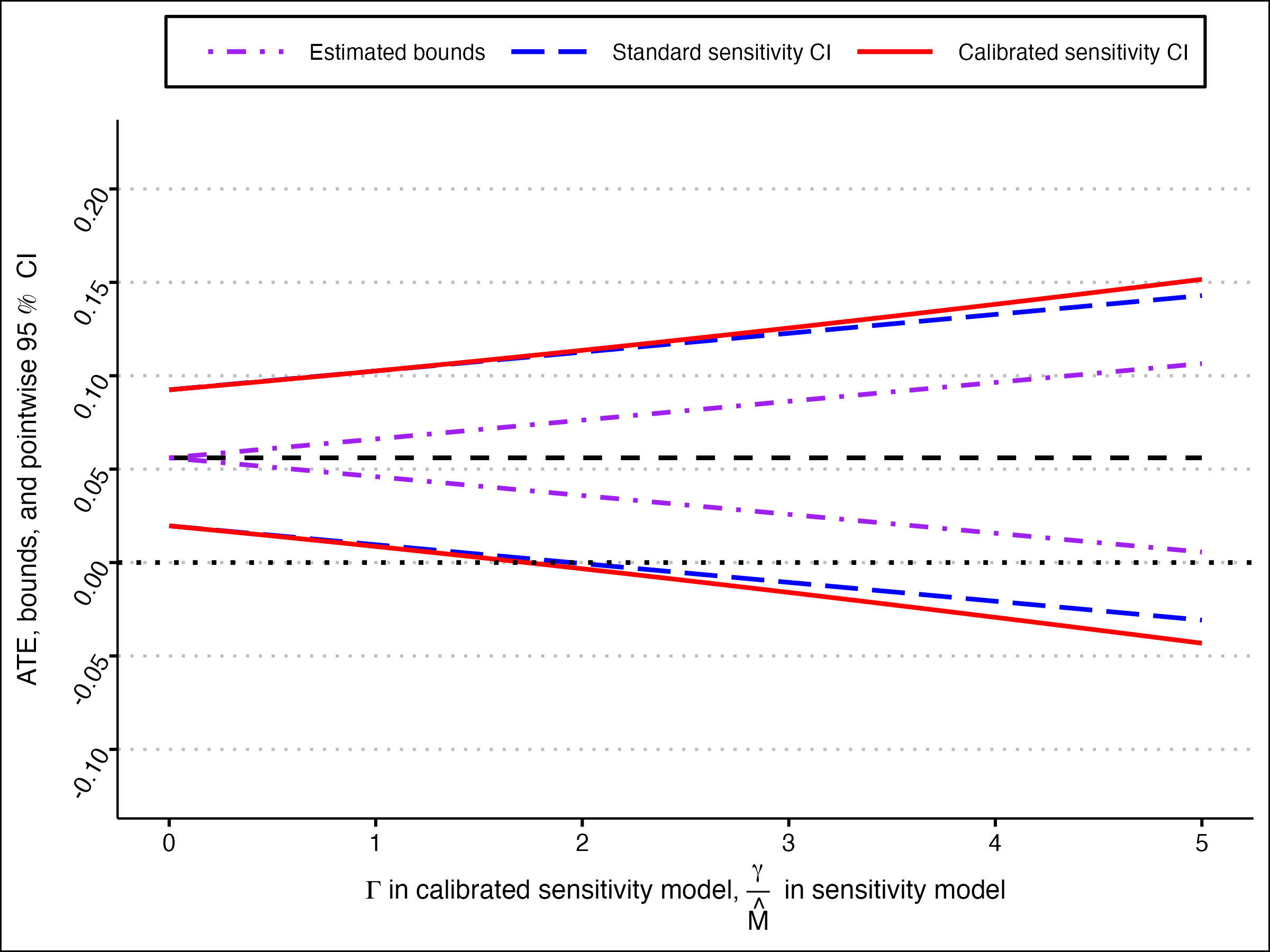} 
	\caption{Effect differences model.}
	\label{fig:darfur_fx}
\end{subfigure} \hfill
\begin{subfigure}{.49\textwidth} 
	\centering
	\includegraphics[width=\linewidth]{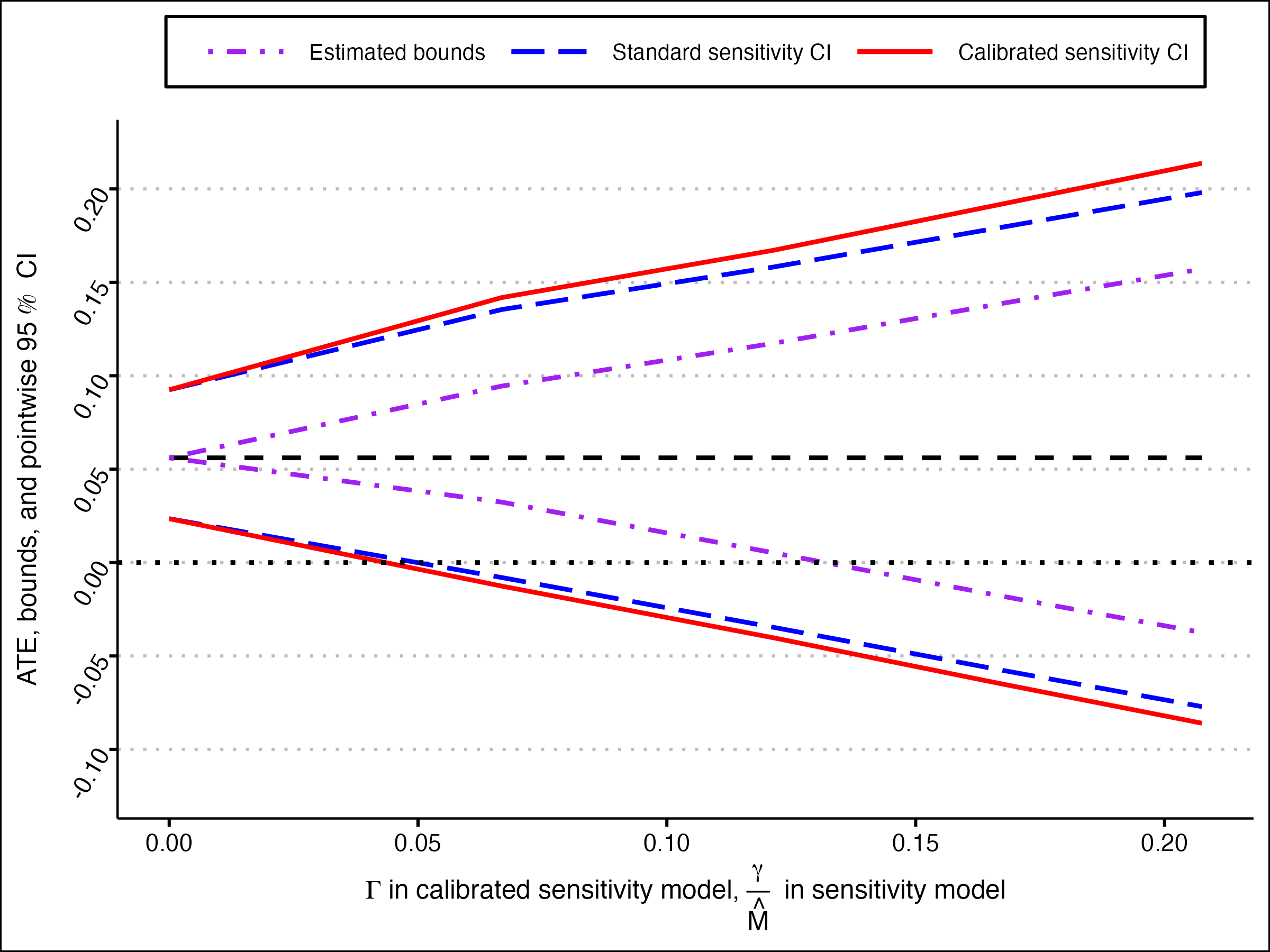} 
	\caption{Odds ratio model.}
	\label{fig:darfur_odds}
\end{subfigure}
\caption{\textbf{The effect of exposure to violence on preference for peace in Darfur.} These figures show estimates for the ATE, estimates for bounds on the ATE, and 95\% confidence intervals for the ATE (y-axis) at different levels of unmeasured confounding (x-axis).} 
\label{fig:darfur}
\end{figure}

\subsection{Discussion: the odds ratio model}

Here, we further discuss the results from the odds ratio model. Without standardizing by measured confounding, the confidence interval in the standard sensitivity model intersects the x-axis for $\gamma \approx 0.11$ (multiplying the point where the long dashed lines in Figure~\ref{fig:darfur_odds} intersect the x-axis by the estimate for measured confounding $\widehat M$). On the exponential scale, which this model is usually expressed in, that corresponds to a sensitivity parameter of roughly $1.12$. Arguably, this indicates much less robustness to unmeasured confounding than is observed in the the effect differences analysis. The lack of robustness to unmeasured confounding may be an artifact of the estimation process --- for computational simplicity, we manually constructed sieve estimators for the nuisance functions $\nu_a^\pm$ and $\theta_a^\pm$ using a linear model with no interactions.  This may be a poor approximation of the true functions, and better estimators could lead to different results. 

\bigskip

Moreover, when measured confounding is incorporated in the calibrated sensitivity model, the results become even less robust to unmeasured confounding --- the calibrated confidence intervals intersect the x-axis for $\Gamma \approx 0.025$, suggesting unmeasured confounding which is only a small multiple of measured confounding would be enough to overturn the conclusions of the study.  Part of the reason this occurs is because the maximum LOO confounding is high and has high variance.  The maximum estimated confounding is $3.34$, with a huge confidence interval of $[0,20]$ (this is a Wald-type confidence interval truncated below at $0$ where the variance was calculated across $100$ bootstrap resamples).  The estimated measured confounding may be high because the best projection estimator is a poor approximation of the true propensity score.  The high variance across bootstrap resamples reinforces this and suggests this estimator is unstable. More generally, understanding why these results indicate such different conclusions from the effect differences analysis is a topic of future work.

\section{Results in Section~\ref{sec:partial_id}} \label{app:id_results}

\subsubsection*{Proposition~\ref{prop:partial_id}}

\begin{newproof}
The proof follows for the effect differences model directly from the definition of the model, while for the odds ratio model it follows by \citet{yadlowsky2022bounds} Lemma 2.1.

\bigskip

For the outcome model,
\begin{align*}
	\bbE(Y^1 - Y^0) &= \bbE \left\{ \bbE(Y^1 -Y^0 \mid X) \right\} \\
	&=  \bbE \left\{ \bbE(Y^1 -Y^0 \mid A = 1, X) \pi_1(X) +  \bbE(Y^1 -Y^0 \mid A = 0, X) \pi_0(X) \right\} \\
	&= \bbE \{ \bbE(Y^1 \mid A = 1, X) - \bbE(Y^0 \mid A = 0, X) \} \\
	&\hspace{0.2in}- \bbE \left\{ \bbE(Y^1 \mid A = 1, X) \pi_0(X) + \bbE(Y^0 \mid A = 1, X) \pi_1(X) \right\}\\
	&\hspace{0.2in}+ \bbE\left\{ \bbE(Y^1 \mid A = 0, X) \pi_0 (X) + \bbE(Y^0 \mid A=0, X) \pi_1 (X) \right\} \\
	&= \psi + \sum_{a \in \{ 0, 1\}} (2a-1) \bbE \left[ \left\{ \bbE(Y^a \mid A = 1-a, X ) - \mu_a(X) \right\} \pi_{1-a}(X) \right].
\end{align*}
Therefore, by H\"{o}lder's inequality,
$$
\left| \psi^\ast - \psi \right| \leq \sum_{a \in \{0,1\}} \Big\lVert \bbE(Y^a \mid A = 1-a, X ) - \mu_a(X)  \Big\rVert_p \Big\lVert \pi_{1-a}(X) \Big\rVert_{r}
$$
where $r = 1$ if $p = \infty$ and $r = \frac{p}{p-1}$ otherwise.  Under calibrated sensitivity model~\ref{csm:avg_lso_out}, the result follows by setting $p = r = 2$ and imposing the model. 
\end{newproof}

\subsubsection*{Proposition \ref{prop:invariance}}

\begin{newproof}
The calibrated sensitivity model imposes the bound $U \leq \Gamma M$. Because $\Gamma \in (0, \infty)$ and $0 < M < \infty$ by assumption, $\Gamma M \in (0, \infty)$. Therefore, this model implies an upper bound on the causal effect of interest: $u(\Gamma M)$ (substitute $\gamma^\prime = \Gamma M$ and use the assumption/construction of the sensitivity model). Because $u(\cdot)$ is a deterministic function, it maps the same inputs to the same outputs, and the result follows.
\end{newproof}

\subsection{Differentiability with the odds ratio model}

Here, we provide the details on the derivatives considered implicitly in Lemma~\ref{lem:show_diff_monotone}. For analyzing the odds ratio model, we require several mild continuity and boundedness assumptions. For this purpose, we define one more nuisance function and a loss function:
\begin{align}
\omega_\theta (Y; t) &:= \big( Y - \theta \big) \one \big( Y > \theta \big) + t\big( Y - \theta \big) \one \big( Y < \theta \big), \label{eq:omega} \\
\ell_t^- (\theta, y) &:= \frac{1}{2} \Big[ \big\{ (y - \theta) \one(y > \theta) \big\}^2 + t \big\{ (y - \theta) \one(y < \theta) \}^2 \Big]\text{, and} \\
\ell_t^+ (\theta, y) &= \ell_{1/t}^-(\theta, y).
\end{align}
Next, we give the continuity and boundedness conditions required (see Lemmas 2.2 and 2.3 and Assumption A.1 in \citet{yadlowsky2022bounds}).
\begin{assumption} \label{asmp:continuity}
\emph{Smooth and bounded data generating process: } 
\begin{enumerate}
	\item $f(x)$, the covariate density, $\pi_1(x)$, and $\theta_a^-(x; t)$ (defined in \eqref{eq:theta}) and $\theta_a^+(x; t)$ for $a \in \{0,1\}$ are continuous in $x$,
	\item for $a \in \{0,1\}$, $p(y \mid A = a, x)$, the conditional density of $Y$ given $A = a$ and $X = x$ is continuous in $y$ and $x$, upper bounded, and has bounded support,
	\item $Y$ is bounded, 
	\item for $a \in \{0,1\}$ and all $x \in \bbR^d$, $\theta_a^-( x; t )$ and $\theta_a^+(x; t)$ are bounded, 
	\item for $a \in \{0,1\}$ and bounded $t$,  $\theta_a^-( x; t )$ and $\theta_a^+(x; t)$ are absolutely integrable in $x$, and
	\item for $a \in \{0,1\}$, $(r, x) \mapsto \bbE \left\{ \ell_{\exp(\Gamma M)}^- \big( r, Y \big) \mid A = a, x \right\} $ and $(r, x) \mapsto \bbE \left\{ \ell_{\exp(\Gamma M)}^+ \big( r, Y \big) \mid A = a, x \right\}$ are continuous on $\bbR \times \bbR^d$, and $\bbE \left( \ell_{\exp(\Gamma M)}^- \Big[ \theta_a^- \{ X; \exp(\Gamma M) \}, Y \Big] \mid A = a \right) < \infty$ and \\ $\bbE \left( \ell_{\exp(\Gamma M)}^+ \Big[ \theta_a^+ \{ X; \exp(\Gamma M) \}, Y \Big] \mid A = a \right) < \infty$. 
\end{enumerate}
\end{assumption}
This assumption imposes mild continuity and boundedness conditions. Condition 1, the continuity of $p(y \mid A = a, x)$, the boundedness of $Y$, and the absolute integrability of $\theta_a^{\pm}$ are beyond what is assumed in \citet{yadlowsky2022bounds}. Condition 2 is contained in Assumption A.1 of that paper, while condition 4 is in Lemma 2.2 and condition 6 is in Lemma 2.3. With this assumption in hand, we can now state and prove the following result --- Lemma~\ref{lem:diff} --- which implies Lemma~\ref{lem:show_diff_monotone}. 
\begin{restatable}{lemma}{lemmadiff} \label{lem:diff}
\emph{\textbf{(Differentiable bounds in the odds ratio model)}} Under the setup of Lemma~\ref{lem:show_diff_monotone},
\begin{align} 
	\frac{\partial}{\partial M} \mathcal{L}(\Gamma) &= - \Gamma \exp (\Gamma M) \bbE \left\{  \frac{ \pi_0(X) f_1 ( X; \theta_1^{-} )}{\nu_1^- (X)} \right\} -  \Gamma \bbE \left\{ \pi_1(X) \frac{\widetilde{f}_0 (X; \theta_0^+ )}{\nu_0^+ (X)}  \right\} \text{ and }\label{eq:odds_diff_lower} \\
	\frac{\partial}{\partial M} \mathcal{U}(\Gamma) &= \Gamma \bbE \left\{ \pi_0(X) \frac{\widetilde{f}_1 ( X; \theta_1^+ )}{\nu_1^+ (X)} \right\} + \Gamma \exp(\Gamma M) \bbE \left\{ \pi_1(X) \frac{f_0 ( X; \theta_0^- )}{\nu_0^-(X)} \right\}, \label{eq:odds_diff_upper}
\end{align}
where $\theta_a^\pm \equiv \theta_a^\pm \{ X; \exp(\Gamma M ) \}$; $\theta_a^\pm$ and $M$ are defined in \eqref{eq:theta} and \eqref{eq:M_odds}, respectively; and 
\begin{align}
	f_a(x; \theta) &= \bbE \left\{ \left( \theta - Y \right) \one\left( \theta > Y \right) \mid A = a, X = x \right\}, \label{eq:f_func} \\ 
	\widetilde f_a(x; \theta) &= \bbE \big\{ ( Y - \theta ) \one(Y > \theta) \mid A = a, X = x\big\}, \label{eq:f_tilde} \\
	\nu_a^-(x) \equiv \nu_a^- \{ x; \theta,  t = \exp(\Gamma M) \} &= \bbP ( Y \geq \theta \mid A = a, X = x) + t \bbP ( Y < \theta \mid A = a, X = x )\text{, and} \label{eq:nu_minus} \\
	\nu_a^+(x) \equiv \nu_a^+ \{ x; \theta, t = \exp(\Gamma M) \} &= \bbP ( Y > \theta \mid A = a, X=x) + \frac{1}{t} \bbP(Y < \theta \mid A = a, X=x). \label{eq:nu_plus}
\end{align}
Moreover, the derivative of the upper bound is positive and bounded while the the derivative of the lower bounds is negative and bounded. 
\end{restatable}

\noindent Before proving Lemma~\ref{lem:diff}, we state and prove several helper lemmas.

\begin{lemma} \label{lem:odds_deriv}
Under the setup of Lemma~\ref{lem:diff}, for $t^\prime > 0$,
\begin{equation}
	\frac{\partial}{\partial t} \theta_1^-(X; t) \Big|_{t = t^\prime} = -\frac{f_1(X; \theta_1^-)^2}{g_1 \big( X; \theta_1^- \big) },
\end{equation}
where $\theta_1^- := \theta_1^-(X; t^\prime)$ on the right-hand side, $f_a(x; \theta)$ is defined in \eqref{eq:f_func}, and 
\begin{equation}
	g_a (x; \theta) = f_a(x; \theta) \bbP \left( Y > \theta \mid A = a, X = x \right) + \widetilde f_a(x; \theta) \bbP\left( Y < \theta \mid A = a, X = x\right), \label{eq:g_func} 
\end{equation}
where $\widetilde f$ is defined in \eqref{eq:f_tilde}.
\end{lemma}

\begin{newproof}
We establish the result using the inverse function theorem, expressing $t$ as a function of $\theta$. Under Assumption~\ref{asmp:continuity}, Lemma 2.3 from \citet{yadlowsky2022bounds} applies, such that $\theta_1^- (X; t)$ is the unique minimizer (up to measure-zero sets) of $\bbE \{ 	\ell_t^- (\theta, Y) \mid A = 1, X \}$ and solves
\begin{equation} 
	\bbE \{ \omega_{\theta_1^- (X; t)} (Y; t) \mid A = 1 , X \} = 0.
\end{equation}
In what follows let $\theta \equiv \theta_1^-(X; t)$ to further simplify notation. Equivalently, the moment condition says that
\begin{equation} \label{eq:moment}
	\bbE \big\{ (Y - \theta) \one(Y > \theta) \mid A =1, X =x \big\} + t(\theta) \bbE \big\{ (Y - \theta) \one(Y < \theta) \mid A =1, X = x \big\} = 0.
\end{equation}
It is possible to subtract $t(\theta) \bbE \big\{ ( Y- \theta) \one(Y < \theta) \mid A = 1 , X = x \big\}$ and then divide by $\bbE \big\{ ( \theta - Y) \one(\theta > Y) \mid A = 1 , X = x \big\}$ when $t > 0$ because
\begin{enumerate}
	\item $\bbE \big\{ ( Y - \theta ) \one(Y > \theta) \mid A = 1, X = x\big\}$ and $\bbE \big\{ ( \theta - Y ) \one(\theta > Y) \mid A = 1 , X = x \big\}$ are both non-negative by definition, and 
	\item $\bbE \big\{ ( Y - \theta ) \one(Y > \theta) \mid A = 1, X = x\big\}$ and $\bbE \big\{ ( \theta - Y ) \one(\theta > Y) \mid A = 1 , X = x \big\}$ cannot each be zero simultaneously because $p(y \mid A = 1, x)$ is continuous and upper bounded,  which implies $Y \mid A = 1, X =x$ is not a point mass.
\end{enumerate}
From 1 and 2, we conclude both $\bbE \big\{ ( Y - \theta ) \one(Y > \theta) \mid A = 1, X = x\big\}$ and $\bbE \big\{ ( \theta - Y ) \one(\theta > Y) \mid A = 1 , X = x \big\}$ must be positive.  Therefore,
$$
t(\theta) = \frac{\bbE \big\{ (Y - \theta) \one(Y > \theta) \mid A = 1, X =x \big\}}{ \bbE \big\{ (\theta - Y) \one (\theta > Y) \mid A =1, X = x \big\}} 
$$
By the continuity assumption on $p(y \mid A = 1, x)$ and the quotient rule and Leibniz' integral rule (see Lemma~\ref{lem:deriv_technical} for specific steps), we can take the derivative of $t(\theta)$ wrt $\theta$, and find 
\begin{align}
	\frac{\partial}{\partial \theta} t(\theta) := - \Bigg( &\frac{ \bbE \big\{ ( \theta - Y ) \one(\theta > Y) \mid A = 1 , X = x \big\} \bbP (Y > \theta \mid A = 1, X =x)}{\Big[ \bbE \big\{ ( \theta - Y ) \one(\theta > Y) \mid A = 1 , X = x \big\} \Big]^2} \nonumber \\
	&+ \frac{\bbE \big\{ ( Y - \theta ) \one(Y > \theta) \mid A = 1, X = x\big\} \bbP( Y \leq \theta \mid A = 1, X = x)}{\Big[ \bbE \big\{ ( \theta - Y ) \one(\theta > Y) \mid A = 1 , X = x \big\} \Big]^2} \Bigg) \label{eq:inverse}
\end{align}
Notice that the denominator is positive by the prior argument and steps 1 and 2 above.  Meanwhile, the numerator is also positive by the same argument and steps 1 and 2 above, and because $\bbP (Y > \theta \mid A = 1, X =x) \vee \bbP( Y \leq \theta \mid A = 1, X = x) > 0$.  Moreover, both the numerator and denominator are bounded by Assumption~\ref{asmp:continuity}.

\bigskip

Hence, $\frac{\partial}{\partial \theta} t(\theta)$ exists and is always negative and bounded. Therefore, the result follows by the inverse function theorem for the reciprocal of \eqref{eq:inverse}.  Moreover, $\frac{\partial}{\partial t} \theta(t)$ is continuous because $p(y \mid A = 1, x)$ is continuous and $\frac{\partial}{\partial \theta} t(\theta)$ is a composition of integrals of $p(y \mid A = 1, x)$.
\end{newproof}

\begin{lemma} \label{lem:odds_simplify}
Under the setup of Lemma~\ref{lem:diff}, then 
\begin{equation}
	\nu_1^- \big(X; \theta_1^{-}, t \big) f_1(X; \theta_1^-) = g_1 \big( X; \theta_1^-, t \big) \label{eq:nufg}
\end{equation}
where $\theta_1^- := \theta_1^-(X; t)$, and $\nu_1^-(X, \theta, t)$, $f_a(x; \theta)$, and $g_a(X; \theta, t)$ are defined in \eqref{eq:nu_minus}, \eqref{eq:f_func}, and \eqref{eq:g_func}, respectively.
\end{lemma}

\begin{newproof}
Omitting arguments, on the left-hand side of \eqref{eq:nufg} we have
\begin{equation}
	\nu f_1 = f_1 \bbP( Y \geq \theta \mid A = 1, X=x ) + t f_1 \bbP( Y \leq \theta \mid A = 1, X = x). \label{eq:nuf}
\end{equation}
For the second summand on the right-hand side of \eqref{eq:nuf},
\begin{align*}
	t f_1 \bbP( Y \leq \theta \mid A = 1, X = x) &= t \bbE \left\{ \left( \theta - Y \right) \one\left( \theta > Y \right) \mid A = a, X = x \right\} \bbP( Y \leq \theta \mid A = 1, X = x) \\
	&= - t \bbE \left\{ \left( Y - \theta\right) \one\left( \theta > Y \right) \mid A = a, X = x \right\} \bbP( Y \leq \theta \mid A = 1, X = x) \\
	&= \bbE \left\{ \left( Y - \theta\right) \one\left( Y > \theta \right) \mid A = a, X = x \right\} \bbP( Y \leq \theta \mid A = 1, X = x)
\end{align*}
where the third equality follows because $\theta$ solves the moment equation $\bbE ( \omega_\theta (Y) \mid A = 1, X = x) = 0$ (see, \eqref{eq:moment}). Hence,
\begin{align*}
	\nu f_1 &= f_1 \bbP( Y \geq \theta \mid A = 1, X=x ) \\
	&+ \bbE \left\{ \left( Y - \theta\right) \one\left( Y > \theta \right) \mid A = a, X = x \right\} \bbP( Y \leq \theta \mid A = 1, X = x) \\
	&= g_1 (X; \theta).
\end{align*}
\end{newproof}

\begin{lemma} \label{lem:upper_diff}
Under the setup of Lemma~\ref{lem:diff}, for $t^\prime > 0$,
\begin{equation}
	\frac{\partial}{\partial t} \theta_1^+(X; t) \Big|_{t = t^\prime} = \frac{\widetilde f_1 (X; \theta_1^+ )^2}{g_1( X; \theta_1^+ )} = \frac{\widetilde f_1(X; \theta_1^+)}{t^\prime \nu_1^+ \big( X; \theta_1^+, t^\prime \big) },
\end{equation}
where $\theta_1^+ := \theta_1^+(X; t^\prime)$ on the right-hand side, and $\widetilde f_a(x; \theta)$, $g_a(x; \theta)$ and $\nu_a^+$ are defined in \eqref{eq:f_tilde}, \eqref{eq:g_func}, and \eqref{eq:nu_plus}, respectively.
\end{lemma}

\begin{newproof}
By the same argument as in \citet{yadlowsky2022bounds} (see, Section 2.3), $\theta_1^+(X; t)$ minimizes $\bbE \{ \ell_t^+ (\theta, Y) \mid A = 1, X \} = \bbE \{ \ell_{1/t}^- (\theta, Y) \mid A =1, X \}$ and solves the moment condition
$$
\bbE \left\{ \left( Y - \theta \right) \one\left( Y > \theta \right) \mid A = 1, X = x \right\} - \frac{1}{t} \bbE \left\{ \left( \theta - Y \right) \one\left( \theta > Y \right) \mid A = 1, X = x \right\} = 0.
$$
Therefore, the first result follows by the same arguments as in Lemma~\ref{lem:odds_deriv} using the inverse function theorem.

\bigskip

\noindent For the second result, omitting $X$ arguments and letting $\theta \equiv \theta_1^+$,
\begin{align*}
	\frac{\widetilde f_1(\theta; t^\prime)}{\frac{\partial}{\partial t} \theta_1^+}  &= \frac{g_1(\theta; t^\prime)}{\widetilde f_1(\theta; t^\prime)} \\
	&= \frac{\bbE \left\{ \left( \theta - Y \right) \one\left( \theta > Y \right) \mid A = 1, X = x \right\}  \bbE \left( Y \geq \theta \mid A = 1, X = x \right)}{\bbE \left\{ \left( Y - \theta \right) \one\left( Y > \theta \right) \mid A = 1, X = x \right\}} \\
	&+ \frac{\bbE \left\{ \left( Y - \theta \right) \one \left( Y > \theta \right) \mid A = 1, X= x\right\} \bbP\left( Y \leq \theta \mid A = 1, X = x\right)}{\bbE \left\{ \left( Y - \theta \right) \one\left( Y > \theta \right) \mid A = 1, X = x \right\}} \\
	&= \frac{\bbE \left\{ \left( \theta - Y \right) \one\left( \theta > Y \right) \mid A = 1, X = x \right\}}{\bbE \left\{ \left( Y - \theta \right) \one\left( Y > \theta \right) \mid A = 1, X = x \right\}}  \bbP \left( Y \geq \theta \mid A = 1, X = x \right) + \bbP(Y \leq \theta \mid A = 1, X = x),
\end{align*}
where the first line follows by rearranging the first result and letting $\frac{\partial}{\partial t} \theta_1^+ \equiv \frac{\partial}{\partial t} \theta_1^+(X; t) \Big|_{t = t^\prime}$. Because $\theta_1^+$ solves the moment equation
$$
\bbE \left\{ \left( Y - \theta \right) \one\left( Y > \theta \right) \mid A = 1, X = x \right\} - \frac{1}{t^\prime} \bbE \left\{ \left( \theta - Y \right) \one\left( \theta > Y \right) \mid A = 1, X = x \right\} = 0
$$
we have
$$
\frac{\widetilde f_1(\theta; t^\prime)}{\frac{\partial}{\partial y} \theta_1^+} = t^\prime \bbP \left( Y > \theta \mid A = 1, X = x \right) + \bbP \left( Y < \theta \mid A = 1, X = x \right) = t^\prime \nu_1^+ \{ X; \theta_1^+ (X; t^\prime), t^\prime\}.
$$
The result follows by rearranging.
\end{newproof}

\begin{lemma} \label{lem:deriv_technical}
Under the setup of Lemma~\ref{lem:diff},
\begin{align}
	\frac{\partial}{\partial \theta} \bbE \{ (Y - \theta) \one(Y > \theta) \mid A = a, X= x \} &= - \bbP(Y > \theta \mid A = a, X = x) \\
	\frac{\partial}{\partial \theta} \bbE \{ (\theta - Y) \one(\theta \geq Y) \mid A = a, X= x \} &= \bbP(Y \leq \theta \mid A = a, X)
\end{align}
\end{lemma}
\begin{newproof}
By the continuity assumption on $p(y \mid A = a, x)$, the boundedness of $Y$, and Leibniz' integral rule,
\begin{align*}
	\frac{\partial}{\partial \theta} \bbE \{ (Y - \theta) \one(Y > \theta) \mid A = a, X= x \} &\equiv \frac{\partial}{\partial \theta} \int_{\theta}^{\overline B} (Y - \theta) p(y \mid A = a, x ) dy\\
	&= 0 + \int_{\theta}^{\overline B} \frac{\partial}{\partial \theta} (Y - \theta) p(y \mid A = a, x ) dy \\
	&= - \int_{\theta}^{\overline B} p(y \mid A = a, x) dy \\
	&= - \bbP( Y > \theta \mid A = a, X = x),
\end{align*}
where $\overline B$ is the upper bound on $Y$, and
\begin{align*}
	\frac{\partial}{\partial \theta} \bbE \{ (\theta - Y) \one(\theta \geq Y) \mid A = a, X= x \} &= \int_{\underline B}^{\theta} \frac{\partial}{\partial \theta} (\theta - Y) p(y \mid A = a, x ) dy \\
	&= \bbP( Y \leq \theta \mid A = a, X = x),
\end{align*}
where $\underline B$ is the lower bound on $Y$.
\end{newproof}

\subsection{Proof of Lemma~\ref{lem:diff}}

Finally, with the helper lemmas in hand, we prove Lemma~\ref{lem:diff}.
\begin{newproof}
The derivatives for the upper and lower bounds on $\bbE(Y^1)$ follow from this proof and Lemmas~\ref{lem:odds_deriv}-\ref{lem:deriv_technical}.  The other derivatives for the upper and lower bounds on $\bbE(Y^0)$ follow the same argument, but replacing $A=1$ by $A=0$.  

\bigskip

Note that, by assumption $0 < \exp(\Gamma M) < \infty$.  Hence, Lemmas~\ref{lem:odds_deriv},~\ref{lem:odds_simplify}, and~\ref{lem:upper_diff} apply, and the result then follows by the chain rule, Assumption~\ref{asmp:positivity} and the dominated convergence theorem (to bring the derivative inside the expectation), and canceling terms.  The boundedness of the derivative follows because $\widetilde f_a, f_a, \nu_a^\pm$ are all bounded and $M$ and $\Gamma$ are non-zero and bounded by assumption.  The sign follows because $f_a, \widetilde f_a, \nu_a^\pm > 0$.
\end{newproof}

\section{Efficient Influence Functions} \label{app:eifs}

To construct doubly robust-style estimators for the bounds on the ATE we will use, when possible, the efficient influence functions (EIFs) for each part of the bound. The EIF is a concept from  semiparametric efficiency theory \citep{bickel1993efficient, van2000asymptotic, tsiatis2006semiparametric, van2002semiparametric, van2003unified}, and can be thought of as the first derivative in a von Mises expansion of an estimand \citep{von1947asymptotic}. The crucial benefit of the EIF is that estimators constructed from it are doubly robust, in the sense that they have bias that is a second-order product of errors of the relevant nuisance function estimators, so they can achieve $\sqrt{n}$-efficiency even when the nuisance functions are estimated at slower nonparametric rates \citep{van2003unified, chernozhukov2018double, kennedy2022semiparametric}. 

\subsection{Effect differences}

For the maximum LOO effects differences model (calibrated sensitivity model~\ref{csm:max_loo_fx}), we only require the EIF of the adjusted mean difference, provided in \eqref{eq:mde_if} in the body of the paper. 

\subsection{Odds ratio}

For the maximum LOO odds ratio model (calibrated sensitivity model~\ref{csm:max_loo_odds}), we can use the EIF from \citet{yadlowsky2022bounds} for the bounds (see, e.g., \eqref{eq:yad_eif_lower} for the EIF of the lower bound on $\bbE(Y^1)$). The efficient influence function for the upper bound on the ATE is
\begin{align}
\varphi_{U} \{ Z; \eta(\Gamma M) \} &:= AY + (1-A) \theta_1^+ \{ X; \exp(\Gamma M) \} + A\frac{\omega_{\theta_1^+}(Y; \exp(\Gamma M)) \pi_0(X)}{\nu_1^+ \{ X; \theta_1^+, \exp(\Gamma M)\} \pi_1(X)} \nonumber \\
&- (1-A)Y + A \theta_0^- \{ X; \exp(\Gamma M) \} + (1-A) \frac{\omega_{\theta_0^-}(Y; \exp(\Gamma M)) \pi_1(X)}{\nu_0^-\{ X; \theta_1^+, \exp(\Gamma M) \} \pi_0(X)} \label{eq:yad_eif_upper}
\end{align}
Meanwhile, the efficient influence function for the lower bound on the ATE is 
\begin{align}
\varphi_{L} \{ Z; \eta(\Gamma M) \} &:= AY + (1-A) \theta_1^- \{ X; \exp(\Gamma M) \} + A \frac{\omega_{\theta_1^-} \{ Y; \exp(\Gamma M) \} \pi_0(X)}{\nu_1^- \{ X; \theta_1^-, \exp(\Gamma M) \} \pi_1(X)} \nonumber \\
&- (1-A)Y + A \theta_0^+ \{ X; \exp(\Gamma M) \} + (1-A) \frac{\omega_{\theta_0^+}\{Y; \exp(\Gamma M) \} \pi_1(X)}{\nu_0^+ \{ X; \theta_1^+, \exp(\Gamma M)\} \pi_0(X)} \label{eq:yad_eif_lower}
\end{align}
where $\theta^{\pm}, \nu^{\pm}$, and $\omega$ are defined in \eqref{eq:theta}, \eqref{eq:nu_minus}, \eqref{eq:nu_plus}, and \eqref{eq:omega}.

\bigskip

Because the measured confounding is a supremum, it is not pathwise differentiable and does not admit an EIF (see, e.g., \citet{hines2022demystifying} for discussion of a lack of pathwise differentiability). Instead, we estimated its projection onto a finite dimensional model, which can be estimated at a $\sqrt{n}$-rate with a plug-in estimator, as in Definition~\ref{def:odds}.

\subsection{Outcome model}

The average LSO outcome model (calibrated sensitivity model~\ref{csm:avg_lso_out}) requires the derivation of two EIFs, for each multiplicand in the bounds in \eqref{eq:out_id}.  The following result provides them.
\begin{restatable}{lemma}{lemouteifs} \label{lem:out_eifs}
The un-centered efficient influence function of $\lVert \pi_{a}(X) \rVert_2^2$ is
\begin{equation} \label{eq:xi}
	\xi_{a}(Z) = \pi_{a}(X)^2 + 2 \pi_{a}(X) \big\{ \one(A = a) - \pi_{a}(X) \big\},
\end{equation}
while the un-centered efficient influence function of $\lVert \mu_a(X_{-S}) - \bbE \big\{ \mu_a(X) \mid A = 1-a, X_{-S} \} \rVert_2^2$ is
\begin{align}
	\lambda_a(Z; S) &= \Big[ \mu_a (X_{-S}) - \bbE \{ \mu_a(X) \mid A = 1 -a, X_{-S} \} \Big]^2 \nonumber \\
	&+ 2 \Big[ \mu_a (X_{-S}) - \bbE \{ \mu_a(X ) \mid A = 1 -a, X_{-S} \} \Big] \Bigg( \frac{\one(A = a)}{\pi_a(X_{-S})} \{ Y - \mu_a (X_{-S}) \} \nonumber \\
	&- \left\{ \frac{\one(A = a)}{\pi_a (X)} \right\} \{ Y - \mu_a (X) \} \left\{ \frac{\pi_{1-a}(X)}{\pi_{1-a}(X_{-S})} \right\} \nonumber \\
	&- \left\{ \frac{\one(A = 1-a)}{\pi_{1-a}(X_{-S})} \right\} \Big[ \mu_a (X) - \bbE \{ \mu_a (X) \mid A =1-a, X_{-S} \} \Big] \Bigg), \label{eq:lambda}
\end{align}
\end{restatable}
Before the proof, we provide a few comments. Each EIF follows the standard form of a plug-in plus a weighted residual.  The only nuance is that the weighted residuals are multiplied by two.  This occurs because each functional includes the mean of a square.  Intuitively, one can think of these extra terms arising in the same way that $\frac{d}{dx} f(x^2) = 2f(x)f^\prime(x)$ by the chain rule.  Indeed, these results are reminiscent of the EIF for the integral of the squared density, a canonical functional in the semiparametric / doubly robust estimation literature \citep{birge1995estimation, laurent1997estimation}.  

\begin{newproof}
We prove Lemma~\ref{lem:out_eifs} in two stages.  First, we show that each functional considered satisfies a von Mises expansion whose influence function is the relevant efficient influence function proposed in the result and which admits a second order remainder term.  We will prove this for both functionals under consideration. Second, we will relate the von Mises expansion to smooth parametric submodels and submodel scores to prove that the proposed efficient influence function is, indeed, the efficient influence function.

\subsubsection*{Second order remainder for $\lVert \pi_a(X) \rVert_2^2$}

We begin with $\lVert \pi_a(X) \rVert_2^2$.  Let $\Xi_a(\mathcal{P}) \equiv \lVert \pi_a(X) \rVert_2^2$ and $\xi_a(\mathcal{P}) \equiv \xi_a(Z)$, where $\xi_a$ was defined in \eqref{eq:xi}.  Then, by a von Mises expansion,
\begin{equation} 
	\Xi_a(\overline{\mathcal{P}}) = \Xi_a(\mathcal{P}) + \int_{\mathcal{Z}} \xi_a(\overline{\mathcal{P}}) d(\overline{\mathcal{P}} - \mathcal{P}) + R_2 (\overline{\mathcal{P}}, \mathcal{P})
\end{equation}
where $\mathcal{P}$ and $\overline{\mathcal{P}}$ are two different distributions at which $\Xi_a$ is evaluated.  Rearranging and canceling $\Xi_a(\overline{\mathcal{P}}) - \Xi_a(\overline{\mathcal{P}})$ we have
\begin{equation}
	R_2 (\overline{\mathcal{P}}, \mathcal{P}) = \Xi_a(\overline{\mathcal{P}}) - \Xi_a(\mathcal{P}) - \int_{\mathcal{Z}} \xi_a(\overline{\mathcal{P}}) d(\overline{\mathcal{P}} - \mathcal{P}) = \int_{\mathcal{Z}} \xi_a(\overline{\mathcal{P}}) - \xi_a(\mathcal{P}) d(\mathcal{P}) 
\end{equation}
The right-hand side of the above equation can be re-written as an expectation over $\mathcal{P}$, but evaluated with nuisance functions from $\overline{\mathcal{P}}$ and $\mathcal{P}$; i.e.,
$$
R_2 (\overline{\mathcal{P}}, \mathcal{P}) \equiv \bbE \left\{ \overline{\xi}_a(Z) - \xi_a (Z) \right\} 
$$
Omitting arguments, so that $\pi_a(X) \equiv \pi$, then, by the definition of $\xi_a$, we have
\begin{align}
	R_2 (\overline{\mathcal{P}}, \mathcal{P}) &= \bbE \Bigg( \overline \pi_{a}^2 + 2 \overline \pi_{a} \big\{ \one(A = a) - \overline \pi_{a} \big\} - \pi_{a}^2 - 2 \pi_{a} \big\{ \one(A = a) - \pi_{a} \big\} \Bigg) \nonumber \\
	&= \bbE \left\{ \overline \pi_{a}^2 - \pi_a^2 + 2 \overline \pi_{a} \big( \pi_a - \overline \pi_{a} \big) \right\} \nonumber \\
	&= - \bbE \left\{ \big( \overline \pi_{a} - \pi_a \big)^2 \right\} \label{eq:xi_bias}
\end{align}
where the second line follows by iterated expectations on $X$ and the third by rearranging.  This shows that the remainder term is second order, which will imply that $\xi_a$ is the efficient influence function, as we show subsequently.

\subsubsection*{General framework for the next section}

Here, we discuss an informal summary of the next section, which hopefully will help the reader as a guide as the algebra becomes more involved.  First, we can informally write the efficient influence function as
$$
\varphi^2 + 2 \rho
$$
so that
$$
R(\overline{\mathcal{P}}, \mathcal{P}) = \bbE \left( \overline{\varphi}^2 + 2 \overline \varphi \overline \rho - \varphi^2 - 2 \varphi \rho \right) 
$$
For the section above, $\varphi$ and $\rho$ were simple enough that we could analyze the remainder term directly and show it is second order.  For the next section, the algebra is more involved, but follows essentially the same format. As we will see subsequently, the algebra required to demonstrate $\lambda$ is the efficient influence functions can be split into three stages:
\begin{enumerate}
	\item Show that $\bbE(\rho) = 0$.
	\item Notice then that 
	$$
	\bbE \left( \overline{\varphi}^2 + 2 \overline \rho - \varphi^2 - 2 \rho \right) = \bbE \left\{ - (\overline \varphi - \varphi)^2 + 2 (\overline \rho + \overline \varphi^2 - \overline \varphi \varphi) \right\}
	$$
	using 1. and then completing the square.
	\item Show that $- \bbE \{ (\overline \varphi - \varphi)^2 \}$ is second order.
	\item Show that $\bbE \{ \overline \rho + \overline \varphi^2 - \overline \varphi \varphi \}$ is second order.
\end{enumerate}

\subsubsection*{Second order remainder for $\lVert \mu_a(X_{-S}) - \bbE \big\{ \mu_a(X) \mid A = 1-a, X_{-S} \} \rVert_2^2$}

Here, we consider $\lVert \mu_a(X_{-S}) - \bbE \big\{ \mu_a(X) \mid A = 1-a, X_{-S} \} \rVert_2^2$.  By the same logic as above, 
\begin{equation} \label{eq:remainder_overall_1}
	R_2 (\overline{\mathcal{P}}, \mathcal{P}) = \bbE \{ \overline{\lambda}_a (Z; S) - \lambda_a(Z; S) \}.
\end{equation}
Let 
\begin{align*}
	\varphi_a(X_{-S}) &= \mu_a (X_{-S}) - \bbE \{ \mu_a (X) \mid A =1-a, X_{-S} \} \text{ and } \\
	\rho_a(Z; S) &= \Big[ \mu_a (X_{-S}) - \bbE \{ \mu_a (X) \mid A =1-a, X_{-S} \} \Big] \Bigg( \frac{\one(A = a)}{\pi_a(X_{-S})} \{ Y - \mu_a (X_{-S}) \} \\
	&- \left\{ \frac{\one(A = a)}{\pi_a (X)} \right\} \{ Y - \mu_a (X) \} \left\{ \frac{\pi_{1-a}(X)}{\pi_{1-a}(X_{-S})} \right\} \\
	&- \left\{ \frac{\one(A = 1-a)}{\pi_{1-a}(X_{-S})} \right\} \Big[ \mu_a (X) - \bbE \{ \mu_a(X) \mid A = 1-a, X_{-S} \} \Big] \Bigg)
\end{align*}
so that
$$
\lambda_a(Z; S) = \varphi_a(X_{-S})^2 + 2 \rho_a(Z; S).
$$
Notice that
\begin{equation} \label{eq:varphi_rho_zero_1}
	\bbE \left\{ \rho_a(Z; S) \right\} =  0 
\end{equation}
by iterated expectations on $A = a, X$ and then iterated expectations on $A = a, X_{-S}$ and $A = 1-a, X_{-S}$.  Second, as in the informal proof sketch, notice that
\begin{equation} \label{eq:lambda_remainder}
	R_2 (\overline{\mathcal{P}}, \mathcal{P}) = \mathbb{E} \left[ - \{ \overline \varphi_a(X_{-S}) - \varphi_a(X_{-S}) \}^2 + 2 \Big\{ \overline \varphi_a(X_{-S})^2 + \overline \rho_a(Z; S) - \overline \varphi_a(X_{-S}) \varphi_a(X_{-S}) \} \right]
\end{equation}
which follows by \eqref{eq:remainder_overall_1}, the definition of $\lambda, \varphi,$ and $\rho$, \eqref{eq:varphi_rho_zero_1}, and completing the square.  Third, notice that $\mathbb{E} \left[ \{ \overline \varphi_a(X_{-S}) - \varphi_a(X_{-S}) \}^2 \right]$ is second order in terms of $\mu_a(X_{-S})$ and $\bbE \{ \mu_a(X) \mid A =1-a, X_{-S} \}$. Fourth, notice that
\begin{align*}
	\bbE \big\{ \overline \rho_a(Z; S)\} &= \bbE \Bigg\{ \Big[ \overline \mu_a (X_{-S}) - \overline \bbE \{ \overline \mu_a (X) \mid A =1-a, X_{-S} \} \Big] \Bigg( \frac{\one(A = a)}{\overline \pi_a(X_{-S})} \{ Y - \overline \mu_a (X_{-S}) \} \\
	&- \left\{ \frac{\one(A =a)}{\overline \pi_a (X)} \right\} \{ Y - \overline \mu_a (X)  \} \left\{ \frac{\overline \pi_{1-a}(X)}{\overline \pi_{1-a}(X_{-S})} \right\}  \\
	&- \left\{ \frac{\one(A = 1-a)}{\overline \pi_{1-a}(X_{-S})} \right\} \Big[ \overline \mu_a (X) - \overline \bbE \{ \overline \mu_a(X) \mid A = 1-a, X_{-S} \} \Big] \Bigg) \Bigg\} \\
	&\hspace{-1in}= \bbE \Bigg\{ \Big[ \overline \mu_a (X_{-S}) - \overline \bbE \{ \overline \mu_a (X) \mid A =1-a, X_{-S} \} \Big] \Bigg( \frac{\pi_a(X_{-S})}{\overline \pi_a(X_{-S})} \{ \mu_a(X_{-S}) - \overline \mu_a (X_{-S}) \} \\
	&- \left\{ \frac{\pi_a(X)}{\overline \pi_a (X)} \right\} \{ \mu_a(X) - \overline \mu_a (X)  \} \left\{ \frac{\overline \pi_{1-a}(X)}{\overline \pi_{1-a}(X_{-S})} \right\} \\
	&- \left\{ \frac{\pi_{1-a}(X_{-S})}{\overline \pi_{1-a}(X_{-S})} \right\} \Big[ (\bbE - \overline \bbE) \{ \overline \mu_a(X) \mid A = 1-a, X_{-S} \} \Big] \Bigg) \Bigg\} \\
	\vspace{0.5in} \\
	&\hspace{-1in}= \bbE \Bigg\{ \overline \varphi_a(X_{-S})  \Bigg( \left\{ \frac{\pi_a(X_{-S}) - \overline \pi_a(X_{-S})}{\overline \pi_a(X_{-S})} \right\} \{ \mu_a(X_{-S}) - \overline \mu_a (X_{-S}) \} + \{ \mu_a(X_{-S}) - \overline \mu_a (X_{-S}) \} \\
	&- \left\{ \frac{\pi_a(X) - \overline \pi_a(X}{\overline \pi_a (X)} \right\} \{ \mu_a(X) - \overline \mu_a (X)  \} \left\{ \frac{\overline \pi_{1-a}(X)}{\overline \pi_{1-a}(X_{-S})} \right\}  \\
	&- \{ \mu_a(X_{[p}) - \overline \mu_a(X) \} \left\{ \frac{\overline \pi_{1-a}(X)}{\overline \pi_{1-a}(X_{-S})} \right\} \\
	&- \left\{ \frac{\pi_{1-a}(X_{-S}) - \overline \pi_{1-a}(X_{-S})}{\overline \pi_{1-a}(X_{-S})} \right\} \Big[ (\bbE - \overline \bbE) \{ \overline \mu_a(X) \mid A = 1-a, X_{-S} \} \Big] \\
	&- (\bbE - \overline \bbE) \{ \overline \mu_a(X) \mid A = 1-a, X_{-S} \} \Bigg) \Bigg\} \\
	&\hspace{-1in} = \bbE \Bigg\{ \overline \varphi_a(X_{-S})  \Bigg( \left\{ \frac{\pi_a(X_{-S}) - \overline \pi_a(X_{-S})}{\overline \pi_a(X_{-S})} \right\} \{ \mu_a(X_{-S}) - \overline \mu_a (X_{-S}) \} \\
	&- \left\{ \frac{\pi_a(X) - \overline \pi_a(X)}{\overline \pi_a (X)} \right\} \{ \mu_a(X) - \overline \mu_a (X)  \} \left\{ \frac{\overline \pi_{1-a}(X)}{\overline \pi_{1-a}(X_{-S})} \right\} \\
	&- \left\{ \frac{\pi_{1-a}(X_{-S}) - \overline \pi_{1-a}(X_{-S})}{\overline \pi_{1-a}(X_{-S})} \right\} \Big[ (\bbE - \overline \bbE) \{ \overline \mu_a(X) \mid A = 1-a, X_{-S} \} \Big] \\
	&- \{ \mu_a(X_{[p}) - \overline \mu_a(X) \}  \left\{ \frac{\overline \pi_{1-a}(X)}{\overline \pi_{1-a}(X_{-S})} - \frac{\pi_{1-a}(X)}{\pi_{1-a}(X_{-S})} \right\} \\
	&-\{ \mu_a(X_{[p}) - \overline \mu_a(X) \} \left\{ \frac{\pi_{1-a}(X)}{\pi_{1-a}(X_{-S})} \right\} \\
	&- (\bbE - \overline \bbE) \{ \overline \mu_a(X) \mid A = 1-a, X_{-S} \} + \{ \mu_a(X_{-S}) - \overline \mu_a (X_{-S}) \}  \Bigg) \Bigg\}
\end{align*}
where the first line follows definition, the second by iterated expectations on $X$ and $A = a$ and then on $X_{-S}$, the third by adding zero, and the fourth again by adding zero and rearranging. 
Notice that $\overline \varphi$ was factored out at the start and that the first four summands are second order.  Focusing on the last two summands, notice that 
\begin{align*}
	\bbE \Bigg( \overline \varphi_a(X_{-S}) &\Bigg[ -\{ \mu_a(X_{[p}) - \overline \mu_a(X) \} \left\{ \frac{\pi_{1-a}(X)}{\pi_{1-a}(X_{-S})} \right\} \\
	&- (\bbE - \overline \bbE) \{ \overline \mu_a(X) \mid A = 1-a, X_{-S} \} + \{ \mu_a(X_{-S}) - \overline \mu_a (X_{-S}) \}  \Bigg] \Bigg) \\
	= \bbE \Bigg( \overline \varphi_a(X_{-S}) \Bigg[ &-\{ \mu_a(X_{[p}) - \overline \mu_a(X) \} \left\{ \frac{\pi_{1-a}(X)}{\pi_{1-a}(X_{-S})} \right\} \\
	&- \overline \varphi(X_{-S}) + \mu_a(X_{-S}) - \bbE \{ \overline \mu_a(X) \mid 1-a, X_{-S} \} ] \\
	&\hspace{-1in}= \bbE \Bigg[ \overline \varphi_a(X_{-S}) \Bigg\{ \varphi_a(X_{-S}) - \overline \varphi_a(X_{-S}) \Bigg\} \Bigg]
\end{align*}
where the second line follows by the definition of $\overline \varphi_a(X_{-S})$ and the final line follows because
\begin{align*}
	\bbE \Bigg[ \{ \mu_a(X_{[p}) - \overline \mu_a(X) \} \left\{ \frac{\pi_{1-a}(X)}{\pi_{1-a}(X_{-S})} \right\} \mid X_{-S} \Bigg] &= \int_{\mathcal{X}_{-S}}  \{ \mu_a(X_{[p}) - \overline \mu_a(x) \} \left\{ \frac{\pi_{1-a}(x)}{\pi_{1-a}(X_{-S})} \right\} \bbP(x \mid X_{-S}) dx \\
	&\hspace{-1in}= \int_{\mathcal{X}_{-S}}  \{ \mu_a(X_{[p}) - \overline \mu_a(x) \} \left\{ \frac{\bbP(x \mid A = 1-a, X_{-S})}{\bbP(x \mid X_{-S})} \right\} \bbP(x \mid X_{-S}) dx \\
	&\hspace{-2in} = \int_{\mathcal{X}_{-S}}  \{ \mu_a(X_{[p}) - \overline \mu_a(x) \} \bbP(x \mid A = 1-a, X_{-S}) \\
	&\hspace{-2in}= \bbE \{ \mu_a (X) - \overline \mu_a (X) \mid A = 1-a, X_{-S} \} 
\end{align*} 
where the first line follows by definition, the second because, under Assumption~\ref{asmp:positivity}, $\frac{\pi_{1-a}(X)}{\pi_{1-a}(X_{-S})} = \frac{\bbP(X \mid A -1, X_{-S})}{\bbP(X \mid X_{-S})}$, the third by canceling, and the final line by definition.

\bigskip

Therefore, we have
$$
\bbE \{ \rho_a(Z; S) \} = \bbE \Big[ f(Z; S)- \overline \varphi_a(X_{-S})^2 + \overline \varphi_a(X_{-S})   \varphi_a(X_{-S}) \Big] 
$$
where $\bbE \{ f(Z; S) \}$ is a second order product of errors.  Plugging this back into \eqref{eq:lambda_remainder}, notice that 
\begin{equation}
	R_2 (\overline{\mathcal{P}}, \mathcal{P}) = \mathbb{E} \left[ - \{ \overline \varphi_a(X_{-S}) - \varphi_a(X_{-S}) \}^2 + 2 \Big\{ f(Z; S) \Big\} \right] \label{eq:second_order1}
\end{equation}
where $\bbE \{ f(Z; S) \}$ is the second order product of errors; specifically, we have
\begin{align}
	\hspace{-0.2in}\bbE \{ f(Z; S) \} = \bbE \Bigg\{ \overline \varphi_a(X_{-S})  \Bigg( &\left\{ \frac{\pi_a(X_{-S}) - \overline \pi_a(X_{-S})}{\overline \pi_a(X_{-S})} \right\} \{ \mu_a(X_{-S}) - \overline \mu_a (X_{-S}) \} \nonumber \\
	&- \left\{ \frac{\pi_a(X) - \overline \pi_a(X)}{\overline \pi_a (X)} \right\} \{ \mu_a(X) - \overline \mu_a (X)  \} \left\{ \frac{\overline \pi_{1-a}(X)}{ \overline \pi_{1-a}(X_{-S})} \right\} \nonumber  \\
	&- \left\{ \frac{\pi_{1-a}(X_{-S}) - \overline \pi_{1-a}(X_{-S})}{\overline \pi_{1-a}(X_{-S})} \right\} \Big[ (\bbE - \overline \bbE) \{ \overline \mu_a(X) \mid A = 1-a, X_{-S} \} \Big] \nonumber \\
	&- \{ \mu_a(X_{[p}) - \overline \mu_a(X) \} \left\{ \frac{\overline \pi_{1-a}(X)}{ \overline \pi_{1-a}(X_{-S})} - \frac{\pi_{1-a}(X)}{\pi_{1-a}(X_{-S})} \right\} \Bigg) \Bigg\}. \label{eq:second_order2}
\end{align}
Therefore, we conclude that the remainder term is second-order.

\subsubsection*{Proof that second order remainder implies efficient influence function}

Now, we relate each of the efficient influence functions back to scores of smooth parametric submodels.  We will make a general argument that applies to both efficient influence functions.  We ignore subscript notation, and let $\psi: \mathcal{P} \to \bbR$ denote a generic functional.  We include this last part of the proof for completeness, but it is an imitation of Lemma 2 in \citet{kennedy2023density}.

\bigskip

Recall from semiparametric efficiency theory that the nonparametric efficiency bound for a functional is given by the supremum of Cramer-Rao lower bounds for that functional across smooth parametric submodels \citep{bickel1993efficient, van2002semiparametric}.  The efficient influence function is the unique mean-zero function that is a valid submodel score satisfying pathwise differentiability; i.e., if $\varphi$ is the efficient influence function of $\psi$, it satisfies
\begin{equation} \label{eq:pathdiff}
	\frac{d}{d \epsilon} \psi (\mathcal{P}_{\epsilon}) \bigg|_{\epsilon = 0} = \int_{\mathcal{Z}} \varphi(\mathcal{P}) \left( \frac{d}{d\epsilon} \log d\mathcal{P}_{\epsilon} \right) \bigg|_{\epsilon = 0} d\mathcal{P}
\end{equation}
for $\mathcal{P}_\epsilon$ any smooth parametric submodel.  

\bigskip

To see that each of the proposed efficient influence functions are, indeed, efficient influence functions, observe that the relevant von Mises expansions imply
\begin{align*}
	\frac{d}{d \epsilon} \psi(\mathcal{P}_{\epsilon}) &= \frac{d}{d\epsilon} \left( \psi(\mathcal{P}) - \int_{\mathcal{Z}} \varphi(\mathcal{P}) d(\mathcal{P} - \mathcal{P}_\epsilon) -  R_2 (\mathcal{P}, \mathcal{P}_\epsilon) \right) \\
	&= \frac{d}{d\epsilon} \int_{\mathcal{Z}}  \varphi(\mathcal{P}) d(\mathcal{P}_\epsilon - \mathcal{P}) -  \frac{d}{d\epsilon}  R_2 (\mathcal{P}, \mathcal{P}_\epsilon)  \\
	&= \int_{\mathcal{Z}}  \varphi(\mathcal{P})  \frac{d}{d\epsilon} d\mathcal{P}_\epsilon -  \frac{d}{d\epsilon}  R_2 (\mathcal{P}, \mathcal{P}_\epsilon) \\
	&= \int_{\mathcal{Z}}  \varphi(\mathcal{P})  \left(\frac{d}{d\epsilon} \log d\mathcal{P}_\epsilon \right) d\mathcal{P}_\epsilon -  \frac{d}{d\epsilon}  R_2 (\mathcal{P}, \mathcal{P}_\epsilon)
\end{align*}
where the second line follows because $\psi(\mathcal{P})$ does not depend on $\epsilon$, the third because $\int \varphi(\mathcal{P}) d\mathcal{P} = 0$, and the fourth and final line since $\frac{d}{d\epsilon} \log d\mathcal{P}_\epsilon  = \frac{1}{d\mathcal{P}_\epsilon} \frac{d}{d\epsilon} d\mathcal{P}_\epsilon$.  Evaluating this expression at $\epsilon = 0$, we have
$$
\int_{\mathcal{Z}}  \varphi(\mathcal{P})  \left(\frac{d}{d\epsilon} \log d\mathcal{P}_\epsilon \right) d\mathcal{P}_\epsilon -  \frac{d}{d\epsilon}  R_2 (\mathcal{P}, \mathcal{P}_\epsilon) \bigg|_{\epsilon = 0} = \int_{\mathcal{Z}}  \varphi(\mathcal{P})  \left(\frac{d}{d\epsilon} \log d\mathcal{P}_\epsilon \right) \bigg|_{\epsilon = 0} d\bbP -  0 
$$
since
$$
\frac{d}{d\epsilon} R_2 (\mathcal{P}, \mathcal{P}_\epsilon) \Big|_{\epsilon = 0} = 0,
$$
which shows that $\varphi$ satisfies the property in eq. \eqref{eq:pathdiff}.  The last equation involving $R_2$ follows because $R_2$ consists of only second-order products of errors between $\mathcal{P}_\epsilon$ and $\mathcal{P}$ by the arguments above.  Therefore, the derivative is composed of a sum of terms, each of which is a product of a derivative term that may not equal zero and an error term involving the differences of components of $\mathcal{P}$ and $\mathcal{P}_\epsilon$, which will be zero when $\epsilon = 0$ since $\mathcal{P} = \mathcal{P}_\epsilon$. 

\bigskip

Since the model is nonparametric, the tangent space is the entire Hilbert space of mean-zero finite-variance functions, and so there is only one influence function satisfying eq. \eqref{eq:pathdiff} and it is the efficient one \citep{tsiatis2006semiparametric}.   Therefore, both efficient influence functions proposed in Lemma~\ref{lem:out_eifs} are indeed, the efficient influence functions.
\end{newproof}

\section{Proofs for Section~\ref{sec:est}} \label{app:est}

Here, we prove Theorems~\ref{thm:fx_conv}-\ref{thm:out_conv}.  Each proof is split into two steps:
\begin{enumerate}
\item Establishing that the estimator for measured confounding is regular and asymptotically linear (possibly under doubly robust conditions), and
\item Establishing that the estimator for the upper bound on the ATE is regular and asymptotically linear.
\end{enumerate}
Let $\bbP$ denote expectation conditional on the training sample throughout what follows.  Moreover, note that the positivity assumption (Assumption~\ref{asmp:positivity}) implies $\mu_a(X_{-S})$ are well-defined for all $S \in 2^{[d]}$.

\subsubsection*{Proof of Theorem~\ref{thm:fx_conv}}
\begin{newproof}
Here, measured confounding is $|\psi - \psi_{-j^\prime}|$.  Therefore, first, we establish convergence guarantees for each adjusted mean difference estimator.  For all $j \in [d] \cup \emptyset$, we have
\begin{equation} \label{eq:psihat_decomp}
	\widehat \psi_{-j} - \psi_{-j} = (\bbP_n - \bbE) \{ \phi(Z_{-j}) \} + (\bbP_n - \bbP) \{ \widehat \phi(Z_{-j}) - \phi(Z_{-j}) \} + \bbP\{ \widehat \phi(Z_{-j}) - \phi(Z_{-j}) \}
\end{equation}
by adding zero, and where $\phi$ is defined in \eqref{eq:mde_if}. Then, by condition 2 of the theorem, sample splitting, and Lemma 2 in \citet{kennedy2020sharp}, the empirical process term satisfies $(\bbP_n - \bbP) \{ \widehat \phi(Z_{-j}) - \phi(Z_{-j}) \} = o_\bbP(n^{-1/2})$. Moreover, by condition 2 and Jensen's inequality, the bias satisfies
$$
\bbP\{ \widehat \phi(Z_{-j}) - \phi(Z_{-j}) \} \leq \sqrt{\bbP\Big[ \{ \widehat \phi(Z_{-j}) - \phi(Z_{-j}) \}^2 \Big]} \equiv \lVert \widehat \phi(Z_{-j}) - \phi(Z_{-j}) \rVert_2 = o_\bbP(1).
$$
Hence, by applying the standard CLT to $(\bbP_n - \bbE) \{ \phi(Z_{-j}) \}$, 
$$
\widehat \psi_{-j}- \psi_{-j} = o_\bbP(1) \ \forall\ j \in [d] \cup \emptyset.
$$
For $j \in \{j^\prime, \emptyset \}$, the convergence guarantee can be strengthened. We have
$$
\widehat \psi_{-j} - \psi_{-j} = (\bbP_n - \bbE) \{ \phi(Z_{-S}) \} + o_\bbP(n^{-1/2})
$$
because the bias term from \eqref{eq:psihat_decomp} satisfies $\bbP\{ \widehat \phi(Z_{-j}) - \phi(Z_{-j}) \} = o_\bbP(n^{-1/2})$ by iterated expectations, Cauchy-Schwarz, and conditions 1 and 3 of the theorem.

\bigskip

Next, we can analyze the convergence of the estimator for measured confounding, $\max_{j \in [d]} | \widehat \psi - \widehat \psi_{-j}|$.  Because $f(x, y) = |x-y|$ is continuous for all $x, y \in \bbR$,
\begin{equation} \label{eq:pf_one}
	| \widehat \psi - \widehat \psi_{-j}| - |\psi - \psi_{-j}| = o_\bbP(1)\ \forall \ j \in [d]
\end{equation}
by the continuous mapping theorem.  Moreover, by Assumption~\ref{asmp:separation_max_fx}, which guarantees $|\psi - \psi_{-j^\prime}|$ is unique, and the delta method,
\begin{equation} \label{eq:pf_two}
	| \widehat \psi - \widehat \psi_{-j^\prime} | - | \psi - \psi_{-j^\prime} | = (\bbP_n - \bbE) \big[ \text{sign}\left(\psi - \psi_{-j^\prime}\right) \left\{ \phi(Z) - \phi(Z_{-j^\prime}) \right\} \big] + o_\bbP(n^{-1/2}).
\end{equation}
Therefore, by Theorem~\ref{thm:technical_max} and \eqref{eq:pf_one} and \eqref{eq:pf_two}, 
$$
\max_{j \in [d]} | \widehat \psi - \widehat \psi_{-j}| - | \psi - \psi_{j^\prime}| = (\bbP_n - \bbE) \big[ \text{sign}\left(\psi - \psi_{-j^\prime}\right) \left\{ \phi(Z) - \phi(Z_{-j^\prime}) \right\} \big] + o_\bbP (n^{-1/2}).
$$
Finally, because $\widehat \psi - \psi = (\bbP_n - \bbE) \{ \phi(Z) \} + o_\bbP(n^{-1/2})$, the result then follows by the delta method and because $0 < \Gamma  | \psi - \psi_{j^\prime}| < \infty$ by Assumption~\ref{asmp:bounded}.
\end{newproof}

\subsubsection*{Proof of Theorem~\ref{thm:odds_conv}}

\begin{newproof}
We begin with the estimator for measured confounding. By conditions 1 and 2 of the result,  $\widehat \pi_1^\perp$ is a regular and asymptotically linear estimator for $\pi_1^\perp$ by Theorem 5.39 in \cite{van2000asymptotic}; moreover, the estimator for the $j^{th}$ coefficient of $\beta$ satisfies 
$$
\widehat \beta_j - \beta_j = e_j^T I_{\beta}^{-1} (\bbP_n - \bbE) \big\{ s(Z; \beta) \big\} + o_\bbP(n^{-1/2})
$$
by the delta method, where $I$ and $s$ are defined in the result.  Then, as with the proof of Theorem~\ref{thm:fx_conv}, by the continuous mapping theorem, delta method, and Theorem~\ref{thm:technical_max},
$$
\widehat M - M = (\bbP_n - \bbP) \{ \phi(Z) \} + o_\bbP(n^{-1/2}),
$$
where $\phi_M(Z)$ is defined in the result, and $\widehat M$ and $M$ are given in Definition~\ref{def:odds}.

\bigskip

Next, we analyze the estimator for the bound.  By iterated expectations in the first line and adding zero in the second line,
\begin{align*}
	\bbP_n \big[ \varphi_{\mathcal{U}} \{ Z; \widehat \eta(\Gamma \widehat M) \} \big] - \mathcal{U}(\Gamma) &= \bbP_n \big[ \varphi_{\mathcal{U}} \{ Z; \widehat \eta(\Gamma \widehat M) \} \big] -  \bbE \big[ \varphi_{\mathcal{U}} \{ Z; \eta(\Gamma M) \} \big] \\
	&= (\bbP_n - \bbE) \big[ \varphi_{\mathcal{U}} \{ Z; \eta(\Gamma M) \} \big] \\
	&+ (\bbP_n - \bbP) \big[ \varphi_{U} \{ Z; \widehat \eta(\Gamma \widehat M) \} -  \varphi_{U} \{ Z; \eta(\Gamma M) \} \big]  \\
	&+ \bbP \big[ \varphi_{U} \{ Z; \widehat \eta(\Gamma \widehat M) \} -  \varphi_{U} \{ Z; \eta(\Gamma \widehat M) \} \big] \\
	&+ \bbP \big[ \varphi_{U} \{ Z; \eta(\Gamma \widehat M) \} -  \varphi_{U} \{ Z; \eta(\Gamma M) \} \big],
\end{align*} 
where $\varphi_{\mathcal{U}}$ is defined in \eqref{eq:yad_eif_upper}. The empirical process term $(\bbP_n - \bbP) \big[ \varphi_{U} \{ Z; \widehat \eta(\Gamma \widehat M) \} -  \varphi_{U} \{ Z; \eta(\Gamma M) \} \big] = o_\bbP(n^{-1/2})$ by sample splitting, Lemma 2 in \citet{kennedy2020sharp}, and condition 3 of the result.  The first bias term satisfies $\bbP \big[ \varphi_{U} \{ Z; \widehat \eta(\Gamma \widehat M) \} -  \varphi_{U} \{ Z; \eta(\Gamma \widehat M) \} \big] = o_\bbP(n^{-1/2})$ by condition 4 of the result. The second bias term can be decomposed using Lemma~\ref{lem:diff} and Taylor's theorem:
$$
\bbP \big[ \varphi_{U} \{ Z; \eta(\Gamma \widehat M) \} -  \varphi_{U} \{ Z; \eta(\Gamma M) \} \big] = \frac{\partial}{\partial t} \bbE \big[ \varphi_{U} \{ Z; \eta(\Gamma t) \} \big] \Big|_{t = M} \left( \Gamma \widehat M - \Gamma M \right) + o \left(\Gamma \left| \widehat M - M \right|\right).
$$
Notice that, by iterated expectations, $\frac{\partial}{\partial t} \bbE \big[ \varphi_{U} \{ Z; \eta(\Gamma t) \} \big] \Big|_{t = M} \equiv \frac{\partial}{\partial M} \mathcal{U}(\Gamma)$, where $\frac{\partial}{\partial M} \mathcal{U}(\Gamma)$ is defined in Lemma~\ref{lem:show_diff_monotone}.  Meanwhile, $o \left(\Gamma \left| \widehat M - M \right|\right) = o_\bbP(n^{-1/2})$ because $\widehat M - M = O_\bbP(n^{-1/2})$ by the argument above for estimating measured confounding. The result then follows by the delta method.
\end{newproof}

\subsubsection*{Proof of Theorem~\ref{thm:out_conv}}

\begin{newproof}
First, we consider estimating the components of measured confounding.  To begin, notice that
$$
\bbP_n \{ \widehat \xi_a(X) \} - \lVert \pi_a(X) \rVert_2^2 = (\bbP_n - \bbE) \{ \xi_a(X) \} + (\bbP_n - \bbP) \{ \widehat \xi_a(X) - \xi_a(X) \} + \bbP\{ \widehat \xi_a(X) - \xi_a(X) \},	
$$
where $\xi$ is defined in \eqref{eq:xi}. 	Then, by condition 3 of the result, sample splitting, and Lemma 2 of \citet{kennedy2020sharp}, the empirical process term satisfies $(\bbP_n - \bbP) \{ \widehat \xi_a(X) - \xi_a(X) \} = o_\bbP(n^{-1/2})$. Next, notice that the bias term is a second-order product of errors by the same proof as for deriving the EIF in Lemma~\ref{lem:out_eifs}.  Indeed, the bias is second-order by \eqref{eq:xi_bias}, but replacing overlines by hats.  Therefore, by condition 6 and Cauchy-Schwarz,  $\bbP\{ \widehat \xi_a(X) - \xi_a(X) \} = o_\bbP(n^{-1/2}).$  Therefore,
$$
\bbP_n \{ \widehat \xi_a(X) \} - \lVert \pi_a(X) \rVert_2^2 = (\bbP_n - \bbE) \{ \xi_a(X) \} + o_\bbP(n^{-1/2}).
$$
So, by Assumption~\ref{asmp:positivity} and the delta method,
$$
\sqrt{\bbP_n \{ \widehat \xi_a(X) \}} - \lVert \pi_a(X) \rVert_2 = (\bbP_n - \bbE) \left\{ \frac{\xi_a(X)}{2 \lVert \pi_a(X) \rVert_2} \right\}  + o_\bbP(n^{-1/2}).
$$

\bigskip

Next,
\begin{align*}
	&\bbP_n \{ \widehat \lambda_a(Z; S) \} - \lVert \mu_a(X_{-S}) - \bbE\{ \mu_a(X) \mid A = 1-a, X_{-S} \} \rVert_2^2 \\
	&\hspace{0.5in}= (\bbP_n - \bbE) \{ \lambda_a(Z; S) \} + (\bbP_n - \bbP) \{  \lambda_a(Z; S) -  \lambda_a(Z; S) \} + \bbP\{ \widehat \lambda_a(Z; S) -  \lambda_a(Z; S) \},
\end{align*}
where $\lambda$ is defined in \eqref{eq:lambda}. The empirical process term is $o_\bbP(n^{-1/2})$ by sample splitting, Lemma 2 of \citet{kennedy2020sharp}, and condition 4.  Then, \eqref{eq:second_order1} and \eqref{eq:second_order2} define the bias of $\bbP\{ \widehat \lambda_a(Z; S) - \lambda_a(Z; S) \}$, replacing overlines by hats. Therefore, the bias term can be controlled by conditions 1, 5, and 7 and Cauchy-Schwarz. In conclusion,
$$
\bbP_n \{ \widehat \lambda_a(Z; S) \} - \lVert \mu_a(X_{-S}) - \bbE\{ \mu_a(X) \mid A = 1-a, X_{-S} \} \rVert_2^2 = (\bbP_n - \bbE) \{ \lambda_a(Z; S) \} + o_\bbP(n^{-1/2})
$$
for $a \in \{ 0,1\}$ and all $S \in \mathcal{S}$.  Then, because measured confounding is non-zero by Assumption~\ref{asmp:bounded}, the delta method yields
$$
\frac{1}{|\mathcal{S}|} \sum_{S \in \mathcal{S}} \bbP_n \{ \widehat \lambda_a(Z; S) \} - \frac{1}{|\mathcal{S}|} \sum_{S\in \mathcal{S}} \lVert \mu_a(X_{-S}) - \bbE \{ \mu_a(X) \mid A =1-a, X_{-S} \} \rVert_2^2  =\frac{1}{|\mathcal{S}|} \sum_{S\in \mathcal{S}} (\bbP_n - \bbP) \{ \lambda_a(Z; S) \} + o_\bbP(n^{-1/2})
$$
and
$$
\hspace{-0.5in}\sqrt{\frac{1}{|\mathcal{S}|} \sum_{S \in \mathcal{S}} \bbP_n \{ \widehat \lambda_a(Z; S) \} } - \sqrt{\frac{1}{|\mathcal{S}|} \sum_{S\in \mathcal{S}} \lVert \mu_a(X_{-S}) - \bbE \{ \mu_a(X) \mid A =1-a, X_{-S} \} \rVert_2^2  } = \frac{1}{|\mathcal{S}|} \sum_{S \in \mathcal{S}} (\bbP_n - \bbP) \left\{ \frac{\lambda_a(Z; S)}{2M_a} \right\} + o_\bbP(n^{-1/2})
$$
where $M_a$ is defined in the result.

\bigskip

Finally, for the other piece of the bound, $\bbP_n \{ \widehat \phi(Z) \} - \psi = (\bbP_n - \bbE) \{ \phi(Z) \} + o_\bbP(n^{-1/2})$, by conditions 1, 2, and 5, and the same analysis as in the proof of Theorem~\ref{thm:fx_conv}. The result then follows by the delta method.
\end{newproof}

\section{Technical results for the estimator of maximum measured confounding} \label{app:technical}

Here, we state and prove a helper result that was used in the proofs of Theorems~\ref{thm:fx_conv} and~\ref{thm:odds_conv}, which, in context of this paper, implies that the estimator of maximum measured confounding is regular and asymptotically linear providing the estimator of the measured confounding corresponding to the maximum is regular and asymptotically linear and the estimators for all the non-maximum sub-quantifications of measured confounding are consistent.

\begin{theorem} \label{thm:technical_max}
Suppose, for $J < \infty$ bounded real-valued functionals $\psi_1, \dots, \psi_J$,
\begin{enumerate}
	\item without loss of generality, $\psi_1 > \psi_j$ for all $j > 1$,
	\item $\widehat \psi_1$ satisfies $\widehat \psi_1 - \psi_1 = (\bbP_n - \bbP) \varphi_1 + o_\bbP(n^{-1/2})$, and
	\item $\widehat \psi_j$ satisfies $\widehat \psi_j - \psi_j = o_\bbP(1)$ for all $j > 1$. 
\end{enumerate}
Let
\begin{equation}
	\widehat j = \argmax_{j \in \{1, \dots, J\}} \widehat \psi_j.
\end{equation} 
Then,
\begin{equation}
	\widehat \psi_{\widehat j} - \psi_1 = (\bbP_n - \bbP) \varphi_1 + o_\bbP(n^{-1/2}) 
\end{equation}
\end{theorem}

\begin{newproof}
We can decompose the error like so:
$$
\widehat \psi_{\widehat j} - \psi_1 = \widehat \psi_{\widehat j} - \widehat \psi_1 + \widehat \psi_1 - \psi_1 = \underbrace{\one(\widehat j \neq 1) \left( \widehat \psi_{\widehat j} - \widehat \psi_1  \right)}_{R_n} +  \underbrace{\widehat \psi_1 - \psi_1}_{Z}.
$$
The second term $Z$ satisfies the result.  Meanwhile, $\one(\widehat j \neq 1) = o_\bbP(n^{-1/2})$ by Lemma~\ref{lem:technical_max}.  Finally, notice that
$$
\left| \widehat \psi_{\widehat j} - \widehat \psi_1  \right| \leq \left| \widehat \psi_{\widehat j} - \psi_{\widehat j} \right| + \left| \psi_{\widehat j} - \psi_1 \right| + \left| \psi_1 - \widehat \psi_1 \right| = O_\bbP(1),
$$
where the last equality follows by conditions 2 and 3 and because $|\psi_{\widehat j} - \psi_1|$ is bounded by assumption.  Therefore, because $|X| = O_\bbP(1) \implies X = O_\bbP(1)$ for any random variable $X$, 
$$
R_n = o_\bbP(n^{-1/2}) O_\bbP(1) = o_\bbP(n^{-1/2}).
$$
The result follows.
\end{newproof}

\begin{lemma} \label{lem:technical_max}
Under the setup of Theorem~\ref{thm:technical_max}, but weakening condition 2 such that $\widehat \psi_1 - \psi_1 = o_\bbP(1)$ (i.e., such that $\widehat \psi_1$ is merely consistent),
$$
\one( \widehat j \neq 1) = o_\bbP(n^{-\alpha}) 
$$
for all $\alpha > 0$.
\end{lemma}

\begin{newproof}
Let $\epsilon > 0$. For all $n > \epsilon^{1/\alpha}$, 
\begin{align*}
	\bbP \left\{ |n^\alpha \one(\widehat j \neq 1)| > \epsilon \right\} &= \bbP (\widehat j \neq 1) \\
	&\leq \sum_{j \in \{ 2, \dots, J\}} \bbP \left( \widehat \psi_{1} < \widehat \psi_{j} \right),
\end{align*}
where the inequality follows by the union bound. Notice that the final term on the right-hand side converges to zero as $n \to \infty$ because, for all $j > 1$,
\begin{align*}
	\bbP( \widehat \psi_1 < \widehat \psi_j ) &= \bbP \left( \widehat \psi_j - \psi_j + \psi_1 - \widehat \psi_1 > \psi_1 - \psi_j \right) \\
	&\leq \bbP \left( |\widehat \psi_j - \psi_j| + |\widehat \psi_1 - \psi_1| > \psi_1 - \psi_j \right) \\
	&\leq \bbP \left( |\widehat \psi_j - \psi_j| > \frac{\psi_1 - \psi_j}{2} \right) + \bbP \left(|\widehat \psi_1 - \psi_1| > \frac{\psi_1 - \psi_j}{2} \right) \\
	&\to 0 \text{ as } n \to \infty,
\end{align*}
where the first line follows by adding zero and rearranging, the third by the union bound and because, for $a, b, c > 0$, $a + b > c \implies a > c/2 \text{ or } b > c/2$, and the fourth because $\widehat \psi_{j} \inprob \psi_j$ for all $j \in [J]$ by conditions 2 and 3 of Theorem~\ref{thm:technical_max} and because $\psi_1 > \psi_j$ for all $j > 1$ by condition 1 of Theorem~\ref{thm:technical_max}. Hence, $\one(\widehat j \neq 1) = o_\bbP(n^{-\alpha})$ by the definition of $o_\bbP$ notation.
\end{newproof}

\section{Proofs of results in Appendix~\ref{app:robustness}} \label{app:proofs-robustness}

\subsection{Proof of Proposition~\ref{prop:gamma_zero_def}}
\begin{proof}
	First, we note that Assumption~\ref{asmp:bounded} ensures Assumption~\ref{asmp:existence} is feasible. Without Assumption~\ref{asmp:bounded}, the bounds in a calibrated sensitivity model would be meaningless --- either equal to the true causal effect when $M = 0$ or infinitely wide. Below is the rest of the proof.

	\smallskip

	Consider first the case where $\mathcal{U}(0) = \mathcal{L}(0) = 0$.  In that case, $\Gamma = 0$ is the unique solution of $\mathcal{U}(\Gamma) \mathcal{L}(\Gamma) = 0$. 

    \smallskip

    Otherwise, $\mathcal{U}(0) = \mathcal{L}(0) \neq 0$. Then, because $\mathcal{U}(\cdot)$ and $\mathcal{L}(\cdot)$ are strictly monotone, $\mathcal{U}(\Gamma) > \mathcal{L}(\Gamma)$ for all $\Gamma > 0$. Therefore, we can consider two cases separately, if $\psi_\ast < 0$ or if $\psi_\ast > 0$.  
	
    \smallskip
	
    If $\psi_\ast < 0$, $\mathcal{L}(\Gamma) < 0$ for all $\Gamma \geq 0$. Moreover, the final part of Assumption~1 applies to $\mathcal{U}(\Gamma)$, and there exists $\Gamma < \infty$ such that $\mathcal{U}(\Gamma) = 0$. Then, by the intermediate value theorem that $\Gamma_0$ can be defined as the unique root of $\mathcal{U}(\Gamma) = 0$.  Therefore, $\Gamma_0$ can be defined as the unique solution of $\mathcal{L}(\Gamma)\mathcal{U}(\Gamma) = 0$.
	
    \medskip
    
    Meanwhile, if $\psi_\ast > 0$, $\Gamma_0$ can be defined as the unique root of $\mathcal{L}(\Gamma) = 0$, and $\mathcal{U}(\Gamma) > 0$ for all $\Gamma \geq 0$. Therefore, $\Gamma_0$ can be defined as the unique solution of $\mathcal{L}(\Gamma)\mathcal{U}(\Gamma) = 0$.
\end{proof}

\subsection{Proof of Theorem~\ref{thm:gamma_zero}}

We establish this result in several steps.  First, we establish a uniform convergence guarantee for $\widehat \Psi$.

\begin{proposition} \label{prop:psi-conv}
    Let $\widehat \Psi(\Gamma) = \widehat{\mathcal{U}}(\Gamma) \widehat{\mathcal{L}}(\Gamma)$.  Suppose the conditions of Theorem~\ref{thm:gamma_zero} hold. Then,
    $$
    \left\| (\widehat \Psi - \Psi) - (\bbP_n - \bbE) \varphi_\Psi \right\|_{\mathcal{G}} = o_{\mathbb{P}}(n^{-1/2}),
    $$
    where $\varphi_\Psi(Z; \Gamma) = \varphi_u(Z; \Gamma) \mathcal{L}(\Gamma) + \varphi_l(Z; \Gamma) \mathcal{U}(\Gamma)$.
\end{proposition}

\begin{proof}
    By adding zero, 
    \begin{align*}
        (\widehat \Psi - \Psi) (\Gamma) &= \widehat{\mathcal{U}}(\Gamma) \widehat{\mathcal{L}}(\Gamma) - \mathcal{U}(\Gamma) \mathcal{L}(\Gamma) \\
        &= \left\{ \widehat{\mathcal{U}}(\Gamma) - \mathcal{U}(\Gamma) \right\} \left\{ \widehat{\mathcal{L}}(\Gamma) - \mathcal{L}(\Gamma) \right\} + \mathcal{U}(\Gamma) \left\{ \widehat{\mathcal{L}}(\Gamma) - \mathcal{L}(\Gamma) \right\} + \mathcal{L}(\Gamma) \left\{ \widehat{\mathcal{U}}(\Gamma) - \mathcal{U}(\Gamma) \right\} 
    \end{align*}
    Therefore, by the definition of $\varphi_\Psi$, 
    \begin{align*}
        (\widehat \Psi - \Psi) (\Gamma) - (\bbP_n - \bbE) \left\{ \varphi_\Psi (Z; \Gamma) \right\} &= \left\{ \widehat{\mathcal{U}}(\Gamma) - \mathcal{U}(\Gamma) \right\} \left\{ \widehat{\mathcal{L}}(\Gamma) - \mathcal{L}(\Gamma) \right\} \\
        &+ \mathcal{U}(\Gamma) \left[ \left\{ \widehat{\mathcal{L}}(\Gamma) - \mathcal{L}(\Gamma) \right\} - (\bbP_n - \bbE) \left\{ \varphi_l(Z; \Gamma) \right\} \right] \\
        &+ \mathcal{L}(\Gamma) \left[ \left\{ \widehat{\mathcal{U}}(\Gamma) - \mathcal{U}(\Gamma) \right\} - (\bbP_n - \bbE) \left\{ \varphi_u(Z; \Gamma) \right\} \right] 
    \end{align*}
    By Assumption~\ref{asmp:est-unif-conv} and the triangle inequality,
    $$
    \left\| \left( \widehat{\mathcal{U}} - \mathcal{U} \right) \left( \widehat{\mathcal{L}} - \mathcal{L} \right) \right\|_{\mathcal{G}} = O_{\mathbb{P}}(n^{-1}) = o_{\mathbb{P}}(n^{-1/2}).
    $$
    Meanwhile, by Assumptions~\ref{asmp:existence} and \ref{asmp:est-unif-conv},
    \begin{align*}
        &\left\| \mathcal{U} \left\{ \left( \widehat{\mathcal{L}} - \mathcal{L} \right) - (\bbP_n - \bbE) \varphi_l \right\} \right\|_{\mathcal{G}} = O(1) o_{\bbP}(n^{-1/2}) = o_{\mathbb{P}}(n^{-1/2}) \text{ and } \\
        &\left\| \mathcal{L} \left\{ \left( \widehat{\mathcal{U}} - \mathcal{U} \right) - (\bbP_n - \bbE) \varphi_u \right\} \right\|_{\mathcal{G}} = O(1) o_{\bbP}(n^{-1/2}) = o_{\mathbb{P}}(n^{-1/2}),
    \end{align*}
    where $O(1)$ denotes boundedness. The result then follows by the triangle inequality.
\end{proof}

\noindent The next result establishes that $\varphi_\Psi(\cdot; \gamma)$ is Donsker in $\Gamma$.

\begin{proposition} \label{prop:varphi-psi-donsker}
    Under the conditions of Theorem~\ref{thm:gamma_zero}, $\varphi_\Psi(\cdot; \Gamma)$ forms a Donsker class. 
\end{proposition}

\begin{proof}
    By definition, $\varphi_\Psi(\cdot; \Gamma) = \varphi_u(\cdot; \Gamma) \mathcal{L}(\Gamma) + \varphi_l(\cdot; \Gamma) \mathcal{U}(\Gamma)$.  By Assumption~\ref{asmp:donsker}, $\varphi_u(\cdot; \Gamma)$ and $\varphi_l(\cdot; \Gamma)$ are Donsker.  Therefore, $\varphi_\Psi(\cdot; \Gamma)$ forms a Donsker class because sums of Donsker classes are Donsker. 
\end{proof}

The next result establishes that $\widehat \Gamma_0$ is a consistent estimator for $\Gamma_0$. 
\begin{proposition} \label{prop:g0-consistency}
    Let $\widehat \Psi(\gamma) = \widehat{\mathcal{U}}(\Gamma) \widehat{\mathcal{L}}(\Gamma)$ and define $\widehat \Gamma_0$ as the solution to $\widehat \Psi(\widehat \Gamma_0) = o_{\mathbb{P}}(n^{-1/2})$. Suppose the conditions of Theorem~\ref{thm:gamma_zero} hold. Then,
    $$
    \widehat \Gamma_0 \inprob \Gamma_0.
    $$
\end{proposition}

\begin{proof}
    By Assumption~\ref{asmp:existence}, $\mathcal{U}(\Gamma)$ and $\mathcal{L}(\Gamma)$ are strictly monotone.  Therefore, the identifiability condition in \citet[Theorem 2.10]{kosorok2008introduction} is satisfied: $\mathcal{L}(\Gamma_n)\mathcal{U}(\Gamma_n) \to 0$ implies $|\Gamma_n - \Gamma_0| \to 0$ for any sequence $\{ \Gamma_n \}\in \mathcal{G}$.  Next, by the triangle inequality and Proposition \ref{prop:psi-conv}, $\lVert \widehat \Psi - \Psi \rVert_{\mathcal{G}} = \left\lVert (\mathbb{P}_n - \bbE) \varphi_\Psi \right\rVert_{\mathcal{G}} + o_{\mathbb{P}}(n^{-1/2}).$  Therefore, $\lVert \widehat \Psi - \Psi \rVert_{\mathcal{G}} = o_{\mathbb{P}}(1)$ because $\varphi_\Psi$ is Donsker (see Proposition~\ref{prop:varphi-psi-donsker}) and therefore Glivenko-Cantelli. Consequently, because $\widehat \Psi(\widehat \Gamma_0) = o_{\bbP}(1)$ by construction and $\lVert \widehat \Psi - \Psi \rVert_{\mathcal{G}} = o_{\bbP}(1)$ by the argument above, $|\widehat \Gamma_0 - \Gamma_0| \inprob 0$ by \citet[Theorem 2.10]{kosorok2008introduction}.
\end{proof}

\subsection*{Proof of Theorem~\ref{thm:gamma_zero}}

\noindent Finally, we conclude with a proof of Theorem~\ref{thm:gamma_zero}.

\begin{proof}
    Theorem~\ref{thm:gamma_zero} follows by \citet[Theorem 3.3.1]{van1996weak}.  Therefore, this proof consists of verifying the following conditions:
    \begin{enumerate}
        \item $\Psi$ is differentiable at $\Gamma_0$ with continuous invertible derivative,
        \item $\Psi(\Gamma_0) = 0$ and $\widehat \Psi (\widehat \Gamma_0) = o_{\mathbb{P}}(n^{-1/2})$,
        \item $\sqrt{n} (\widehat \Psi - \Psi) (\Gamma_0) \indist Z$ for some tight random element $Z$, and
        \item $\sqrt{n} (\widehat \Psi - \Psi) (\widehat \Gamma_0) - \sqrt{n}(\widehat \Psi - \Psi) (\Gamma_0) = o_{\mathbb{P}} \left( 1 + \sqrt{n} |\widehat \Gamma_0 - \Gamma_0 | \right)$.
    \end{enumerate}
    The first condition is assumed to be true, in Assumption~\ref{asmp:g0-diff}.  The second condition follows by construction.  The third condition follows by Proposition~\ref{prop:psi-conv} and the triangle inequality, where
    $$
    Z = N \left( 0, \bbV \{ \varphi_\Psi(Z; \Gamma_0) \} \right).
    $$
    Therefore, the rest of the proof consists of establishing condition 4. 

    \medskip

    First, notice that by the triangle inequality and Proposition~\ref{prop:psi-conv}, 
    $$
    \sqrt{n} (\widehat \Psi - \Psi) (\widehat \Gamma_0) - \sqrt{n}(\widehat \Psi - \Psi) (\Gamma_0) = o_{\mathbb{P}}(1) + \sqrt{n} (\bbP_n - \bbE) \varphi_\Psi(Z; \widehat \Gamma_0) - \sqrt{n} (\bbP_n - \bbE) \varphi_\Psi(Z; \Gamma_0).
    $$
    We can prove the right-hand side satisfies the convergence result using \citet[Lemma 3.3.5]{van1996weak}.  Notice that $\varphi_\Psi(\cdot; \Gamma)$ is Donsker by Proposition~\ref{prop:varphi-psi-donsker}. Meanwhile, $\widehat \Gamma_0 \inprob \Gamma_0$ by Proposition~\ref{prop:g0-consistency}.  Therefore, it remains to show
    $$
    \bbE \{ \varphi_\Psi(Z; \Gamma) - \varphi_\Psi(Z; \Gamma_0) \}^2 \to 0 \text{ as } \Gamma \to \Gamma_0.
    $$
    By the definition of $\varphi_\Psi$, adding zero, and because $(a+b)^2 \leq 2(a^2 + b^2)$,
    \begin{align*}
        \bbE \{ \varphi_\Psi(Z; \Gamma) - \varphi_\Psi(Z; \Gamma_0) \}^2 &\lesssim \left\{ \mathcal{L}(\Gamma) - \mathcal{L}(\Gamma_0) \right\}^2 \bbE \{ \varphi_u(\Gamma) \}^2 + \left\{ \mathcal{U}(\Gamma) - \mathcal{U}(\Gamma_0) \right\}^2 \bbE \{ \varphi_l(\Gamma)\}^2 \\
        &+ \mathcal{L}(\Gamma_0)^2 \bbE\{ \varphi_u(\Gamma) - \varphi_u(\Gamma_0) \}^2 + \mathcal{U}(\Gamma_0)^2 \bbE\{ \varphi_l(\Gamma) - \varphi_l(\Gamma_0)\}^2.
    \end{align*}
    The first two summands converge to zero as $\Gamma \to \Gamma_0$ because $\mathcal{L}(\cdot)$ and $\mathcal{U}(\cdot)$ are continuous and because $\varphi_u$ and $\varphi_l$ have bounded mean. Meanwhile, the second two summands converge to zero as $\Gamma \to \Gamma_0$ because $\mathcal{L}(\Gamma_0)$ and $\mathcal{U}(\Gamma_0)$ are bounded by Assumption~\ref{asmp:smooth-ifs}.

    \medskip

    Because the conditions of \citet[Lemma 3.3.5]{van1996weak} have been satisfied, it follows that 
    $$
    \sqrt{n} (\bbP_n - \bbE) \varphi_\Psi(Z; \widehat \Gamma_0) - \sqrt{n} (\bbP_n - \bbE) \varphi_\Psi(Z; \Gamma_0) = o_{\mathbb{P}} \left( 1 + \sqrt{n} |\widehat \Gamma_0 - \Gamma_0 | \right)
    $$
    and therefore
    $$
    \sqrt{n} (\widehat \Psi - \Psi) (\widehat \Gamma_0) - \sqrt{n}(\widehat \Psi - \Psi) (\Gamma_0) = o_{\mathbb{P}} \left( 1 + \sqrt{n} |\widehat \Gamma_0 - \Gamma_0 | \right).
    $$
    Hence, the conditions of \citet[Theorem 3.3.1]{van1996weak} are satisfied and the result follows.
\end{proof}

\subsection{Proof of Corollary~\ref{cor:gamma_zero}}

\begin{proof}
	Under the assumptions of the theorem, $\Psi^\prime(\Gamma) = \mathcal{U}^\prime(\Gamma) \mathcal{L}(\Gamma) + \mathcal{U}(\Gamma) \mathcal{L}^\prime(\Gamma)$. Moreover, $\varphi_\Psi(Z; \Gamma) = \varphi_u(Z; \Gamma) \mathcal{L}(\Gamma) + \varphi_l(Z; \Gamma) \mathcal{U}(\Gamma)$. If $\mathcal{U}(\Gamma_0) = 0$, then $\Psi^\prime(\Gamma_0) = \mathcal{U}^\prime(\Gamma_0) \mathcal{L}(\Gamma_0)$, and $\varphi_\Psi(Z; \Gamma_0) = \varphi_u(Z; \Gamma_0) \mathcal{L}(\Gamma_0)$. The result follows. If $\mathcal{L}(\Gamma_0) = 0$, then the roles of $\varphi_u$ and $\varphi_l$ are reversed, and the result follows similarly. 
\end{proof}

\end{document}